\documentclass[english]{article}
\usepackage[table]{xcolor}

\usepackage[T1]{fontenc}
\usepackage[latin9]{inputenc}
\usepackage{amsmath,amsthm} 
\usepackage{graphicx}
\usepackage{amssymb}
\usepackage{enumitem}
\usepackage{natbib} 
\usepackage{bussproofs}
\usepackage{fullpage}
\usepackage{thmtools,thm-restate}
\usepackage{tikz}
\usepackage{pgf} 
\usetikzlibrary{patterns,automata,arrows,shapes,snakes,topaths,trees,backgrounds,positioning,through,calc}
\usetikzlibrary{decorations.markings}
\usepackage{setspace}
\usepackage{multicol}
\usepackage{graphicx}
\usepackage{stmaryrd}  
\usepackage{comment}
\usepackage{hyperref}
\usepackage{doi}
\hypersetup{colorlinks=true, citecolor=blue, linkcolor=blue}  
\usepackage{latexsym}

\usepackage{mathrsfs} 

\definecolor{medgreen}{rgb}{0.0, 0.75, 0.0}

\setcitestyle{round}

\usepackage{array, makecell}

\usepackage{natbib}

\theoremstyle{definition}
\newtheorem{theorem}{Theorem}[section]
\newtheorem{fact}[theorem]{Fact}
\newtheorem{proposition}[theorem]{Proposition}
\newtheorem{corollary}[theorem]{Corollary}
\newtheorem{example}[theorem]{Example}
\newtheorem{definition}[theorem]{Definition}
\newtheorem{lemma}[theorem]{Lemma}
  
\newtheorem{remark}[theorem]{Remark} 
 
\newtheorem{conjecture}[theorem]{Conjecture} 

\usepackage[page,toc,titletoc,title]{appendix}


\usepackage{xr}
\onehalfspace
 
\begin{document} 
  
 \title{Split Cycle: A New Condorcet Consistent Voting Method Independent of Clones and Immune to Spoilers}
 \author{Wesley H. Holliday$^\dagger$ and Eric Pacuit$^\ddagger$ \\ \\
 $\dagger$ University of California, Berkeley {\normalsize (\href{mailto:wesholliday@berkeley.edu}{wesholliday@berkeley.edu})} \\
$\ddagger$ University of Maryland {\normalsize (\href{mailto:epacuit@umd.edu}{epacuit@umd.edu})}}
 
 \date{{\normalsize Published in \textit{Public Choice}, Vol.~197, 1-62, 2023.}}
\maketitle

\begin{abstract} We propose a Condorcet consistent voting method that we call Split Cycle. Split Cycle belongs to the small family of known voting methods satisfying the anti-vote-splitting criterion of \textit{independence of clones}. In this family, only Split Cycle satisfies a new criterion we call \textit{immunity to spoilers}, which concerns adding candidates to elections, as well as the known criteria of \textit{positive involvement} and \textit{negative involvement}, which concern adding voters to elections. Thus, in contrast to other clone-independent methods, Split Cycle mitigates both ``spoiler effects'' and ``strong no show paradoxes.''\end{abstract}

\tableofcontents

\section{Introduction}\label{Introduction}

A voting method is Condorcet consistent if in any election in which one candidate is preferred by majorities to each of the other candidates, this candidate---the Condorcet winner---is the unique winner of the election. Condorcet consistent voting methods form an important class of methods in the theory of voting (see, e.g., \citealt{Fishburn1977}; \citealt[\S~8]{Brams2002}; \citealt[\S~2.4]{Zwicker2016}; \citealt[\S~3.1.1]{Pacuit2019}). Although Condorcet methods are not currently used in government elections, they have been used by several private organizations (see \citealt{wikiCondorcet}) and in over 30,000 polls through the Condorcet Internet Voting Service (\href{https://civs.cs.cornell.edu}{https://civs.cs.cornell.edu}). Recent initiatives in the U.S. to make available Instant Runoff Voting (\citealt{Kambhampaty2019}), which uses the same ranked ballots needed for Condorcet methods, bring Condorcet methods closer to political application. Indeed, Eric Maskin and Amartya Sen have recently proposed the use of Condorcet methods in U.S. presidential primaries (\citealt{Maskin2016,Maskin2017a,Maskin2017b}). In the meantime, Condorcet methods continue to be used by committees, clubs, etc.

In this paper, we propose a Condorcet consistent voting method that we call Split Cycle, which has a number of attractive axiomatic properties.\footnote{After submitting this paper, we learned from Jobst Heitzig of his notion of the ``immune set'' discussed in a 2004 post on the Election-Methods mailing list (\citealt{Heitzig2004}), which is equivalent to the set of winners for Split Cycle after replacing `stronger' with `at least as strong' in Heitzig's definition in the post. See Remark \ref{HeitzigRemark} for further connections with \citealt{Heitzig2002}. We subsequently learned from Markus Schulze of Steve Eppley's notion of the ``Beatpath Criterion Method'' in a 2000 post on the Election-Methods mailing list (\citealt{Eppley2000}), which is defined analogously to Split Cycle except that it measures strength of majority preference using winning votes (the number of voters who rank $x$ above $y$) whereas Split Cycle uses the margin of victory (the number of voters who rank $x$ above $y$ minus the number of voters who rank $y$ above $x$). As far as we know, the Split Cycle voting method has not been studied in the research literature. In a companion paper, \citealt{HP2021}, we study Split Cycle as what is known as a \textit{collective choice rule} in the social choice theory literature.} Split Cycle responds to a concern well expressed by a 2004 letter to the Washington Post sent by a local organizer of the Green Party, as quoted by Miller \citeyearpar[p.~119]{Miller2019}:
\begin{quote}
[Electoral engineering] isn't rocket science. Why is it that we can put a man on the moon but can't come up with a way to elect our president that allows voters to vote for their favorite candidate, allows multiple candidates to run and present their issues and\dots [makes] the `spoiler' problem\dots go away?
\end{quote}
Starting with the problem of spoilers, Split Cycle satisfies not only the \textit{independence of clones} criterion proposed by Tideman \citeyearpar{Tideman1987} as an anti-vote-splitting criterion but also a new criterion we call \textit{immunity to spoilers} that rules out spoiler effects not ruled out by independence of clones. What the Green Party organizer meant by a voting method that ``allows voters to vote for their favorite candidate'' is open to multiple interpretations; if it means a reasonable voting method that never provides an incentive for strategic voting, as Miller takes it to mean, then such a method is unavailable by well-known theorems on strategic voting (see \citealt{Gibbard1973}, \citealt{Satterthwaite1973}, \citealt{Taylor2005}). More modestly, one may ask for a voting method such that at the very least, voters will never cause their favorite candidate to be defeated by going to the polls and expressing that their favorite candidate is their favorite. Understood this way, one is asking for a voting method that satisfies the criterion of \textit{positive involvement} (\citealt{Saari1995}). Split Cycle satisfies this criterion, as well as a number of other desirable criteria, including the Condorcet loser criterion, independence of Smith-dominated alternatives, negative involvement, non-negative responsiveness, reversal symmetry, and a criterion concerning the possibility of ties among winners that we call rejectability. In fact, Split Cycle can be distinguished from all voting methods we know of in any of the following three ways:
\begin{itemize}
\item Only Split Cycle satisfies independence of clones, positive involvement, and at least one of  Condorcet consistency, non-negative responsiveness, and immunity to spoilers.\footnote{Among proposed non-Condorcet methods, we believe only Instant Runoff satisfies both independence of clones and positive involvement, but it fails the non-negative responsiveness criterion, which Split Cycle satisfies, as well as immunity to spoilers and negative involvement (see Appendix \ref{RankedChoiceAppendix}).} 
\item Only Split Cycle satisfies independence of clones and negative involvement.
\item Only Split Cycle satisfies independence of clones, immunity to spoilers, and rejectability.
\end{itemize}

Split Cycle is an example of a head-to-head (or pairwise) voting method. We compare each pair of candidates $a$ and $b$ in a \textit{head-to-head match}. If more voters rank $a$ above $b$ than rank $b$ above $a$, then $a$ \textit{wins} the head-to-head match and $b$ \textit{loses} the head-to-head match. If $a$ wins against $b$, then the number of voters who rank $a$ above $b$ minus the number who rank $b$ above $a$ is $a$'s \textit{margin of victory over $b$}. If one candidate wins its matches against all other candidates, that candidate is the winner of the election. But there is a chance that every candidate will lose a match to some other candidate.\footnote{It is also possible that while no one candidate wins its matches against all others, there is at least one candidate who wins \textit{or ties} its matches against all others, where a match between $a$ and $b$ is tied if the same number of voters rank $a$ above $b$ as rank $b$ above $a$. All such candidates will count as \textit{undefeated} according to the definition of Split Cycle below.} When this happens, there is a \textit{majority cycle}: a list of candidates where each candidate wins against the next in the list, and the last candidate wins against the first. For example, candidates $a,b,c$ form a majority cycle if $a$ wins against $b$, $b$ wins against $c$, and $c$ wins against $a$ (for a real election with such a cycle, see the 2021 Minneapolis City Council election at \href{https://github.com/voting-tools/election-analysis}{https://github.com/voting-tools/election-analysis}). There can also be cycles involving more than three candidates.

Split Cycle deals with the problem of majority cycles as follows:\footnote{There is a more computationally efficient way to calculate the Split Cycle winners (see Footnote \ref{FloydWarshall}), but this simple two-step procedure is appropriate for explaining the method to voters.}
\begin{enumerate}
\item In each cycle, identify the head-to-head win(s) with the smallest margin of victory in that cycle.
\item After completing step 1 for all cycles, discard the identified wins. All remaining wins count as \textit{defeats} of the losing candidates.
\end{enumerate}
For example, if $a$ wins against $b$ by 1,000 votes, $b$ wins against $c$ by 2,000 votes, and $c$ wins against $a$ by $3,000$ votes, then $a$'s win against $b$ is discarded. Candidate $b$'s win against $c$ counts as a defeat of $c$ unless it appears in another cycle (involving some other candidates) with the smallest margin of victory in that cycle. The same applies to $c$'s win against $a$.  Crucially, after step 2, there is always an \textit{undefeated} candidate (see Section \ref{Defining}). If there is only one, that candidate wins the election. If there is more than one, then a tiebreaker must be used (see Section \ref{Comparison}).

In the rest of this introduction, we provide additional background to the benefits of Split Cycle: we spell out the problem of ``spoiler effects'' that the independence of clones and immunity to spoiler criteria mitigate (Section \ref{SpoilerSection}), followed by the ``strong no show paradox'' that the positive involvement and negative involvement criteria rule out (Section \ref{NoShowSection}). We then provide a roadmap of the rest of the paper in Section~\ref{Roadmap}. 

\subsection{The Problem of Spoilers}\label{SpoilerSection}

Let us begin with one of the most famous recent examples of a spoiler effect in a U.S. election.

\begin{example}\label{BushNaderGore} In the 2000 U.S. Presidential election in Florida, run using the Plurality voting method, George W. Bush, Al Gore, and Ralph Nader received the following votes:
\begin{center}
\begin{minipage}{2in}\begin{tabular}{ccc}
$2,912,790$ & $2,912,253$    &  $97,488$  \\\hline
Bush & Gore  &  Nader  \\
\end{tabular}\end{minipage}
\end{center}
It is reasonable to assume that if Nader had dropped out before Election Day, then a sufficiently large number of his 97,488 voters would have voted for Gore so that Gore would have won the election (\citealt{Magee2003}, \citealt{Herron2007}). It is also reasonable to assume that while for some Gore voters, Nader may have been their favorite but they strategically voted for Gore, still many more voters preferred Gore to Nader than vice versa. So Nader would pose no direct threat to Gore in a two-person election, but by drawing enough votes away from Gore in the three-person election, he handed the election to Bush. Thus,  Nader ``spoiled'' the election for Gore.\end{example}

In elections where voters submit rankings of the candidates, rather than only indicating their favorite, we can give  precise content to the claim that one candidate spoiled the election for another. One possible formalization uses Tideman's \citeyearpar{Tideman1987} criterion of \textit{independence of clones}. A set $C$ of two or more candidates is a set of clones in an election with ranked ballots if no candidate outside of $C$ appears in between two candidates from $C$ on any voter's ballot. Suppose, for example, that if we had collected ranked ballots in the 2000 Florida election, the results would have been as follows:
\begin{center}
\begin{minipage}{2in}\begin{tabular}{ccc}
$2,912,790$ & $2,912,253$    &  $97,488$  \\\hline
Bush & Gore  &  Nader  \\
Gore & Nader & Gore \\
Nader & Bush & Bush
\end{tabular}\end{minipage}
\end{center}
In this imaginary election with ranked ballots, $\{\mbox{Gore}, \mbox{Nader}\}$ is a set of clones, because Bush never appears between Gore and Nader on any ballot. The independence of clones criterion says (in part) that a non-clone candidate---in this case, Bush---should win in an election if and only if they would win after the removal of a clone from the election. But if we remove Nader, who is a clone of Gore, we obtain the following election:
\begin{center}
\begin{minipage}{2in}\begin{tabular}{ccc}
$2,912,790$ & $2,912,253$    &  $97,488$  \\\hline
Bush & Gore  &  Gore  \\
Gore & Bush  & Bush
\end{tabular}\end{minipage}
\end{center}
In this election, Gore wins. Thus, this imaginary example shows that the Plurality voting method violates independence of clones. Of course, independence of clones would not literally account for the sense in which Nader spoiled the 2000 Florida election for Gore, even if we had in fact collected ranked ballots. For surely some ballots would have had Bush in between Gore and Nader. 

Next we give an example of a spoiler effect that cannot be captured by independence of clones.

\begin{example}\label{IRVExample} This example involving Instant Runoff Voting (IRV), also known as Ranked Choice Voting, the Alternative Vote, or the Hare method, comes from ElectionScience.org \\ (\href{https://www.electionscience.org/library/the-spoiler-effect/}{https://www.electionscience.org/library/the-spoiler-effect/}, accessed 2/19/2020) as ``a simplified approximation of what happened in the 2009 IRV mayoral election in Burlington, Vermont.'' Consider the following election for two candidates, a Democrat $d$ and a Progressive $p$ (where we split the voters who prefer $d$ to $p$ into two columns for comparison with the table to follow):
\begin{center}
\begin{tabular}{ccc}
$37$ & $29$    &  $34$  \\\hline
$d$ & $d$  &  $p$  \\
$p$ & $p$  &  $d$  \\
\end{tabular}
\end{center}
For two candidates, Instant Runoff is simply Majority Voting, so the Instant Runoff winner is $d$. But now suppose an additional Republican candidate $r$ joins the race:
\begin{center}
\begin{tabular}{ccc}
$37$ & $29$    &  $34$  \\\hline
$r$ & $d$  &  $p$  \\
$d$ & $p$  &  $d$  \\
$p$ & $r$  &  $r$  
\end{tabular}
\end{center}
Instant Runoff works by first removing the candidate who received the fewest first place votes---in this case, candidate $d$---from all ballots, resulting in the following:
\begin{center}
\begin{tabular}{ccc}
$37$ & $29$    &  $34$  \\\hline
$r$ & $p$  &  $p$  \\
$p$ & $r$  &  $r$  \\
\end{tabular}
\end{center}
Now $p$ has a majority of first place votes, so $p$ is declared the Instant Runoff winner. Note, however, that in the three-person election, $d$ was the Condorcet winner: a majority of voters (66) prefer $d$ to $p$, and a majority of voters (63) prefer $d$ to $r$. Yet the addition of $r$ kicks $d$ out of the winning spot and results in $p$ being the Instant Runoff winner. Thus, $r$ spoiled the election for $d$.

The independence of clones criterion cannot account for the sense in which $r$ spoiled the election for $d$, because $r$ is not a clone of any candidate. Moreover, Instant Runoff satisfies independence of clones. Thus, independence of clones does not address all spoiler effects. A more recent example occurred in the 2022 Special General Election for U.S. Representative in Alaska on August 16, 2022: had one of the Republicans, Palin, not been on the ballot, then (holding voter rankings fixed) the other Republican, Begich, would have won; moreover,  a majority of voters ranked Begich above Palin; and yet with Palin included, Instant Runoff elected the Democrat in the race, making Palin a spoiler (see \href{https://github.com/voting-tools/election-analysis}{https://github.com/voting-tools/election-analysis}). Below we will propose a criterion of \textit{immunity to spoilers} that accounts for cases like this and that of the Burlington mayoral election.\end{example}

\noindent The cost of using a voting method that allows spoiler effects is not just that elections will actually be spoiled, as in Examples \ref{BushNaderGore} and \ref{IRVExample} (for discussion of the 2016 U.S. Presidential election, see \citealt{Kurrild-Klitgaard2018}, \citealt{Woon2020}, and \citealt{Potthoff2021}, and for examples outside of the U.S., see \citealt[\S~20.3.2]{Kaminski2015} and \citealt{Feizi2020}). Another cost is that potential candidates may be discouraged from entering close races in the first place on the grounds that they might be spoilers. 

What kind of ``spoiler effects'' should we try to prevent? This question mixes the conceptual question of what a ``spoiler'' is and the normative question of what effects we should prevent. Note that here we are dealing only with spoiler effects in single-office elections, as matters are more complicated in multi-office elections (see, e.g., \citealt{Kaminski2018}).

First consider an obviously flawed definition of a spoiler: $b$ is a ``spoiler'' for $a$ just in case $a$ would win without $b$ in the election, but when $b$ joins, then $b$ but not $a$ wins. This is of course not the relevant notion, since spoilers are not winners.

Thus, consider a second definition: $b$ is a ``spoiler'' for $a$ just in case  $a$ would win without $b$ in the election, but when $b$ joins, neither $a$ nor $b$ wins. It is clearly necessary, in order for $b$ to be a spoiler for $a$, that neither $a$ nor $b$ wins after $b$ joins, but is it sufficient? Whether or not it is sufficient according to the ordinary concept of a spoiler, we do not think that we should prevent \textit{all} such effects.\footnote{Thus, we think it is too strong to require that a voting method satisfy the condition known as the \textit{A\"{i}zerman property} (\citealt[p.~41]{Laslier1997}) or \textit{weak superset property} or $\widehat{\alpha}_\subseteq$ (\citealt{Brandt2018}), which is equivalent to the condition that if $a$ would win were no candidate from a set $N$ in the election, then after the candidates in $N$ (the ``newcomers'') join the election, if none of the candidates in $N$ wins, then $a$ still wins. For the same reason, we think it is too strong to require that a voting method satisfy the \textit{strong candidate stability} property (studied for resolute voting methods in \citealt{Dutta2001} and \citealt{Ehlers2003} and generalized to irresolute methods in \citealt{Eraslan2004} and \citealt{Rodriguez-Alvarez2006}), which implies that if $b$ would not win were $b$ to join the election, then $a$ would win with $b$ in the election if and only if $a$ would win without $b$ in the election (cf.~$\widehat{\alpha}$ in \citealt{Brandt2018}). The problem with these conditions is that they ignore the majority preference relations between $a$ and the new candidates, which our condition of immunity to spoilers takes into account.} Consider the following example, where the diagram on the right indicates that, e.g., the number of voters who prefer $a$ to $c$ is one greater than the number who prefer $c$ to $a$:
\begin{center}
\begin{minipage}{2in}\begin{tabular}{ccc}
$2$ & $3$ & $4$   \\\hline
$b$ & $a$ &  $c$ \\
$a$ &  $c$ & $b$ \\
$c$ &  $b$ &  $a$ \\
\end{tabular}\end{minipage}\begin{minipage}{2in}\begin{tikzpicture}

\node[circle,draw, minimum width=0.25in] at (0,0) (a) {$a$}; 
\node[circle,draw,minimum width=0.25in] at (3,0) (b) {$b$}; 
\node[circle,draw,minimum width=0.25in] at (1.5,1.5) (c) {$c$}; 

\path[->,draw,thick] (c) to node[fill=white] {$5$} (b);
\path[->,draw,thick] (b) to node[fill=white] {$3$} (a);
\path[->,draw,thick] (a) to node[fill=white] {$1$} (c);

\end{tikzpicture}
\end{minipage}\end{center}
For this election, we agree with proponents of voting methods such as Minimax, Ranked Pairs, and Beat Path (all defined in Appendix \ref{OtherMethodsAppendix}) that $c$ should be the winner. Everyone suffers a majority loss to someone, but while $c$ suffers a slight majority loss to $a$, $a$ suffers a larger majority loss to $b$, who suffers an even larger majority loss to $c$. The electorate is in a sense incoherent, and the fairest way to respond in this case is to elect $c$.\footnote{At least for a deterministic voting method. A voting method that outputs a probability distribution on the set of candidates (see \citealt{Brandt2017}) could assign nonzero probabilities to each candidate in this example. But in this paper we do not consider probabilistic voting methods.} But if $b$ had not been in the election, so we would not have had to account for the majority preferences for $b$ over $a$ and for $c$ over $b$, then $a$ would have been the appropriate winner in the two-person election. Since we agree with all of these verdicts, we do not think a voting method should prevent all effects of the kind described in the second definition.

Similar remarks apply to a third definition (from \citealt{WikiSpoiler}): $b$ is a ``spoiler'' for $a$ just in case $a$ would win without $b$ in the election, \textit{and \textnormal{(}most of\textnormal{)} the voters who prefer $b$ over $c$ also prefer $a$ over $c$}, but when $b$ joins, neither $a$ nor $b$ wins but rather $c$ wins. Based on the example above, in which \textit{all} voters who prefer $b$ over $c$ also prefer $a$ over $c$, we do not think a voting method should prevent all such~effects.

The problem with the definitions of spoiler effects above is that they ignore the voters' preferences for $a$ vs. $b$. If a majority of voters prefer $b$ to $a$, then $b$ may legitimately make $a$ a loser, even if $b$ does not replace $a$ as a winner. Thus, the only spoiler effects that we ought to rule out are those in which $a$ is majority preferred to $b$. This leads to the idea that one ought to use a voting method with the following property:

\begin{itemize}
\item Immunity to spoilers: if $a$ would win without $b$ in the election, \textit{and more voters prefer $a$ to $b$ than prefer $b$ to $a$}, then it is not the case that when $b$ joins the election, both $a$ and $b$ lose.
\end{itemize}
This captures Example \ref{BushNaderGore} (in the imaginary version with ranked ballots), as a majority of voters prefer Gore to Nader. Unlike independence of clones, it also captures Example \ref{IRVExample}, as a majority of voters prefer the Democrat to the Republican in the Burlington election. 

One way to avoid a spoiler effect of the kind identified in immunity to spoilers is that when $a$ would win without $b$ in the election, and a majority of voters prefer $a$ to $b$, then when $b$ joins the election, $a$ loses but $b$ \textit{wins}. In this case, we say that $b$ \textit{steals the election from $a$}. It is hardly more desirable for $b$ to steal the election from $a$ than to spoil the election for $a$, so we propose the property~of
\begin{itemize}
\item Immunity to stealers: if $a$ would win without $b$ in the election, and more voters prefer $a$ to $b$ than prefer $b$ to $a$, then it is not the case that when $b$ joins the election, $a$ loses but $b$ wins.
\end{itemize}
The combination of immunity to spoilers and immunity to stealers is equivalent to a criterion we call 
\begin{itemize}
\item Stability for winners: if $a$ would win without $b$ in the election, and more voters prefer $a$ to $b$ than prefer $b$ to $a$, then when $b$ joins, $a$ still wins.
\end{itemize}
This criterion can be seen as extending the idea of Condorcet consistency\footnote{More accurately, the idea that any Condorcet winner ought to be at least tied for winning the election.} to the variable-candidate setting: a candidate who would be a winner without the newcomers and is majority preferred to all the newcomers remains a winner after the addition of the newcomers.  We will show that Split Cycle satisfies stability for winners and hence immunity to spoilers and stealers.

Note that we are not claiming that in a particular election using a particular voting method, if there is a  candidate $a$ such that for some $b$, $a$ would win in the election without $b$, and more voters prefer $a$ to $b$ than prefer $b$ to $a$, then $a$ \textit{ought to win}. For example, suppose that in an election using the method of \textit{dictatorship}, $a$ would have won without $b$ in the election, i.e., $a$ is the favorite candidate of the dictator besides $b$, and a majority of voters prefer $a$ to $b$. It does not follow that $a$ ought to win the election, as, e.g., there may be a Condorcet winner $c$ who ought to win. Our claim is rather that one ought to use a voting method satisfying stability for winners, which rules out dictatorship as a candidate voting method in the first place.

Another qualification is that in the statement of the axioms above, by `win' we really mean \textit{at least tie for the win}. It is too much to require that these axioms hold even when we incorporate tiebreaking to select a single winner. For a simple example, suppose that in a two-candidate election, the same number of voters prefer $a$ to $c$ as prefer $c$ to $a$, so there is a perfect tie; and suppose that with $b$ included in the election, both $a$ and $c$ beat $b$ head-to-head, but $a$'s margin of victory over $b$ is much larger than $c$'s margin of victory over $b$.  In this case, \textit{both} $a$ and $c$ can argue that without $b$ in the election, they would have won, and each of them beats $b$ head-to-head, so $a$ and $c$ should both still win, thereby avoiding a spoiler effect with respect to $b$. Indeed, we will count both $a$ and $c$ as \textit{undefeated}, according to Split Cycle. But if we need to select a single winner, we think it is reasonable to break the tie in $a$'s favor when $b$ is included in the election, as $a$ has a larger majority victory over $b$ than $c$ does, thereby breaking the symmetry between $a$ and $c$. Thus, while stability for winners should hold for the method of selecting the pre-tiebreaking winners, we should not require it of tiebreaking procedures. We propose weaker axioms on tiebreaking procedures in Section \ref{NewCriteria}.

\subsection{The Strong No Show Paradox}\label{NoShowSection}

The term ``no show paradox'' was coined by Fishburn and Brams \citeyearpar{Fishburn1983} for violations of what is now called the \textit{negative involvement} criterion (see \citealt{Perez2001}). This criterion states that if a candidate $x$ is not among the winners in an initial election scenario, then if we add to that scenario some new voters who rank $x$ as the (unique) last place candidate on their ballots, then the addition of those voters should not make $x$ a winner. Perhaps surprisingly, well-known voting methods such as Instant Runoff, Ranked Pairs, and Beat Path fail to satisfy the negative involvement criterion. In an example of Fishburn and Brams, two voters are unable to make it to an Instant Runoff election, due to their car breaking down. They later realize that had they voted in the election, their \textit{least favorite} candidate would have won. For a simplified version of the Fishburn and Brams example, consider the following example for Instant Runoff (\citealt[\S~3.3]{Pacuit2019}):
\begin{center}
\begin{tabular}{cccc}
$2$ & $3$ & $1$ & $3$   \\\hline
$a$ & $b$ &  $c$ & $c$ \\
$b$ &  $c$ & $a$ & $b$ \\
$c$ &  $a$ &  $b$ & $a$ \\
\end{tabular}
\end{center}
Candidate $a$ receives the fewest first place votes, so $a$ is eliminated in the first round. With $a$ eliminated from the ballots, $b$ receives a majority of first place votes and hence wins according to Instant Runoff. But now suppose that two additional voters with the ranking $abc$ (so $a$ is preferred to $b$ and $c$, and $b$ is preferred to $c$) make it to the election---their car does not break down---resulting in the following:
\begin{center}
\begin{tabular}{cccc}
$4$ & $3$ & $1$ & $3$   \\\hline
$a$ & $b$ &  $c$ & $c$ \\
$b$ &  $c$ & $a$ & $b$ \\
$c$ &  $a$ &  $b$ & $a$ \\
\end{tabular}
\end{center}
Now candidate $b$ receives the fewest first place votes, so $b$ is eliminated in the first round. With $b$ eliminated from the ballots, $c$ receives a majority of first place votes and hence wins according to Instant Runoff. Thus,  the addition of two voters who rank $c$ last makes $c$ the winner. This is a failure of negative involvement. Remarkably, Instant Runoff has
violated negative involvement in real elections, such as the Alaska election mentioned in
Example \ref{IRVExample} (again see \href{https://github.com/voting-tools/election-analysis}{https://github.com/voting-tools/election-analysis}).

The dual of the negative involvement criterion is the \textit{positive involvement} criterion (again see \citealt{Saari1995}, \citealt{Perez2001}).\footnote{Also see \citealt{Kasper2019}, where positive and negative involvement are called the ``Top Property'' and ``Bottom Property'', respectively. A closely related criterion for unique winners is given by Richelson \citeyearpar{Richelson1978} under the name `voter adaptability'.} This criterion states that if a candidate $x$ is among the winners in an initial election scenario, then if we add to that scenario some new voters who rank $x$ as the (unique) first place candidate on their ballots, then the addition of these new voters should not make $x$ a loser. Moulin \citeyearpar{Moulin1988} gives the following example of a failure of positive involvement for the Sequential Elimination voting method in which $a$ faces $b$ in the first round, and then the winner of the first round faces $c$.\footnote{Unlike Instant Runoff, this is a Condorcet consistent voting method. We have slightly modified Moulin's example to avoid the use of a tiebreaking rule, at the expense of adding two new voters rather than one in the second election scenario.} In the initial election scenario, we have the following ballots:
\begin{center}
\begin{tabular}{ccc}
$2$ & $2$ & $1$     \\\hline
$a$ & $b$ &  $c$   \\
$b$ &  $c$ & $a$   \\
$c$ &  $a$ &  $b$   \\
\end{tabular}
\end{center}
In the first round, $a$ beats $b$ (3 voters prefer $a$ to $b$, and only $2$ prefer $b$ to $a$), and then in the second, $c$ beats $a$ (3 voters prefer $c$ to $a$, and only $2$ prefer $a$ to $c$). But now suppose two additional voters make it to the election with the ballot $cba$, so we have:
\begin{center}
\begin{tabular}{cccc}
$2$ & $2$ & $1$ & $2$   \\\hline
$a$ & $b$ &  $c$ & $c$ \\
$b$ &  $c$ & $a$ & $b$ \\
$c$ &  $a$ &  $b$ & $a$ \\
\end{tabular}
\end{center}
Now in the first round, $b$ beats $a$ (4 voters prefer $b$ to $a$, and only 3 prefer $a$ to $b$), and then in the second, $b$ beats $c$ (4 voters prefer $b$ to $c$, and only $3$ prefer $c$ to $b$). Thus, adding voters whose favorite candidate is $c$ turns $c$ from being a winner to a loser. This is a failure of positive involvement.

What is wrong with voting methods that fail negative or positive involvement? Our objection to them is \textit{not} that they incentivize a certain kind of strategic (non-)voting. All reasonable voting methods incentivize some kind or other of strategic voting (again see \citealt{Taylor2005}).  Suppose we have a group of voters who will definitely cast their ballots and vote sincerely, regardless of the electoral consequences. Thus, the voters in the previous example who rank $c$ first will come to the polls and cast their ballots, resulting in $c$ losing the election that otherwise $c$ would have won. Since the voters do not stay home strategically, is the fact that the voting method fails positive involvement unproblematic? Not at all. The problem is that the voting method is responding in the wrong way to additional unequivocal support from a voter for a candidate ($c$ is the voter's unique favorite). As an analogy, a voting method failing the non-negative responsiveness criterion (see Section \ref{MonotonicitySection}) also means that it can incentivize a certain kind of strategic voting; but even for a group of always sincere voters, failing non-negative responsiveness is a flaw of a voting method because it means that the voting method is responding in the wrong way to voters purely improving a candidate's position relative to other candidates.

The failure of positive or negative involvement is now sometimes called the ``strong no show paradox.''\footnote{Perez \citeyearpar{Perez2001} calls violations of positive involvement the ``positive strong no show paradox'' and violations of negative involvement the ``negative strong no show paradox.'' Felsenthal and Tideman \citeyearpar{Felsenthal2013} and Felsenthal and Nurmi \citeyearpar{Felsenthal2016} call them the ``P-TOP'' and ``P-BOT'' paradoxes, respectively.} The reason seems to be that Moulin \citeyearpar{Moulin1988} changed the meaning of ``no show paradox'' to stand not for a violation of negative (or positive) involvement but rather for a violation of the \textit{participation} criterion: if a candidate $x$ is the winner in an initial election, then if we add to that scenario some new voters who rank $x$ above $y$, then the addition of these new voters should not make $y$ the winner. Crucially, it is not required here that $x$ is at the top of the new voters' ballots or that $y$ is at the bottom. In our view, this participation criterion is problematic. To see why (as made precise in Appendix \ref{ParticipationAppendix}), note that if the new voters do not rank $x$ at the top of their ballots and do not rank $y$ at the bottom, then in the presence of majority cycles, new voters having ballots with the ranking $x' x\, y\, y'$ increase the number of people who prefer $x'$ to $x$, which may result in $x'$ knocking $x$ out of contention, and increase the number of people who prefer $y$ to $y'$, which may result in $y'$ no longer knocking $y$ out of contention. No wonder, then, that the winner may change from $x$ to $y$. In fact, remarkably, adding new voters who rank $x$ above $y$ may make $y$'s new position vis-\`{a}-vis other candidates perfectly symmetrical to $x$'s old position vis-\`{a}-vis other candidates, up to a renaming of the candidates (again see Appendix \ref{ParticipationAppendix}). A certain kind of neutrality then requires that the winner changes from $x$ to $y$. Of course, a method not satisfying participation will incentivize some strategic non-voting, as the voters in question will have an incentive not to vote (sincerely). But again, all voting methods incentivize strategic behavior. Thus, we are not so troubled by results showing that all Condorcet consistent voting methods fail versions of participation\footnote{\label{IrresNote}Note that Moulin \citeyearpar{Moulin1988} only proves participation failure for Condorcet consistent voting methods that are \textit{resolute}, i.e., always pick a unique winner, which requires imposing an arbitrary tiebreaking rule that violates anonymity or neutrality (see Section \ref{ANSection}). Since none of the standard Condorcet consistent voting methods are resolute, one may wonder about the significance of the fact that resolute Condorcet methods all fail participation. For discussion of the irresolute case, see \citealt{Perez2001}, \citealt{Jimeno2009}, and \citealt{Sanver2012}.} and therefore incentivize some strategic behavior. By contrast, we are troubled by failures of positive or negative involvement, as this shows that the method responds in the wrong way to unequivocal support for (resp.~rejection of) a candidate.

Unlike well-known voting methods such as Instant Runoff, Ranked Pairs, and Beat Path, the method we propose in this paper, Split Cycle, satisfies positive and negative involvement. Hence it is not only immune to spoilers but also immune to the strong no show paradox. Figure \ref{RoadsFig} illustrates how solving the spoiler problem and the strong no show paradox leads uniquely to Split Cycle as opposed to standard voting methods.

\begin{figure}[h]
\begin{center}
\begin{tikzpicture}

    \path[-,rounded corners,fill=medgreen!30,opacity=0.35] (1.95,2.5) to  (1.65, 0) to (-1.5, -0.5) to (-1.5, -5.35) to (-6, -5.35) to (-6, 0) to (-2.75, 2.5);
    
    \path[-,rounded corners,fill=red!30,opacity=0.35](-1,2.5) to (-0.85, 0) to  (-0.85, -5.35) to (2, -5.35) to (2.15, 0) to (2.35,2.5);

\path[-,rounded corners,fill=blue!30,opacity=0.35] (-0.75, 2.5) to (-0.5,0) to (3, -0.5) to (3, -5.35) to  (5.85, -5.35) to (5.85, 0)  to (3.35,2.5) ;

    \path[-,draw,rounded corners,]  (1.95,2.5) to  (1.65, 0) to (-1.5, -0.5) to (-1.5, -5.35) to (-6, -5.35) to (-6, 0) to (-2.75, 2.5);
\path[-,draw,rounded corners,](-1,2.5) to (-0.85, 0) to  (-0.85, -5.35) to (2, -5.35) to (2.15, 0) to (2.35,2.5);
\path[-,draw,rounded corners,] (-0.75, 2.5) to (-0.5,0) to (3, -0.5) to (3, -5.35) to  (5.85, -5.35) to (5.85, 0)  to (3.35,2.5) ;
 \node at (0.625,1.3) {\bf \begin{minipage}{1in}\begin{center}Split Cycle\end{center}\end{minipage}};

 \node[gray] at (-3.9, -0.25) {\small\begin{minipage}{1in}\begin{center}GETCHA/GOCHA\end{center}\end{minipage}};

\node[gray] at (-3.75, -0.7) {\small\begin{minipage}{1in}\begin{center}Uncovered\ Set\end{center}\end{minipage}};
\node at (-3.75,-1.65) {\begin{minipage}{1in}\begin{center}Immunity to \\ Spoilers\end{center}\end{minipage}};
\path[-,draw,thick] (-5.5,-2.25) to (-2,-2.25);
\node[gray] at (-3.75,-3.05) {\small \begin{minipage}{1in}\begin{center}Beat Path\\ Ranked Pairs Instant Runoff\end{center}\end{minipage}};
\node at (-3.75,-4.5) {\begin{minipage}{1.5in}\begin{center}Independence of \\ Clones\end{center}\end{minipage}};
 \node at (-3.75,-5.75) {Spoiler Problem};
 \node at (2.5,-5.75) {Strong No Show Paradox};
\node[gray] at (0.575, -0.7) {\small \begin{minipage}{1.25in}\begin{center}Minimax\end{center}\end{minipage}};
\node at (0.575, -1.65) { \begin{minipage}{1.25in}\begin{center}Condorcet\\ Winner\end{center}\end{minipage}};
\path[-,draw,thick] (-0.65,-2.25) to (1.85,-2.25);
\node[gray] at (0.615, -2.85) {\small \begin{minipage}{1.25in}\begin{center}Instant Runoff\ Plurality\end{center}\end{minipage}};
\node at (0.5,-4.5) { \begin{minipage}{1.25in}\begin{center}Positive \\ Involvement\end{center}\end{minipage}};
\node at (4.45, -1.65) { \begin{minipage}{1.25in}\begin{center}Condorcet\\ Loser\end{center}\end{minipage}};
\path[-,draw,thick] (3.25,-2.25) to (5.55,-2.25);
\node[gray] at (4.45, -2.85) {\small \begin{minipage}{1.25in}\begin{center}Minimax\\ Plurality\end{center}\end{minipage}};
\node at (4.45,-4.5) {\begin{minipage}{1.25in}\begin{center}Negative \\ Involvement\end{center}\end{minipage}};
 \path[-,draw,rounded corners,  thick] (1.95,2.5) to  (1.65, 0) to[sharp corners] (0.65,-0.16) to[rounded corners]  (-0.5,0) to  (-0.75, 2.5);
\end{tikzpicture}
\end{center}
\caption{An illustration of how solving the Spoiler Problem and the Strong No Show Paradox leads uniquely to Split Cycle as opposed to standard voting methods, some of which are displayed in gray (and defined in Appendix~\ref{OtherMethodsAppendix}). For each of the three ``roads'' in the diagram and each of the displayed voting methods, a voting method is shown on a road if and only if it satisfies all of the criteria that appear in the road lower in the diagram than the voting method. For example, Minimax satisfies Condorcet Winner and Positive Involvement, explaining its location on the middle road; it satisfies Negative Involvement but not Condorcet Loser, explaining its location on the right road; and it satisfies Immunity to Spoilers but not Independence of Clones, explaining its absence from the left road.  All voting methods except Split Cycle are blocked from entering the overlap of the three roads, but some are blocked even earlier by other criteria, as indicated by the horizontal lines. For example, Instant Runoff is  blocked by the Condorcet Winner criterion.}\label{RoadsFig}
\end{figure}

\subsection{Organization}\label{Roadmap}

The rest of the paper is organized as follows. In Section \ref{PrelimSection}, we review some preliminary notions: profiles, margin graphs, and voting methods, as well as operations on profiles. In Section \ref{SplitCycleSection}, we motivate and define our proposed voting method, Split Cycle. In Section \ref{NewCriteria}, we discuss our new voting criteria concerning spoilers, stealers, and stability, as well as a stronger criterion based on Sen's \citeyearpar{Sen1971,Sen1993} choice-functional condition of expansion consistency. In Section \ref{OtherCriteria}, we test Split Cycle against a number of other criteria from the literature.  Our axiomatic analysis of Split Cycle and other methods is summarized in Figure \ref{AxiomTable}. We conclude in Section~\ref{Conclusion} with a brief summary and directions for further research on Split Cycle. Appendices \ref{ClonesAppendix} and \ref{ParticipationAppendix} contain proofs deferred from the main text. Appendices \ref{OtherMethodsAppendix} and \ref{Quant} contain definitions of other voting methods in Figure \ref{AxiomTable} and data from simulations of Split Cycle and other voting methods, respectively.

\begin{remark}\label{CodeRemark} The Split Cycle voting method is currently in use at \href{https://stablevoting.org}{stablevoting.org}, as described in Section~\ref{Comparison}. An implementation in Python of Split Cycle and other methods referenced in this paper is available at \href{https://github.com/epacuit/splitcycle}{https://github.com/epacuit/splitcycle}. All of the examples in the paper have been verified in a Jupyter notebook available in the linked repository. Most of the proofs of properties of Split Cycle have been formalized in the Lean Theorem Prover at \href{https://github.com/chasenorman/Formalized-Voting}{https://github.com/chasenorman/Formalized-Voting}, as described in \citealt{HNP2021}.\end{remark}

\renewcommand{\arraystretch}{1.3}

\begin{figure}
\begin{center}\setlength{\tabcolsep}{3.25pt}\begin{tabular}{l|c|c|c|c|c|c|c|c|c|}
 &  \makecell{  Split\\[-2pt] Cycle}  &  \makecell{  Ranked\\[-2pt] Pairs} & \makecell{  Beat\\[-2pt] Path} & \makecell{Mini-\\[-2pt] max} & Copeland   &   \makecell{  GETCHA\\[-2pt] /GOCHA}    &  \makecell{  Uncovered\\[-2pt] Set} &   \makecell{  Instant\\[-2pt] Runoff} &  Plurality\\\Xhline{3\arrayrulewidth}

 \makecell[l]{ {Immunity to}\\[-2pt] {Spoilers} (\ref{SpoilersSection})} & \cellcolor{gray!50}$\checkmark$&\cellcolor{gray!50}$-$ & \cellcolor{gray!50}$-$& \cellcolor{gray!50}$\checkmark$&  \cellcolor{gray!50}$\checkmark$ & \cellcolor{gray!50}{$\checkmark$} & \cellcolor{gray!50}{$\checkmark$} & \cellcolor{gray!50}$-$ &\cellcolor{gray!50}$-$ \\\hline
 
 \makecell[l]{ {Immunity to}\\[-2pt] {Stealers} (\ref{StealersSection})}  & \cellcolor{white}{$\checkmark$} & \cellcolor{white}{$\checkmark^\star$}&\cellcolor{white}{$-$} &\cellcolor{white}{$-$} & \cellcolor{white}{$-$} & \cellcolor{white}{$\checkmark$} &\cellcolor{white}{$\checkmark$}& \cellcolor{white}{$-$} & \cellcolor{white}$-$\\\hline
 
  \makecell[l]{ {Stability for}\\[-2pt] {Winners} (\ref{StabilitySection})} & \cellcolor{gray!50}$\checkmark$&\cellcolor{gray!50}$-$ & \cellcolor{gray!50}$-$& \cellcolor{gray!50}$-$& \cellcolor{gray!50}$-$ & \cellcolor{gray!50}{$\checkmark$} & \cellcolor{gray!50}{$\checkmark$} & \cellcolor{gray!50}$-$ &\cellcolor{gray!50}$-$ \\\hline
  
\makecell[l]{{Expansion}\\[-2pt] {Consistency, $\gamma$ (\ref{ExpansionSection})}} & \cellcolor{white}$\checkmark$&\cellcolor{white}$-$ & \cellcolor{white}$-$& \cellcolor{white}$-$& \cellcolor{white}$-$ & \cellcolor{white}{$\checkmark$/$-$} & \cellcolor{white}{$\checkmark^\dagger$}& \cellcolor{white}$-$ &\cellcolor{white}$-$ \\\Xhline{3\arrayrulewidth}

 \makecell[l]{ {Anonymity and}\\[-2pt] {Neutrality} (\ref{ANSection})}  & \cellcolor{gray!50}$\checkmark$&\cellcolor{gray!50}$\checkmark$& \cellcolor{gray!50}$\checkmark$&\cellcolor{gray!50}$\checkmark$ & \cellcolor{gray!50}$\checkmark$ & \cellcolor{gray!50}$\checkmark$ &\cellcolor{gray!50}$\checkmark$& \cellcolor{gray!50}$\checkmark$ &\cellcolor{gray!50}$\checkmark$ \\\hline

 \makecell[l]{ {Reversal}\\[-2pt] {Symmetry} (\ref{RSSection})}  & \cellcolor{white}$\checkmark$&\cellcolor{white}$\checkmark$& \cellcolor{white}$\checkmark$&\cellcolor{white}$-$ & \cellcolor{white}$\checkmark$ & \cellcolor{white}$\checkmark$ &\cellcolor{white}$\checkmark$& \cellcolor{white}$-$ &\cellcolor{white}$-$ \\\Xhline{3\arrayrulewidth}

{Pareto (\ref{ParetoSection})} & \cellcolor{gray!50}{$\checkmark$} & \cellcolor{gray!50}{$\checkmark$} & \cellcolor{gray!50}{$\checkmark$} & \cellcolor{gray!50}{$\checkmark$} & \cellcolor{gray!50}{$\checkmark$} & \cellcolor{gray!50}{$-$} &\cellcolor{gray!50}{$\checkmark$}& \cellcolor{gray!50}{$\checkmark$} &  \cellcolor{gray!50}{$\checkmark$} \\\hline

 \makecell[l]{ {Condorcet}\\[-2pt] {Winner} (\ref{CondorcetSection})} & \cellcolor{white}{$\checkmark$} & \cellcolor{white}{$\checkmark$} & \cellcolor{white}{$\checkmark$} & \cellcolor{white}{$\checkmark$} & \cellcolor{white}{$\checkmark$} & \cellcolor{white}{$\checkmark$} &\cellcolor{white}{$\checkmark$}& \cellcolor{white}{$-$} &  \cellcolor{white}{$-$} \\\hline

 \makecell[l]{ {Condorcet}\\[-2pt] {Loser} (\ref{CondorcetSection})}& \cellcolor{gray!50}{$\checkmark$} & \cellcolor{gray!50}{$\checkmark$}& \cellcolor{gray!50}{$\checkmark$}& \cellcolor{gray!50}{$-$} & \cellcolor{gray!50}{$\checkmark$} & \cellcolor{gray!50}{$\checkmark$}&\cellcolor{gray!50}{$\checkmark$} & \cellcolor{gray!50}{$\checkmark$} &  \cellcolor{gray!50}{$-$} \\\hline

\makecell[l]{{Smith (\ref{SmithSchwartz})}} & \cellcolor{white}$\checkmark$&\cellcolor{white}$\checkmark$& \cellcolor{white}$\checkmark$&\cellcolor{white}$-$ & \cellcolor{white}$\checkmark$ & \cellcolor{white}$\checkmark$ &\cellcolor{white}$\checkmark$ &\cellcolor{white}$-$ &\cellcolor{white}$-$ \\\Xhline{3\arrayrulewidth}

\makecell[l]{{ISDA (\ref{ISDASection})}} & \cellcolor{gray!50}$\checkmark$&\cellcolor{gray!50}$\checkmark$& \cellcolor{gray!50}$\checkmark$&\cellcolor{gray!50}$-$ &  \cellcolor{gray!50}$\checkmark$ & \cellcolor{gray!50}$\checkmark$ & \cellcolor{gray!50}$\checkmark$ & \cellcolor{gray!50}$-$ &\cellcolor{gray!50}$-$ \\\hline

\makecell[l]{{Independence of}\\[-2pt] {Clones (\ref{ClonesSection})}} & \cellcolor{white}$\checkmark$&\cellcolor{white}$\checkmark^*$ &\cellcolor{white}$\checkmark$ &\cellcolor{white}$-$ &\cellcolor{white} $-$  &\cellcolor{white}$\checkmark$ &\cellcolor{white}$\checkmark^\dagger$&\cellcolor{white}$\checkmark^\ddagger$&\cellcolor{white}$-$ \\\Xhline{3\arrayrulewidth}

{Rejectability (\ref{Rejectability})} & \cellcolor{gray!50}$\checkmark$& \cellcolor{gray!50}$\checkmark$& \cellcolor{gray!50}$\checkmark$ & \cellcolor{gray!50}$\checkmark$& \cellcolor{gray!50}$-$ & \cellcolor{gray!50}$-$ & \cellcolor{gray!50}$-$ & \cellcolor{gray!50}$\checkmark$ & \cellcolor{gray!50}$\checkmark$\\\hline

\makecell[l]{Resolvability (\ref{Resolvability})} & \cellcolor{white}$-$&\cellcolor{white}$\checkmark$&\cellcolor{white}$\checkmark$ &\cellcolor{white}$\checkmark$&\cellcolor{white}$-$ &\cellcolor{white}$-$ &\cellcolor{white}$-$ &\cellcolor{white}$\checkmark$ &\cellcolor{white}$\checkmark$\\\Xhline{3\arrayrulewidth}

\makecell[l]{{Non-negative}\\[-2pt] {Responsiveness (\ref{MonotonicitySection})}} & \cellcolor{gray!50}{$\checkmark$} & \cellcolor{gray!50}{$\checkmark$}&\cellcolor{gray!50}{$\checkmark$} &\cellcolor{gray!50}{$\checkmark$} & \cellcolor{gray!50}{$\checkmark$}&\cellcolor{gray!50}{$\checkmark$} &\cellcolor{gray!50}{$\checkmark$}& \cellcolor{gray!50}{$-$} & \cellcolor{gray!50}$\checkmark$\\\hline

\makecell[l]{{Positive}\\[-2pt] {Involvement (\ref{InvolvementSection})}} & \cellcolor{white}$\checkmark$&\cellcolor{white}$-$& \cellcolor{white}$-$&\cellcolor{white}$\checkmark$ & \cellcolor{white}$-$ & \cellcolor{white}$-$ &\cellcolor{white}$-$ & \cellcolor{white}$\checkmark$ &\cellcolor{white}$\checkmark$ \\\hline

\makecell[l]{{Negative}\\[-2pt] {Involvement (\ref{InvolvementSection})} } & \cellcolor{gray!50}$\checkmark$&\cellcolor{gray!50}$-$& \cellcolor{gray!50}$-$&\cellcolor{gray!50}$\checkmark$ & \cellcolor{gray!50}$-$ & \cellcolor{gray!50}$-$ & \cellcolor{gray!50}$-$ & \cellcolor{gray!50}$-$ &\cellcolor{gray!50}$\checkmark$ \\\hline

\end{tabular}
\end{center}
\caption{Comparison of Split Cycle to standard voting methods in terms of selected voting criteria. A $\checkmark$ indicates that the criterion is satisfied, while $-$ indicates that it is not. The $\checkmark^\star$ indicates that Ranked Pairs satisfies immunity to stealers in uniquely-weighted election profiles (Definition \ref{MarginGraphDef}) but not in general. The $\checkmark^*$ indicates that there are subtleties in how one must define Ranked Pairs to ensure full  independence of clones (together with anonymity), as discussed in Remark \ref{RPnote}. For the Uncovered Set column, there are several definitions of the Uncovered Set that are equivalent for an odd number of voters with linear ballots but inequivalent in general; the $\checkmark^\dagger$ indicates that while one version of the Uncovered Set (\citealt{Fishburn1977}) fails to satisfy independence of clones and expansion consistency for all profiles, other definitions satisfy both axioms for all profiles, and all definitions do so for profiles with an odd number of voters with linear ballots. The $\checkmark^\ddagger$ indicates that whether Instant Runoff satisfies independence of clones depends on how ties for the fewest first-place votes are handled. For proofs of these claims and those in the table about voting methods other than Split Cycle, see Appendix~\ref{OtherMethodsAppendix}.}\label{AxiomTable}

\end{figure}

\newpage

\section{Preliminaries}\label{PrelimSection}

\subsection{Profiles, Margin Graphs, and Voting Methods}\label{ProfilesSection}

Fix infinite sets $\mathcal{V}$ and $\mathcal{X}$ of \textit{voters} and \textit{candidates}, respectively. A given election will involve only finite subsets $V\subseteq\mathcal{V}$ and $X\subseteq\mathcal{X}$, but we want no upper bound on the number of voters or candidates who may participate in elections. A binary relation $P$ on $X$ is \textit{asymmetric} if for all $x,y\in X$, if $xPy$, then \textit{not} $yPx$. Let  $\mathcal{B}(X)$ be the set of all asymmetric binary relations on $X$.

\begin{definition} A \textit{profile} is a pair $(\mathbf{P}, X(\mathbf{P}))$ where $\mathbf{P}: V(\mathbf{P})\to \mathcal{B}(X(\mathbf{P}))$ for some nonempty finite $X(\mathbf{P})\subseteq \mathcal{X}$ and nonempty finite $V(\mathbf{P})\subseteq \mathcal{V}$. We conflate the profile with the function $\mathbf{P}$.\footnote{We officially define a profile as a pair $(\mathbf{P}, X(\mathbf{P}))$ due to a technicality: unlike the set of voters, the set of candidates cannot necessarily be recovered from the function $\mathbf{P}$.} We call $X(\mathbf{P})$  and $V(\mathbf{P})$ the sets of \textit{candidates in $\mathbf{P}$} and \textit{voters in $\mathbf{P}$}, respectively. We call $\mathbf{P}(i)$ voter $i$'s \textit{ballot}, and we write `$x\mathbf{P}_iy$' for $(x,y)\in\mathbf{P}(i)$.
\end{definition}
\noindent As usual, we take $x\mathbf{P}_iy$ to mean that voter $i$ strictly prefers candidate $x$  to candidate $y$. It is standard to assume that $\mathbf{P}_i$ satisfies additional constraints beyond asymmetry, such as transitivity and even negative transitivity (if not $x\mathbf{P}_iy$ and not $y\mathbf{P}_iz$, then not $y\mathbf{P}_iz$). More generally, one may consider the following classes of profiles: $\mathscr{P}$, the class of all profiles; $\mathscr{A}$, the class of \textit{acyclic profiles}, in which each voter's ballot is \textit{acyclic}, meaning that there are no $x_1,\dots,x_n\in X(\mathbf{P})$ with $n>1$ such that for $k\in \{1,\dots,n-1\}$, we have $x_{k}\mathbf{P}_i x_{k+1}$, and $x_n= x_1$; $\mathscr{S}$, the class of \textit{strict weak order profiles}, in which each voter's ballot is a \textit{strict weak order}, meaning that it is asymmetric and negatively transitive (which together imply transitivity); $\mathscr{L}$, the class of \textit{linear profiles},  in which each voter's ballot is a \textit{linear order}, meaning that it is transitive and for all $x,y\in X(\mathbf{P})$ with $x\neq y$, we have either $x\mathbf{P}_iy$ or $y\mathbf{P}_ix$. 

Proving that a voting method satisfies some universal axiom with respect to a larger class of profiles, like $\mathscr{P}$ or $\mathscr{A}$, is obviously stronger than proving that it satisfies the universal axiom with respect to a smaller class of profiles, like $\mathscr{L}$. On the other hand, proving that a voting method does not satisfy some universal axiom with respect to a smaller class of profiles, like $\mathscr{L}$, is stronger than proving that it does not satisfy the universal axiom with respect to a larger class of profiles, such as $\mathscr{P}$ or $\mathscr{A}$.

\begin{remark} By not requiring linear profiles, we can represent elections in which voters are not required to rank all the candidates up for election. Depending on the official interpretation of what it means to leave a candidate unranked, a ballot with unranked candidates could mean either that (i) all ranked candidates are strictly preferred to all unranked candidates, and there are no strict preferences between unranked candidates or that (ii) there are no strict preferences at all involving unranked candidates.\end{remark}

Next we define the notions of an abstract margin graph and the margin graph of a particular profile.

\begin{definition} A \textit{margin graph}  is a weighted directed graph $\mathcal{M}$ with positive integer weights whose edge relation is asymmetric. We say $\mathcal{M}$ has \textit{uniform parity} if all weights of edges are even or all weights of edges are odd, and if there are two nodes with no edge between them, then all weights are even.\end{definition}

Two examples of margin graphs already appeared in Section \ref{SpoilerSection}.

\begin{definition}\label{MarginGraphDef} Let $\mathbf{P}$ be a profile and $a,b\in X(\mathbf{P})$. Then \[Margin_\mathbf{P}(a,b)=\vert \{i\in V(\mathbf{P})\mid a\mathbf{P}_ib\}\vert -\vert \{i\in V(\mathbf{P})\mid b\mathbf{P}_ia\}\vert .\]
The \textit{margin graph of $\mathbf{P}$}, $\mathcal{M}(\mathbf{P})$, is the weighted directed graph whose set of nodes is $X(\mathbf{P})$ with an edge from $a$ to $b$ weighted by $Margin_\mathbf{P}(a,b)$ when $Margin_\mathbf{P}(a,b)>0$, in which case we say that \textit{$a$ is majority preferred to $b$}. We write $a\overset{\alpha}{\to}_\mathbf{P}b\mbox{ if }\alpha = Margin_\mathbf{P}(a,b)> 0$, omitting the $\alpha$ when the size of the margin is not important and $\mathbf{P}$ when the profile in question is clear. We say that $\mathbf{P}$ is \textit{uniquely weighted} if for all $x,y,x',y'\in X(\mathbf{P})$, if $x\neq y$, $x'\neq y'$, and $(x,y)\neq (x',y')$, then $Margin_\mathbf{P}(x,y)\neq Margin_\mathbf{P}(x',y')$. 

We call the unweighted directed graph underlying $\mathcal{M}(\mathbf{P})$ the \textit{majority graph} of $\mathbf{P}$, denoted $M(\mathbf{P})$, and we call the edge relation of $M(\mathbf{P})$ the \textit{majority relation} of $\mathbf{P}$.\end{definition}

The key fact about the relation between margin graphs and profiles is given by Debord's Theorem.

\begin{theorem}[\citealt{Debord1987}]\label{DebordThm} For any margin graph $\mathcal{M}$, there is a strict weak order profile $\mathbf{P}$ such that $\mathcal{M}$ is the margin graph of $\mathbf{P}$; and if $\mathcal{M}$ has uniform parity, then there is a linear profile $\mathbf{P}$ such that $\mathcal{M}$ is the margin graph of $\mathbf{P}$.
\end{theorem}

Finally, we define what we mean by a voting method for the purposes of this paper.

\begin{definition}\label{VotingMethodDef} Given a set $\mathscr{D}$ of profiles, a \textit{voting method on $\mathscr{D}$} is a function $F$ such that for all profiles $\mathbf{P}\in\mathscr{D}$, we have $\varnothing\neq F(\mathbf{P})\subseteq X(\mathbf{P})$. We call $F(\mathbf{P})$ the \textit{set of winners} or \textit{winning set for $\mathbf{P}$ under $F$}. We write $\mathrm{dom}(F)$ for the set $\mathscr{D}$ on which $F$ is defined.
\end{definition}

\noindent As usual, if $F(\mathbf{P})$ contains multiple winners, we assume that some further tiebreaking process would then apply, though we do not fix the nature of this process (see \citealt[pp.~14-5]{Schwartz1986} for further discussion). Options include the use of a deterministic tiebreaking procedure, an even-chance lottery on $F(\mathbf{P})$, a runoff election with the candidates in $F(\mathbf{P})$ in which a different set of voters may participate, etc.

\subsection{Operations on Profiles}\label{OperationsSection}

Sometimes we will be interested in combining two profiles for the same set of candidates and disjoint sets of voters, for which we use the following notation.

\begin{definition}\label{DisjointUnion} Given profiles $\mathbf{P}$ and $\mathbf{P}'$ such that $X(\mathbf{P})=X(\mathbf{P}')$ and $V(\mathbf{P})\cap V(\mathbf{P}')=\varnothing$, we define the profile $\mathbf{P}+\mathbf{P}': V(\mathbf{P})\cup V(\mathbf{P}')\to \mathcal{B}(X(\mathbf{P}))$ such that $(\mathbf{P}+\mathbf{P}')(i)=\mathbf{P}(i)$ if $i\in V(\mathbf{P}) $ and $(\mathbf{P}+\mathbf{P}')(i)=\mathbf{P}'(i)$ if  $i\in V(\mathbf{P}')$. To add $\mathbf{P}$ to itself, we may take $\mathbf{P}+\mathbf{P}^*$ where $\mathbf{P}^*$ is a copy of $\mathbf{P}$ with a disjoint set of voters.\footnote{I.e., $X(\mathbf{P})=X(\mathbf{P}^*)$, $V(\mathbf{P})\cap V(\mathbf{P}^*)=\varnothing$, and there is a bijection $h:V(\mathbf{P})\to V(\mathbf{P}^*)$ such that for all $i\in V(\mathbf{P})$ and $x,y\in X(\mathbf{P})$, we have $x\mathbf{P}_i y$ if and only if $x\mathbf{P}^*_{h(i)} y$.}
\end{definition}

We will also be interested in deleting some candidates from every ballot in a profile, as follows.

\begin{definition} Given a profile $\mathbf{P}$ and nonempty $Y\subseteq X(\mathbf{P})$, define the restricted profile $\mathbf{P}\vert _{ Y}$ to be the profile with $X(\mathbf{P}\vert _{ Y})=Y$ and $V(\mathbf{P}\vert _{ Y})=V(\mathbf{P})$ such that for each $i\in V(\mathbf{P}\vert _{ Y})$, $\mathbf{P}\vert _{ Y}(i)$ is the restriction of the relation $\mathbf{P}(i)$ to $Y$. As a special case, when $\vert X(\mathbf{P})\vert >1$ and $x\in X(\mathbf{P})$, let $\mathbf{P}_{-x} = \mathbf{P}\vert _{X(\mathbf{P})\setminus\{x\}}$, i.e., the result of removing candidate $x$ from each ballot.\end{definition}

\section{Split Cycle}\label{SplitCycleSection}

\subsection{Three Main Ideas}\label{ThreeIdeasSection}

The ``Paradox of Voting'' is the phenomenon that cycles may occur in the margin graph of a profile, e.g., $a$ is majority preferred to $b$, $b$ is majority preferred to $c$, and $c$ is majority preferred to $a$. Recall the formal definition of a  cycle.

\begin{definition} Given a directed graph $\mathcal{G}$ (e.g., a margin graph), a \textit{path in $\mathcal{G}$} is a sequence $\langle x_1,\dots,x_n\rangle$ of nodes from $\mathcal{G}$ such that $n>1$ and for all $i\in \{1,\dots,n-1\}$, we have $x_i\to x_{i+1}$, where $\to $ is the edge relation of the graph. A \textit{cycle in $\mathcal{G}$} is defined in the same way but requiring $x_1=x_n$. The cycle is \textit{simple} if for all distinct $i,j\in \{1,\dots,n\}$, $x_i=x_j$  only if $i,j\in \{1,n\}$ (i.e., all nodes are distinct except $x_1=x_n$).\end{definition}

The voting method we propose in this paper, Split Cycle, provides a way of dealing with the problem of majority cycles. It is based on three main ideas:

\begin{paragraph}{1. Group incoherence raises the threshold for one candidate to defeat another, but not infinitely.} By ``group incoherence'' we mean cycles in the majority relation. Consider the margin graph on the left again:

\begin{center}
\begin{minipage}{2in}\begin{tikzpicture}

\node[circle,draw, minimum width=0.25in] at (0,0) (a) {$a$}; 
\node[circle,draw,minimum width=0.25in] at (3,0) (b) {$b$}; 
\node[circle,draw,minimum width=0.25in] at (1.5,1.5) (c) {$c$}; 

\path[->,draw,thick] (c) to node[fill=white] {$5$} (b);
\path[->,draw,thick] (b) to node[fill=white] {$3$} (a);
\path[->,draw,thick] (a) to node[fill=white] {$1$} (c);

\end{tikzpicture}
\end{minipage} \begin{minipage}{2in}\begin{tikzpicture}

\node[circle,draw, minimum width=0.25in] at (0,0) (a) {$a$}; 
\node[circle,draw,minimum width=0.25in] at (3,0) (b) {$b$}; 
\node[circle,draw,minimum width=0.25in] at (1.5,1.5) (c) {$c$}; 

\path[->,draw,thick] (c) to node[fill=white] {$5$} (b);
\path[->,draw,thick] (b) to node[fill=white] {$3$} (a);

\end{tikzpicture}
\end{minipage}\end{center}
Due to the group incoherence, the margin of $1$ for $a$ over $c$ is not sufficient for $a$ to defeat $c$. But if we raise the threshold for defeat to \textit{winning by more than $1$}, and we redraw the graph with an arrow from $x$ to $y$ if and only if $Margin_\mathbf{P}(x,y)>1$, as on the right, then the group is no longer incoherent at this threshold. Since the group is no longer incoherent with respect to the \textit{win by more than $1$} threshold, we think it is reasonable to take $c$ to defeat $b$ and $b$ to defeat $a$, leaving $c$ as the winner. Thus, as suggested, group incoherence does not raise the threshold for $b$ to defeat $a$ \textit{infinitely} but rather only enough to eliminate any incoherence in which $b$ and $a$ are involved. This  shows that our proposal differs from the GETCHA and GOCHA methods (Section \ref{SmithSchwartz}), which take all 3-cycles to result in three-way ties regardless of the margins.\footnote{Cf.~Tideman \citeyearpar[p.~206]{Tideman1987}: ``The GOCHA rule, in a sense, is only half a voting rule. It does not address the issue of what should be done to resolve cycles.''}\end{paragraph}

\begin{paragraph}{2. Incoherence can be localized.} Consider the following margin graph: 
\begin{center}
\begin{minipage}{2in}\begin{tikzpicture}

\node[circle,draw, minimum width=0.25in] at (0,0) (a) {$a$}; 
\node[circle,draw,minimum width=0.25in] at (3,0) (b) {$b$}; 
\node[circle,draw,minimum width=0.25in] at (1.5,1.5) (c) {$c$}; 

\node[circle,draw,minimum width=0.25in] at (1.5,-1.5) (d) {$d$}; 

\path[->,draw,thick] (c) to node[fill=white] {$3$} (b);
\path[->,draw,thick] (b) to[pos=.7] node[fill=white] {$3$} (a);
\path[->,draw,thick] (a) to node[fill=white] {$3$} (c);

\path[->,draw,thick] (c) to[pos=.7] node[fill=white] {$1$} (d);
\path[->,draw,thick] (b) to node[fill=white] {$1$} (d);
\path[->,draw,thick] (a) to node[fill=white] {$1$} (d);

\end{tikzpicture}
\end{minipage}\end{center}
It would be a mistake to think that the margin of 1 for $a$ over $d$ is not sufficient for $a$ to defeat $d$, due to the incoherence involving $a$, $b$, and $c$, which is only eliminated by raising the threshold to \textit{win by more than 3}. For there is no incoherence with respect to $d$ and the other candidates, all of whom are majority preferred to $d$, so they all defeat $d$. The lesson from this example is that when deciding whether the margin of $a$ over $d$ is sufficient for $a$ to defeat $d$, we set the threshold in terms of \textit{the cycles \textnormal{(}if any\textnormal{)} involving $a$ and $d$}. This shows that our proposal differs from the Minimax method, which takes the winner in the example above to be the Condorcet loser~$d$ (see Definition \ref{CondWinLoss}).\end{paragraph}

\begin{paragraph}{3. Defeat is direct.} On our view, for a candidate $x$ to defeat a candidate $y$, so that $y$ is not in the set of winners, $x$ must have a positive margin over $y$. Consider the following margin graph (note that if there is no edge between two candidates, then the margin of each candidate over the other is 0):
\begin{center}
\begin{tikzpicture}

\node[circle,draw, minimum width=0.25in] at (0,0) (b) {$b$}; 
\node[circle,draw,minimum width=0.25in] at (3,0) (a) {$a$}; 
\node[circle,draw,minimum width=0.25in] at (1.5,1.5) (c) {$c$};  

\path[->,draw,thick] (b) to node[fill=white] {$4$} (a);
\path[->,draw,thick] (c) to node[fill=white] {$4$} (b);
\path[->,draw,thick] (a) to node[fill=white] {$4$} (c);

\node[circle,draw, minimum width=0.25in] at (5,0) (e) {$e$}; 
\node[circle,draw,minimum width=0.25in] at (8,0) (d) {$d$}; 
\node[circle,draw,minimum width=0.25in] at (6.5,1.5) (f) {$f$};  

\path[->,draw,thick] (e) to node[fill=white] {$4$} (d);
\path[->,draw,thick] (f) to node[fill=white] {$4$} (e);
\path[->,draw,thick] (d) to node[fill=white] {$4$} (f);

\path[->,draw,thick, bend left] (a) to node[fill=white] {$4$} (f);
\path[->,draw,thick, bend left] (d) to node[fill=white] {$2$} (a);

  \end{tikzpicture}
  \end{center}
In this case, we think $a$ should defeat $f$, but $a$ should not defeat $d$. Some other voting methods, such as Beat Path, commit one to a view that we find dubious: that even though $a$ is not majority preferred to $d$, nonetheless $a$ should kick $d$ out of the set of winners because of the indirect path from $a$ to $f$ to $e$ to $d$ with margins of 4 at each step. By contrast, we adopt a direct pairwise perspective: \textit{for $a$ to kick $d$ out of the winning set, $a$ must be majority preferred to $d$}. We find it difficult to try to explain to $d$'s supporters that although $a$ was not majority preferred to $d$, nonetheless $a$ kicks $d$ out of the winning set because of $a$'s relation to \textit{other candidates}, $f$ and $e$, \textit{neither of whom defeat $d$}!\footnote{We assume that $e$ does not defeat $d$ because of the perfect cycle involving $d$, $f$, and $e$.} Of course reasonable definitions of defeat cannot fully satisfy the independence of irrelevant alternatives (IIA) criterion (\citealt{Arrow1963}),\footnote{Here we take IIA to state that if two profiles are alike with respect to how everyone votes on $x$ vs.~$y$, then $x$ defeats $y$ in the one profile if and only if $x$ defeats $y$ in the other.} but in our view this seems too flagrant a violation of the idea behind IIA. We endorse the following weakening of IIA, known as \textit{weak IIA} (\citealt{Baigent1987}): if two profiles are alike with respect to how everyone votes on $x$ vs.~$y$, then it should not be possible that in  one profile, $x$ defeats $y$, while in the other, $y$ defeats $x$ (though it should be possible that in one, $x$ defeats $y$, while in the other, neither $x$ defeats $y$ nor $y$ defeats $x$, due to a cycle). Let $\mathbf{P}$ be a profile whose margin graph is shown above, and let $\mathbf{P}'$ be a profile just like $\mathbf{P}$ with respect to how everyone votes on $a$ vs. $d$ but in which all voters have either $a$ followed by $d$ or $d$ followed by $a$ at the top of their ballots, followed by the linear order $b\mathbf{P}_i'c\mathbf{P}_i'e\mathbf{P}_i'f$. In $\mathbf{P}'$, since $d$ is majority preferred to $a$ by 2 and there are no cycles, surely $d$ should defeat $a$, kicking $a$ out of the winning set. Then it follows by weak IIA that in $\mathbf{P}$, $a$ does not defeat $d$. Thus, weak IIA is inconsistent with the indirect notion of defeat according to Beat Path. By contrast, it is satisfied by the direct notion of defeat we will define for Split~Cycle.\footnote{In fact, in \citealt{HP2021}, we characterize the Split Cycle defeat relation using an axiom of Coherent IIA that is stronger than weak IIA.}\end{paragraph}

\subsection{Defining Split Cycle}\label{Defining}

To define Split Cycle, in line with our first idea above, we first measure the degree of incoherence of a cycle by the smallest margin occurring on an edge in the cycle---for if we raise our threshold above that margin, then we split the cycle, restoring coherence at the higher threshold as in the second graph in Section \ref{ThreeIdeasSection}.

\begin{definition} Let $\mathbf{P}$ be a profile and $\rho$ a simple cycle in $\mathcal{M}(\mathbf{P})$. The \textit{splitting number} of $\rho$, $Split\#_\mathbf{P}(\rho)$, is the smallest margin between consecutive candidates in $\rho$ (e.g., the splitting number of $a\overset{3}{\to} b \overset{1}{\to} c \overset{5}{\to} a$ is~$1$). We omit the subscript for $\mathbf{P}$ when the profile is clear from context. 
\end{definition}
\noindent Thus, for example, the splitting number of the cycle in the three-candidate margin graph in Section \ref{ThreeIdeasSection} is 1, while the splitting number of the cycle in the four-candidate margin graph in Section \ref{ThreeIdeasSection} is 3.

In line with our second idea that incoherence can be localized, when deciding whether $a$ defeats $b$, we look at all and only the simple cycles containing $a$ and $b$ (not at the other cycles that do not contain $a$ and $b$); and in line with our third idea about the directness of defeat, for $a$ to defeat $b$, we require that \textit{the direct margin of $a$ over $b$} exceeds the splitting number of every simple cycle containing $a$ and $b$, which means that  that direct margin survives after we raise the threshold above those splitting numbers.

\begin{definition}\label{DefeatDef} Let $\mathbf{P}$ be a profile and $a,b\in X(\mathbf{P})$. Then \textit{$a$ defeats $b$ in $\mathbf{P}$} if $Margin_\mathbf{P}(a,b)>0$ and
\[Margin_\mathbf{P}(a,b)>Split\#(\rho)\mbox{ for every simple cycle $\rho$ in }\mathcal{M}(\mathbf{P})\mbox{ containing $a$ and $b$}.\]
A candidate $b$ is \textit{undefeated in $\mathbf{P}$} if there is no candidate who defeats $b$.
\end{definition}

\begin{remark} Just as some sports have a \textit{win~by~$2$} rule for defeat, Split Cycle says that for $a$ to defeat $b$, $a$ must win by \textit{more than} $n$ over $b$, where $n$ is the smallest number such that there are no cycles involving $a$ and $b$ in the \textit{wins by more than $n$} relation, defined by $xW_\mathbf{P}^ny$ if $Margin_\mathbf{P}(x,y)>n$. 
\end{remark}

Finally, we can define the voting method we call Split Cycle:

\begin{definition}\label{SCDef} For any profile $\mathbf{P}$, the set of Split Cycle winners, $SC(\mathbf{P})$, is the set of candidates who are undefeated in $\mathbf{P}$.
\end{definition}

\noindent As explained in Section \ref{Introduction}, one can determine  $SC(\mathbf{P})$ in a simple two-step process (see Footnote \ref{FloydWarshall} for a faster algorithm): 1. For each simple cycle, identify the edges with the smallest margin in that cycle. 2. After completing step 1 for all simple cycles, discard the identified edges. All remaining edges count as defeats.

\begin{remark}\label{SCofMarginGraph} Since the only information Split Cycle uses about a profile $\mathbf{P}$ is its margin graph, we can also think of Split Cycle as assigning to each margin graph $\mathcal{M}$ a set $SC(\mathcal{M})$ of winners.
\end{remark}

Let us consider some examples of calculating the set of Split Cycle winners.

\begin{example} The Split Cycle winners for the margin graphs illustrating our three main ideas in Section \ref{ThreeIdeasSection} are as follows: in the three-candidate example, the unique Split Cycle winner is $c$; in the four-candidate example, the Split Cycle winners are $a$, $b$, and $c$; and in the six-candidate example, the Split Cycle winners are all candidates except $f$.\end{example}

\begin{example}\label{OverlappingCycles} For a more complicated example, consider the following margin graph, repeated three times to highlight the three different simple cycles:
\begin{center}
\begin{minipage}{1.5in}\begin{tikzpicture}

\node[circle,draw, minimum width=0.25in] at (0,0) (b) {$b$}; 
\node[circle,draw,minimum width=0.25in] at (3,0) (a) {$a$}; 
\node[circle,draw,minimum width=0.25in] at (1.5,1.5) (c) {$c$}; 
\node[circle,draw,minimum width=0.25in] at (1.5,-1.5) (d) {$d$}; 

\path[->,draw,thick] (b) to (a);
\path[->,draw,thick,red] (d) to (c);
\path[->,draw,thick,red] (c) to node[fill=white] {$8$} (b);
\path[->,draw,thick] (a) to node[fill=white] {$6$} (c);
\path[->,draw,thick] (a) to node[fill=white] {$4$} (d);
\path[->,draw,thick,red] (b) to node[fill=white] {$4$} (d);

\node[fill=white] at (1.5,.5)  {{\color{red}$6$}}; 
\node[fill=white] at (2,0)  {$8$}; 

\node[circle,draw,minimum width=0.25in] at (1.5,-3.5) (e) {$e$}; 
\path[->,draw,thick] (d) to node[fill=white] {$2$} (e);

  \end{tikzpicture}
\end{minipage}\hspace{.25in}\begin{minipage}{1.5in}\begin{tikzpicture}

\node[circle,draw, minimum width=0.25in] at (0,0) (b) {$b$}; 
\node[circle,draw,minimum width=0.25in] at (3,0) (a) {$a$}; 
\node[circle,draw,minimum width=0.25in] at (1.5,1.5) (c) {$c$}; 
\node[circle,draw,minimum width=0.25in] at (1.5,-1.5) (d) {$d$}; 

\path[->,draw,thick,blue] (b) to (a);
\path[->,draw,thick] (d) to (c);
\path[->,draw,thick,blue] (c) to node[fill=white] {$8$} (b);
\path[->,draw,thick,blue] (a) to node[fill=white] {$6$} (c);
\path[->,draw,thick] (a) to node[fill=white] {$4$} (d);
\path[->,draw,thick] (b) to node[fill=white] {$4$} (d);

\node[fill=white] at (1.5,.5)  {$6$}; 
\node[fill=white] at (2,0)  {{\color{blue}$8$}}; 

\node[circle,draw,minimum width=0.25in] at (1.5,-3.5) (e) {$e$}; 
\path[->,draw,thick] (d) to node[fill=white] {$2$} (e);

  \end{tikzpicture}
\end{minipage}\hspace{.25in}\begin{minipage}{1.5in}\begin{tikzpicture}

\node[circle,draw, minimum width=0.25in] at (0,0) (b) {$b$}; 
\node[circle,draw,minimum width=0.25in] at (3,0) (a) {$a$}; 
\node[circle,draw,minimum width=0.25in] at (1.5,1.5) (c) {$c$}; 
\node[circle,draw,minimum width=0.25in] at (1.5,-1.5) (d) {$d$}; 

\path[->,draw,thick,medgreen] (b) to (a);
\path[->,draw,thick,medgreen] (d) to (c);
\path[->,draw,thick,medgreen] (c) to node[fill=white] {$8$} (b);
\path[->,draw,thick] (a) to node[fill=white] {$6$} (c);
\path[->,draw,thick,medgreen] (a) to node[fill=white] {$4$} (d);
\path[->,draw,thick] (b) to node[fill=white] {$4$} (d);

\node[fill=white] at (1.5,.5)  {{\color{medgreen}$6$}}; 
\node[fill=white] at (2,0)  {{\color{medgreen}$8$}}; 

\node[circle,draw,minimum width=0.25in] at (1.5,-3.5) (e) {$e$}; 
\path[->,draw,thick] (d) to node[fill=white] {$2$} (e);

  \end{tikzpicture}
\end{minipage}\end{center}
The splitting number of the cycle $b\to d \to c\to b$ is 4; the splitting number of the cycle $b\to a \to c\to b$ is 6; and the splitting number of the cycle $b\to a\to d\to c\to b$ is 4. In each cycle, the edge with the smallest margin in that cycle is not a defeat. After discarding these edges (i.e., the $b\to d$ edge in the red cycle, the $a\to c$ edge in the blue cycle, and the $a\to d$ edge in the green cycle), the remaining edges are defeats:
\begin{center}
\begin{minipage}{1.5in}\begin{tikzpicture}

\node[circle,draw, minimum width=0.25in] at (0,0) (b) {$b$}; 
\node[circle,draw,minimum width=0.25in] at (3,0) (a) {$a$}; 
\node[circle,draw,minimum width=0.25in] at (1.5,1.5) (c) {$c$}; 
\node[circle,draw,minimum width=0.25in] at (1.5,-1.5) (d) {$d$}; 

\path[->,draw,thick] (b) to (a);
\path[->,draw,thick] (d) to (c);
\path[->,draw,thick] (c) to node[fill=white] {$D$} (b);

\node[fill=white] at (1.5,.5)  {$D$}; 
\node[fill=white] at (2,0)  {$D$}; 

\node[circle,draw,minimum width=0.25in] at (1.5,-3.5) (e) {$e$}; 
\path[->,draw,thick] (d) to node[fill=white] {$D$} (e);

  \end{tikzpicture}
  \end{minipage}
  \end{center}
Since $d$ is the only undefeated candidate, $d$ is the unique Split Cycle winner.\end{example}

Let us now show that the set of Split Cycle winners is always nonempty.

\begin{lemma}\label{NoCycles} For a profile $\mathbf{P}$, let the \textit{defeat graph of $\mathbf{P}$} be the directed graph whose set of nodes is $X(\mathbf{P})$ with an edge from $a$ to $b$ when $a$ defeats $b$ in $\mathbf{P}$.  Then for any profile $\mathbf{P}$, the defeat graph of $\mathbf{P}$ contains no cycles. Thus, $SC(\mathbf{P})\neq\varnothing$.
\end{lemma}
\begin{proof} Suppose there is a cycle $a_1Da_2D\dots Da_nDa_1$ in the defeat graph of $\mathbf{P}$, which we may assume is simple (since if there is any cycle, there is a simple one). This yields a simple cycle $\rho=a_1 \overset{\alpha_1}{\longrightarrow} a_2 \overset{\alpha_2}{\longrightarrow} \dots \overset{\alpha_{n-1}}{\longrightarrow} a_n \overset{\alpha_n}{\longrightarrow} a_1$ in  $\mathcal{M}(\mathbf{P})$ where each margin $\alpha_i$ is greater than the splitting number of any simple cycle containing $a_i,a_{i+1\;\mathrm{mod}\; n}$ and hence greater than the splitting number of $\rho$ itself, which is impossible.\end{proof}

\begin{remark} Like defeat relations in sports tournaments, the Split Cycle defeat relation is not necessarily transitive: it may be, as in Example \ref{OverlappingCycles}, that $d$ defeated $c$, and $c$ defeated $b$, while $d$ is not among those who defeated $b$---nonetheless, $b$ is not among the winners of the tournament, having been defeated by $c$.  \textit{Acyclicity}, as in Lemma \ref{NoCycles}, is sufficient for there always to be a nonempty set of winners---transitivity is not required. That the Split Cycle defeat relation is acyclic but not necessarily transitive explains how it can satisfy weak IIA without contradicting Baigent's \citeyearpar{Baigent1987} generalization of Arrow's impossibility theorem (cf.~\citealt{Campbell2000}), which states that under Arrow's axioms but with IIA weakened to weak IIA, there must be a weak dictator (a voter $i$ such that if $i$ prefers $x$ to $y$, then $y$ does not defeat $x$ socially). Baigent's theorem requires that the social defeat relation is not only acyclic but a strict weak order.\footnote{Baigent's theorem also assumes that profiles assign strict weak orders to voters, not just linear orders, but his result also holds for the domain of all linear profiles (also see \citealt{Campbell2000}).} Implicit here is that we can view Split Cycle as a \textit{collective choice rule}, i.e., a function mapping each profile $\mathbf{P}$ to a binary relation on $X(\mathbf{P})$ (cf.~\citealt[Ch.~2*]{Sen2017}), by taking the binary relation to be the defeat relation. This is the perspective on Split Cycle adopted in \citealt{HP2021}. However, in this paper we focus on Split Cycle as a voting method (as in Definition \ref{VotingMethodDef}) that maps each profile to a set of winners.\end{remark}

Another useful lemma about Split Cycle is that if a candidate $z$ is not a winner for a profile $\mathbf{P}$, then there is some winner $x$ and a path in the defeat graph of $\mathbf{P}$ from $x$ to $z$.

\begin{lemma}\label{BeatPathFromWinner} For any profile $\mathbf{P}$ and $z\in X(\mathbf{P})\setminus SC(\mathbf{P})$, there is an $x\in SC(\mathbf{P})$ and distinct ${y_1,\dots,y_n\in X(\mathbf{P})}$ with $y_1=x$ and $y_n=z$ such that $y_1Dy_{2}D\dots Dy_{n-1} D y_n$.
\end{lemma}
\begin{proof} We first find $w_1,\dots,w_n\in X(\mathbf{P})$ such that $w_nDw_{n-1}D\dots Dw_2 D w_1$ and then relabel $w_1,\dots,w_n$ as $y_n,\dots,y_1$, so that $y_1Dy_{2}D\dots Dy_{n-1} D y_n$. If $z\in X(\mathbf{P})\setminus SC(\mathbf{P})$, then setting $w_1=z$, there is a $w_2$ such that $w_2Dz$. If $w_2\in SC(\mathbf{P})$, then we are done with $x=w_2$; otherwise, there is a $w_3$ such that $w_3 Dw_2$; and so on. Since $X(\mathbf{P})$ is finite and there are no cycles in the defeat graph of $\mathbf{P}$ by Lemma \ref{NoCycles}, we eventually find the desired $w_n\in SC(\mathbf{P})$.\end{proof}

Yet another useful lemma about Split Cycle is that to check whether $a$ defeats $b$, it suffices to check the splitting number of just the simple cycles in which $b$ immediately follows $a$, rather than all simple cycles containing $a$ and $b$.

\begin{lemma}\label{OnlySomeCycles} Let $\mathbf{P}$ be a profile and $a,b\in X(\mathbf{P})$. Then $a$ defeats $b$ in $\mathbf{P}$ if and only if $Margin_\mathbf{P}(a,b)>0$ and \[Margin_\mathbf{P}(a,b)>Split\#(\rho)\mbox{ for every simple cycle $\rho$ in }\mathcal{M}(\mathbf{P})\mbox{ of the form } a \rightarrow b\rightarrow x_1\rightarrow \dots\rightarrow x_n\rightarrow a.\]\end{lemma}
\begin{proof}  Obviously if $Margin_\mathbf{P}(a,b)$ is greater than the splitting number of every simple cycle containing $a$ and $b$, then it is greater than the splitting number of every simple cycle of the form $a \rightarrow b\rightarrow x_1\rightarrow \dots\rightarrow x_n\rightarrow a$. Conversely, assume $Margin_\mathbf{P}(a,b)$ is greater than the splitting number of every simple cycle of the form $a \rightarrow b\rightarrow x_1\rightarrow \dots\rightarrow x_n\rightarrow a$. To show that $Margin_\mathbf{P}(a,b)$ is greater than the splitting number of every simple cycle containing $a$ and $b$, let $\rho$ be a simple cycle containing $a$ and $b$ whose splitting number is maximal among all such cycles. If $\rho$ contains $a\to b$, then we are done. So suppose $\rho$ does not contain $a\to b$. Without loss of generality, we may assume $\rho$ is of the form $b\to x_1\to \dots \to x_n\to a \to y_1 \to\dots \to y_m\to b$. Let $\rho'$ be $a\to b \to x_1\to\dots \to x_n \to a$. It follows from our initial assumption that the splitting number of $\rho'$ is not equal to the margin of the $a\to b$ edge.  Since $\rho$ has maximal splitting number of any simple cycle containing $a$ and $b$, it follows that one of the edges in $\rho'$ after the $a\to b$ edge has this splitting number as its margin; for if none of the edges in $\rho'$ after the $a\to b$  edge has this splitting number as its margin, then since the splitting number is defined as a minimum,  $\rho'$ has a higher splitting number than $\rho$, contradicting the fact that $\rho$ has maximal splitting number of any simple cycle containing $a$ and $b$. Thus, $\rho'$ has splitting number at least that of $\rho$, and by assumption $Margin_\mathbf{P}(a,b)> Split\#(\rho')$, so we have $Margin_\mathbf{P}(a,b)>Split\#(\rho)$. Thus, $Margin_\mathbf{P}(a,b)$ is greater than the splitting number of every simple cycle containing $a$ and $b$.\end{proof}

\begin{remark}\label{HeitzigRemark} After submitting this paper, we learned from Markus Schulze that Lemma \ref{OnlySomeCycles} relates Split Cycle to the notion of \textit{immunity to binary arguments} in \citealt{Heitzig2002}. In particular, Split Cycle (along with Beat Path and Ranked Pairs) satisfies all of Heitzig's axioms $(\mathrm{Im}_{M_\alpha})$ for $1/2<\alpha\leq 1$. Although when defining choice rules, Heitzig \citeyearpar[Lemma 2 and following]{Heitzig2002} only defines rules based on his notion of \textit{strong} immunity to binary arguments,\footnote{Compare Heitzig's notion of strong immunity to Schwartz's \citeyearpar{Schwartz1986} characterization of GOCHA in Lemma \ref{GOCHALem} below.} which includes Beat Path (in his notation, the rule that selects the common optimal elements of the chain $\{\mathrm{tr}_S(M_\alpha)\mid \frac{1}{2}<\alpha\leq 1\}$), not Split Cycle, it is natural in that setting to consider the Split Cycle rule formulated as in Lemma \ref{OnlySomeCycles} as well. Heitzig's axioms $(\mathrm{Im}_{M_\alpha})$ are also closely related to the notion of a \textit{stack} from \citealt{ZavistTideman1989}, defined in Section \ref{Comparison}.\end{remark}

It will facilitate reasoning about the defeat relation to introduce one more convenient piece of notation.

\begin{definition} Let $\mathbf{P}$ be a profile and $a,b\in X(\mathbf{P})$. The \textit{cycle number of $a$ and $b$ in $\mathbf{P}$} is \[Cycle\#_\mathbf{P}(a,b)=\mathrm{max}(\{0\}\cup \{Split\#(\rho) \mid \rho\mbox{ a simple cycle of the form } a \rightarrow b\rightarrow x_1\rightarrow \dots\rightarrow x_n\rightarrow a\}) .\]
\end{definition}

Then we can equivalently rewrite the definition of the defeat relation as follows.

\begin{lemma}\label{CycleNumLem} Let $\mathbf{P}$ be a profile and $a,b\in X(\mathbf{P})$. Then \textit{$a$ defeats $b$ in $\mathbf{P}$} if and only if \[Margin_\mathbf{P}(a,b)>Cycle\#_\mathbf{P}(a,b).\]
\end{lemma}
\noindent We will often apply Lemmas \ref{OnlySomeCycles} and \ref{CycleNumLem} in proofs without comment.

\subsection{Refinements of Split Cycle}\label{Comparison}

\citealt{HP2021} argues that the Split Cycle defeat relation from Definition \ref{DefeatDef} provides the right notion of one candidate defeating another in a democratic election using ranked ballots. Let us say that a voting method $F$ is regarded as a \textit{pre-tiebreaking voting method} if one regards $F(\mathbf{P})$ as the set of undefeated candidates and regards any further narrowing of $F(\mathbf{P})$ as ``tiebreaking.'' The political significance of this distinction is that if $F(\mathbf{P})$ contains a single winner, then that winner may be viewed as having a stronger mandate from voters, as a result of a more unambiguous election, than a candidate who is among several undefeated candidates in $F(\mathbf{P})$ but wins by some further tiebreaking process. In this paper, we are proposing Split Cycle as a pre-tiebreaking voting method.

Since there can be multiple undefeated candidates, the question arises of how to pick an ultimate winner from among the undefeated. In addition to the usual non-anonymous, or non-neutral, or non-deterministic tiebreaking procedures (e.g., let the Chair decide among the undefeated, or use seniority to decide among the undefeated, or randomly choose an undefeated candidate), one can apply an anonymous, neutral, and deterministic tiebreaker before resorting to tiebreakers that violates one of these properties. Indeed, one can view voting methods that refine Split Cycle as deterministic tiebreakers. On this approach, an election result consists in an announcement of undefeated candidates according to Split Cycle and, in the event of multiple undefeated candidates, the announcement of a tiebreak winner. This is precisely how Split Cycle is used on the election website \href{https://stablevoting.org}{stablevoting.org}, where the tiebreaking procedure is the recently proposed Stable Voting method (\citealt{HP2022}).

Other refinements of Split Cycle are the well-known Beat Path (\citealt{Schulze2011,Schulze2022}) and Ranked Pairs (\citealt{Tideman1987,ZavistTideman1989}) voting methods, as well as variants of Ranked Pairs such as the River method (\citealt{Heitzig2004b}). These methods may pick different candidates from among the undefeated candidates according to Split Cycle, as shown in Example \ref{TwoCondorcetian} below in the case of Ranked Pairs and Beat Path. All of these methods, including Stable Voting, satisfy the following property, which implies that non-anonymous, non-neutral, or non-deterministic tiebreaking is only needed in non-uniquely weighted profiles.

\begin{definition} A voting method $F$ is \textit{quasi-resolute} if for every uniquely-weighted $\mathbf{P}\in\mathrm{dom}(F)$,  ${\vert F(\mathbf{P})\vert  =1}$.
\end{definition}
\noindent That Split Cycle is not quasi-resolute is shown by Example \ref{TwoCondorcetian} in the next section.

 According to Beat Path, $a$ wins in $\mathbf{P}$ if for every other candidate $b$, the \textit{strongest path} from $a$ to $b$ in $\mathcal{M}(\mathbf{P})$ is at least as strong as the strongest path from $b$ to $a$ in $\mathcal{M}(\mathbf{P})$, where the strength of a path is the smallest margin between consecutive candidates in the path. We can relate Split Cycle to  Beat Path using the following lemma.\footnote{\label{FloydWarshall}Lemma 3.17 immediately suggests an efficient algorithm for computing the Split Cycle defeat relation: for all candidates $a$ and $b$, if the margin of $a$ over $b$ is positive, check if there is a path from $b$ back to $a$ of strength at least the margin of $a$ over $b$; if not, $a$ defeats $b$. Lemma 3.17 also shows we can compute Split Cycle using a modification of the Floyd-Warshall algorithm used by Schulze  \citeyearpar{Schulze2011} to compute Beat Path. See the Python implementations at https://github.com/epacuit/splitcycle.}

\begin{lemma}\label{PathStrengthLem} Let $\mathbf{P}$ be a profile and $a,b\in X(\mathbf{P})$. Then $a$ defeats $b$ in $\mathbf{P}$ if and only if $Margin_\mathbf{P}(a,b)>0$ and $Margin_\mathbf{P}(a,b)>\mbox{the strength of the strongest path from $b$ to $a$}$.\end{lemma}

\begin{proof} By Lemma \ref{OnlySomeCycles}, $a$ defeats $b$ if and only if $Margin_\mathbf{P}(a,b)$ is greater than 0 and the splitting number of every simple cycle of the form $a \rightarrow b\rightarrow x_1\rightarrow \dots\rightarrow x_n\rightarrow a$. But this is equivalent to  $Margin_\mathbf{P}(a,b)$ being greater than 0 and the strength of every path of the form $b\rightarrow x_1\rightarrow \dots\rightarrow x_n\rightarrow a$, which is equivalent to $Margin_\mathbf{P}(a,b)$ being greater than 0 and the strength of the strongest path from $b$ to $a$.\end{proof}

\begin{lemma}\label{SubsetLem} For any profile $\mathbf{P}$, $BP(\mathbf{P})\subseteq SC(\mathbf{P})$, where $BP$ is the Beat Path method.
\end{lemma}
\begin{proof} Suppose $a\not\in SC(\mathbf{P})$, so there is a $b\in X(\mathbf{P})$ such that $b$ defeats $a$ according to Split Cycle. Hence $Margin_\mathbf{P}(b,a)$ is greater than the strength of the strongest path from $a$ to $b$ by Lemma \ref{PathStrengthLem}. Since $b\to a$ is a path from $b$ to $a$, it follows that the strength of the strongest path  from $b$ to $a$ is greater than the strength of the strongest path from $a$ to $b$. Hence $a\not\in BP(\mathbf{P})$.\end{proof}

We can prove an analogous lemma for Ranked Pairs. To compute the Ranked Pairs winners, given a linear order $T$ of the edges of $\mathcal{M}(\mathbf{P})$, order the edges in $\mathcal{M}(\mathbf{P})$ from largest to smallest margin, breaking ties according to $T$. Considering each edge in turn, ``lock in'' the edge if adding the edge to the list of already locked-in edges does not create a cycle of locked-in edges. Then $a\in RP(\mathbf{P})$ if there is some $T$ such that after running the above algorithm with $T$, there is no locked-in edge pointing to $a$. We also make use of an alternative characterization of Ranked Pairs due to Zavist and Tideman \citeyearpar{ZavistTideman1989}. Given a profile $\mathbf{P}$, a linear order $L$ on $X(\mathbf{P})$ is a \textit{stack for $\mathbf{P}$} if for any $a,b\in X(\mathbf{P})$, if $aLb$, then there are distinct $x_1,\dots,x_n\in X(\mathbf{P})$ with $x_1=a$ and $x_n=b$ such that $x_i L x_{i+1}$ and $Margin_\mathbf{P}(x_i, x_{i+1})\geq Margin_\mathbf{P}(b,a)$ for all $i\in \{1,\dots, n-1\}$.

\begin{lemma}[\citealt{ZavistTideman1989}]\label{RPlem2} For any profile $\mathbf{P}$ and $a\in X(\mathbf{P})$, we have $a\in RP(\mathbf{P})$ if and only if $a$ is the maximum element in some stack for $\mathbf{P}$.\footnote{Using this lemma, we can also prove a strengthened version of Lemma \ref{SubsetLem2}: for any profile $\mathbf{P}$ and $a,b\in X(\mathbf{P})$, if $b$ defeats $a$ according to Split Cycle (Definition \ref{DefeatDef}), then $bLa$ for any stack $L$ for $\mathbf{P}$.}
\end{lemma}

\begin{lemma}\label{SubsetLem2} For any profile $\mathbf{P}$, $RP(\mathbf{P})\subseteq SC(\mathbf{P})$, where $RP$ is the Ranked Pairs method.
\end{lemma}
\begin{proof} Suppose $a\not\in SC(\mathbf{P})$, so there is some $b\in X(\mathbf{P})$ such that $Margin_\mathbf{P}(b,a)>0$ and $Margin_\mathbf{P}(b,a)>Split\#(\rho)$ for every simple cycle $\rho$ containing $b$ and $a$. Now suppose for contradiction that $a\in RP(\mathbf{P})$. Then by Lemma \ref{RPlem2}, there are distinct $y_1,\dots, y_m \in X(\mathbf{P})$ such that $a\overset{\alpha_0}{\longrightarrow} y_1\to\dots\to y_m\overset{\alpha_m}{\longrightarrow} b$ with $\alpha_i\geq Margin_\mathbf{P}(b,a)$ for each $i\in \{0,\dots,m\}$. But then $\rho:= b\to a\overset{\alpha_0}{\longrightarrow} y_1\to\dots\to y_m\overset{\alpha_m}{\longrightarrow} b$ is a simple cycle such that $Margin_\mathbf{P}(b,a)\not > Split\#(\rho)$, which is a contradiction. Hence $a\not\in RP(\mathbf{P})$.\end{proof}

The Stable Voting method is a refinement of Split Cycle by definition: to find the Stable Voting winner in $\mathbf{P}$, order the pairs $(a,b)$ of candidates such that $a$ is undefeated in $\mathbf{P}$ from largest to smallest value of $Margin_\mathbf{P}(a,b)$, and declare as Stable Voting winners the candidate(s) $a$ from the earliest pair(s) $(a,b)$ such that $a$ is a Stable Voting winner in $\mathbf{P}_{-b}$. Thus, Stable Voting is defined \textit{recursively}, where in the case of a profile with a single candidate, that candidate is the Stable Voting winner. Remarkably, the Simple Stable Voting procedure that is defined just like Stable Voting but without the requirement that $a$ is undefeated appears to always select from the undefeated candidates anyway in profiles that are uniquely weighted.\footnote{For non-uniquely weighted profiles, there are extremely rare examples in which $SSV(\mathbf{P})\not\subseteq SC(\mathbf{P})$. However, if one defines Simple Stable Voting to use parallel-universe tiebreaking of tied margins in the style of Ranked Pairs (see \citealt{Wangetal2019}), then Conjecture \ref{SVConjecture} implies that $SSV_{PUT}(\mathbf{P})\subseteq SC(\mathbf{P})$ and $SSV_{PUT}(\mathbf{P})= SV_{PUT}(\mathbf{P})$ for all profiles $\mathbf{P}$.}

\begin{conjecture}\label{SVConjecture} For any uniquely-weighted profile $\mathbf{P}$, $SSV(\mathbf{P})\subseteq SC(\mathbf{P})$ and $SSV(\mathbf{P})= SV(\mathbf{P})$, where $SSV$ and $SV$ are the Simple Stable Voting and Stable Voting methods, respectively. 
\end{conjecture}

The cost of deterministic tiebreaking is the violation of variable-candidate and variable-voter axioms from Sections \ref{SpoilerSection} and \ref{NoShowSection} that are satisfied by Split Cycle. Beat Path and Ranked Pairs violate even weaker axioms of \textit{partial} immunity to spoilers and stealers, defined in Section \ref{NewCriteria}, while Beat Path, Ranked Pairs, and Stable Voting all violate positive involvement. Of the known quasi-resolute refinements of Split Cycle, we prefer Stable Voting on the grounds that it satisfies an axiom of \textit{stability for winners with tiebreaking}, also defined in Section \ref{NewCriteria}, which implies partial immunity to spoilers and stealers.

\section{Spoilers, Stealers, and Stability}\label{NewCriteria}

In Section \ref{SpoilerSection}, we discussed the problem of spoiler effects in elections with more than two candidates. As noted, the \textit{independence of clones} criterion (\citealt{Tideman1987}) is often cited as an anti-spoiler axiom. In Section \ref{ClonesSection}, we show that Split Cycle satisfies this axiom. However, independence of clones only rules out a special type of spoiler effect, namely, vote splitting by the introduction of a similar candidate. But sometimes a candidate $b$ can spoil the election for a dissimilar candidate $a$, as in Example \ref{IRVExample}, where the Republican spoiled the election for the Democrat; and in real cases of vote splitting, such as Example \ref{BushNaderGore}, the ``similar candidate'' will almost never qualify as a clone in the formal sense (see Definition \ref{CloneDef}). 

To capture perverse effects of the introduction of a candidate not covered by independence of clones, in Section \ref{SpoilerSection} we suggested the concepts of \textit{spoilers} and \textit{stealers}, defined formally as follows.

\begin{definition} Let $F$ be a voting method, $\mathbf{P}\in\mathrm{dom}(F)$, and $a,b\in X(\mathbf{P})$. Then we say that:
\begin{enumerate}
\item \textit{$b$ spoils the election for $a$ in $\mathbf{P}$} if $a\in F(\mathbf{P}_{-b})$, $Margin_\mathbf{P}(a,b)>0$, $a\not\in F(\mathbf{P})$, and $b\not\in F(\mathbf{P})$; 
\item $b$ \textit{steals the election from $a$ in $\mathbf{P}$} if $a\in F(\mathbf{P}_{-b})$, $Margin_\mathbf{P}(a,b)>0$, $a\not\in F(\mathbf{P})$, and $b\in F(\mathbf{P})$.
\end{enumerate}
\end{definition}

Recall the three axioms from Section \ref{SpoilerSection}, now defined formally as well.

\begin{definition}\label{SpoilStealStable} Let $F$ be a voting method.
\begin{enumerate}
\item $F$ satisfies \textit{immunity to spoilers} if for  $\mathbf{P}\in\mathrm{dom}(F)$ and $a,b\in X(\mathbf{P})$, $b$ does not spoil the election for $a$.
\item $F$ satisfies \textit{immunity to stealers} if for  $\mathbf{P}\in\mathrm{dom}(F)$ and $a,b\in X(\mathbf{P})$, $b$ does not steal the election from~$a$.
\item $F$ satisfies \textit{stability for winners} if for  $\mathbf{P}\in\mathrm{dom}(F)$ and $a,b\in X(\mathbf{P})$, if   $a\in F(\mathbf{P}_{-b})$ and ${Margin_\mathbf{P}(a,b)>0}$, then $a\in F(\mathbf{P})$.
\end{enumerate}
\end{definition}

The following is immediate from the definition.

\begin{fact} Stability for winners is equivalent to the conjunction of immunity to spoilers and stealers.
\end{fact}

It is useful to have terminology for a candidate $a$  who would win without another candidate $b$ in the election, such that a majority of voters prefer $a$ to $b$.

\begin{definition} Given a voting method $F$, profile $\mathbf{P}$, and $a\in X(\mathbf{P})$, we say that \textit{$a$ is Condorcetian for $F$ in $\mathbf{P}$} if there is some $b\in X(\mathbf{P})$ such that  $a\in F(\mathbf{P}_{-b})$ and $Margin_\mathbf{P}(a,b)>0$.\footnote{In \citealt{HP2022}, such an $a$ is called \textit{stable for $F$ in $\mathbf{P}$}.}

That $a$ is \textit{weakly Condorcetian for $F$ in $\mathbf{P}$} is defined in the same way but with $Margin_\mathbf{P}(a,b)\geq 0$.
\end{definition}
Recall that $a$ is a Condorcet winner if $a$ wins against every other candidate head-to-head. In a similar spirit, a candidate who is Condorcetian wins against $X(\mathbf{P}_{-b})$ according to $F$ and wins against $b$ head-to-head. In fact, the notions of a Condorcet winner and a Condorcetian candidate are related as follows. Recall that $F$ is Condorcet consistent if $F(\mathbf{P})=\{a\}$ whenever $a$ is the Condorcet winner in $\mathbf{P}\in \mathrm{dom}(F)$.

\begin{lemma}\label{CondorcetianLem} Let $F$ be a voting method. Then (i) if $F$ is Condorcet consistent, $\mathbf{P}\in\mathrm{dom}(F)$  with $\vert X(\mathbf{P})\vert >1$, and $c$ is a Condorcet winner in $\mathbf{P}$, then $c$ is the unique Condorcetian candidate for $F$ in $\mathbf{P}$; and (ii) if for any $\mathbf{P}\in\mathrm{dom}(F)$ with a unique Condorcetian candidate $c$, $F(\mathbf{P})=\{c\}$, then $F$ is Condorcet consistent.
\end{lemma}
\begin{proof} For part (i), let $\mathbf{P}$ and $c$ be as in the statement. Then for any $a,b\in X(\mathbf{P})$, if $Margin_\mathbf{P}(a,b)>0$, then $b\neq c$, so $c\in X(\mathbf{P}_{-b})$ and hence $c$ is the Condorcet winner in $\mathbf{P}_{-b}$, so $F(\mathbf{P}_{-b})=\{c\}$ by the assumption that $F$ is Condorcet consistent. It follows that $c$ is the unique Condorcetian candidate.

For part (ii), suppose $c$ is a Condorcet winner in $\mathbf{P}$. If $\vert X(\mathbf{P})\vert  =1$, then $F(\mathbf{P})=\{c\}$. If $\vert X(\mathbf{P})\vert  >1$, then by part (i), $c$ is the unique Condorcetian candidate in $\mathbf{P}$ and hence by the assumption on $F$, $F(\mathbf{P})=\{c\}$. Thus, $F$ is Condorcet consistent.\end{proof}

Note that in a given profile, there may be no Condorcetian candidates. For example, let $F$ be a voting method that applies majority rule in two-candidate profiles. Then in an election with three candidates in a majority cycle, no candidate is Condorcetian for $F$; for if $a$ is majority preferred to $b$, then removing $b$ from the election results in a two-candidate profile in which $a$ loses. On the other hand, there may be more than one Condorcetian candidate, as in the following example.

\begin{example}\label{TwoCondorcetian} Assume $F$ is a voting method such that in any three-candidate profile with a majority cycle, if there is a candidate who has both the uniquely largest majority victory and the uniquely smallest majority loss, that candidate is among the winners. Then in a four-candidate profile with the following margin graph, both candidates $a$ and $c$ are Condorcetian for $F$ in $\mathbf{P}$ (for $a$, remove $c$, and for $c$, remove $b$):

\begin{center}
\begin{minipage}{2in}\begin{tikzpicture}
\node[circle,draw,minimum width=0.25in] at (0,0)      (a) {$b$}; 
\node[circle,draw,minimum width=0.25in] at (3,0)      (b) {$d$}; 
\node[circle,draw,fill=blue!25,minimum width=0.25in] at (1.5,1.5)  (c) {$c$}; 
\node[circle,draw,fill=blue!25,minimum width=0.25in] at (1.5,-1.5) (d) {$a$};
\path[->,draw,thick] (a) to[pos=.7] node[fill=white] {$10$} (b);
\path[->,draw,thick] (c) to node[fill=white] {$2$} (a);
\path[->,draw,thick] (d) to node[fill=white] {$12$} (a);
\path[->,draw,thick] (c) to node[fill=white] {$8$} (b);
\path[->,draw,thick] (b) to node[fill=white] {$6$} (d);
\path[->,draw,thick] (d) to[pos=.7]  node[fill=white] {$4$} (c);
\end{tikzpicture}
\end{minipage}
\end{center}
Both $a$ and $c$ are undefeated according to Split Cycle; Ranked Pairs picks only $a$; and Beat Path picks only~$c$.\end{example}

While what we called a pre-tiebreaking method in Section \ref{Comparison} can---and we think should---satisfy stability for winners, Example \ref{TwoCondorcetian} suggests that a voting method that incorporates tiebreaking cannot. Indeed, we prove the following impossibility theorem in other work.

\begin{theorem}[\citealt{HPZ2022}] There is no voting method (whose domain contains all profiles with up to four candidates) satisfying anonymity, neutrality, stability for winners, and quasi-resoluteness.
\end{theorem}
\noindent We will prove a related impossibility theorem in Section \ref{Resolvability}. 

While tiebreaking is therefore inconsistent with selecting \textit{all} the Condorcetian candidates, it is compatible with selecting from \textit{among} the Condorcetian candidates. These observations lead us to the following modified axioms applicable to tiebreaking procedures.

\begin{definition}\label{PartialAxioms} Let $F$ be a voting method.
\begin{enumerate}
 \item $F$ satisfies \textit{partial immunity to spoilers} if for all $\mathbf{P}\in\mathrm{dom}(F)$ and $a,b\in X(\mathbf{P})$, if $a$ is the unique Condorcetian candidate in $\mathbf{P}$, then $b$ does not spoil the election for $a$.
\item $F$ satisfies \textit{partial immunity to stealers} if for all $\mathbf{P}\in\mathrm{dom}(F)$ and $a,b\in X(\mathbf{P})$, if $a$ is the unique Condorcetian candidate in $\mathbf{P}$, then $b$ does not steal the election from~$a$.
\item $F$ satisfies \textit{partial stability for winners} if for all $\mathbf{P}\in\mathrm{dom}(F)$  and $a\in X(\mathbf{P})$, if $a$ is the unique Condorcetian candidate in $\mathbf{P}$, then $a\in F(\mathbf{P})$.
\item\label{PartialAxioms4} A voting method $F$ satisfies \textit{stability for winners with tiebreaking} if for all $\mathbf{P}\in\mathrm{dom}(F)$ , if some candidate is Condorcetian for $F$ in $\mathbf{P}$, then all candidates in $F(\mathbf{P})$ are Condorcetian for $F$ in $\mathbf{P}$.
\end{enumerate}
\end{definition}

The logical relations between these axioms are recorded in the following.

\begin{fact}\label{AxRels} (i) Partial stability for winners is equivalent to the conjunction of partial immunity to spoilers and partial immunity to stealers. (ii) Stability for winners with tiebreaking implies partial stability for winners. (iii) Stability for winners with tiebreaking and stability for winners are incomparable in strength. 
\end{fact}

Parts (i) and (ii) are immediate from the definitions, while part (iii) follows from facts proved in Section~\ref{StabilitySection}: Split Cycle satisfies stability for winners but not stability for winners with tiebreaking, whereas Stable Voting satisfies stability for winners with tiebreaking but not stability for winners. But again, stability for winners and stability for winners with tiebreaking are intended as axioms on different kinds of functions: we regard stability for winners as an appropriate axiom for \textit{pre-tiebreaking} voting methods, whereas stability for winners with tiebreaking is an appropriate axiom for voting methods that are conceived as incorporating tiebreaking: if there are Condorcetian candidates, then the ultimate tiebreak winner must be one of them.

\subsection{Spoilers}\label{SpoilersSection}

It is easy to check that GETCHA, GOCHA, Uncovered Set, and Minimax satisfy immunity to spoilers. However, Beat Path and Ranked Pairs do not.

\begin{proposition}\label{BPspoiler} Beat Path does not satisfy even partial immunity to spoilers. 

In fact, there are uniquely-weighted linear profiles $\mathbf{P}$ and distinct $a,b,c\in X(\mathbf{P})$ such that with respect to Beat Path: $b$ spoils the election for $a$ in $\mathbf{P}$; $a$ is the unique Condorcetian candidate;  the largest margin in $\mathbf{P}$ is $Margin_\mathbf{P}(a,b)$; $BP(\mathbf{P})=\{c\}$, $c$ is the Condorcet loser in $\mathbf{P}_{-b}$, and $Margin_\mathbf{P}(a,c)>0$.\end{proposition}
\begin{proof} Let $\mathbf{P}$ be a linear profile whose margin graph is displayed on the right below, so the margin graph of $\mathbf{P}_{-b}$ is displayed on the left below:
\begin{center}
\begin{minipage}{2.5in}
\begin{tikzpicture}
\node[circle,draw,minimum width=0.25in] at (2,2)  (a) {$d$}; 
\node[circle,draw,minimum width=0.25in] at (-1,2)  (b) {$e$}; 
\node[circle,draw,minimum width=0.25in,fill=red!50] at (-1,-2) (c) {$c$}; 
\node[circle,draw,fill=medgreen!50,minimum width=0.25in] at (2,-2) (d) {$a$}; 
\path[->,draw,thick] (a) to node[fill=white] {$16$} (b);
\path[->,draw,thick] (d) to node[fill=white] {$18$} (a);
\path[->,draw,thick] (b) to node[fill=white] {$6$} (c);
\path[->,draw,very thick] (d) to node[fill=white] {$8$} (c);
\path[->,draw,thick] (a) to[pos=.7] node[fill=white] {$4$} (c);
\path[->,draw,thick] (b) to[pos=.7] node[fill=white] {$10$} (d);
\end{tikzpicture}
\end{minipage}
\begin{minipage}{2.5in}
\begin{tikzpicture}
\node[circle,draw,minimum width=0.25in] at (2,2)  (a) {$d$}; 
\node[circle,draw,minimum width=0.25in] at (-1,2)  (b) {$e$}; 
\node[circle,draw,fill=medgreen!50,minimum width=0.25in] at (-1,-2) (c) {$c$}; 
\node[circle,draw,minimum width=0.25in] at (2,-2) (d) {$a$}; 
\node[circle,draw,minimum width=0.25in] at (4,0)  (e) {$b$};
\path[->,draw,thick] (a) to node[fill=white] {$16$} (b);
\path[->,draw,thick] (d) to node[fill=white] {$18$} (a);
\path[->,draw,thick] (e) to node[fill=white] {$14$} (a);
\path[->,draw,thick] (b) to node[fill=white] {$6$} (c);
\path[->,draw,very thick] (d) to node[fill=white] {$8$} (c);
\path[->,draw,thick] (d) to node[fill=gray!25] {$\mathbf{20}$} (e);
\path[->,draw,thick] (a) to[pos=.7] node[fill=white] {$4$} (c);
\path[->,draw,thick] (b) to[pos=.7] node[fill=white] {$10$} (d);
\path[->,draw,thick] (b) to[pos=.7] node[fill=white] {$2$} (e);
\path[->,draw,thick] (c) to[pos=.7] node[fill=white] {$12$} (e);
\end{tikzpicture}
\end{minipage}
\end{center}
First, it is easy to see that $BP(\mathbf{P}_{-b})=\{a\}$. Then since $Margin_\mathbf{P}(a,b)>0$,  $a$ is Condorcetian for Beat Path in $\mathbf{P}$.  One can also check that $BP(\mathbf{P})=\{c\}$. Therefore, $b$ is a spoiler for $a$ in $\mathbf{P}$. Moreover, no other candidate is Condorcetian for Beat Path in $\mathbf{P}$, since $b\not\in BP(\mathbf{P}_{-d})$, $c\not\in BP(\mathbf{P}_{-b})$, $d\not\in BP(\mathbf{P}_{-c})$, $d\not\in BP(\mathbf{P}_{-e})$, $e\not\in BP(\mathbf{P}_{-a})$, $e\not\in BP(\mathbf{P}_{-b})$, and $e\not\in BP(\mathbf{P}_{-c})$. Finally, observe that the largest margin in $\mathbf{P}$ is $Margin_\mathbf{P}(a,b)$, that in $\mathbf{P}_{-b}$, $c$ is the Condorcet loser, and that $Margin_\mathbf{P}(a,c)>0$.\end{proof}
\begin{remark} For $\mathbf{P}_{-b}$ as in the proof of Proposition \ref{BPspoiler}, Split Cycle also picks $\{a\}$, so Ranked Pairs and Stable Voting do as well, while Minimax picks the Condorcet loser $\{c\}$. In $\mathbf{P}$, Split Cycle picks $\{a,c\}$, Ranked Pairs and Stable Voting pick $\{a\}$, and Minimax picks $\{c\}$.\end{remark}

\begin{proposition}\label{RPspoiler} Ranked Pairs does not satisfy even partial immunity to spoilers. 

In fact, there are uniquely-weighted linear profiles $\mathbf{P}$ and distinct $a,b,c\in X(\mathbf{P})$ such that with respect to Ranked Pairs: $b$ spoils the election for $a$ in $\mathbf{P}$; $a$ is the unique Condorcetian candidate in $\mathbf{P}$; the largest margin in $\mathbf{P}$ is $Margin_\mathbf{P}(a,b)$;
$RP(\mathbf{P})=\{c\}$ and $Margin_\mathbf{P}(a,c)>0$.

\end{proposition}

\begin{proof} Let $\mathbf{P}$ be a profile  whose margin graph is displayed on the right below, so the margin graph of $\mathbf{P}_{-b}$ is displayed on the left below:
\begin{center} 
\begin{minipage}{2.5in}
\begin{tikzpicture}
\node[circle,draw,fill = medgreen!50, minimum width=0.25in] at (2,2)  (a) {$a$}; 
\node[circle,draw,minimum width=0.25in] at (-1,2)  (b) {$d$}; 
\node[circle,draw,minimum width=0.25in] at (-1,-2) (c) {$e$}; 
\node[circle,draw,minimum width=0.25in] at (2,-2) (d) {$c$}; 
\path[->,draw,thick] (b) to node[fill=white] {$8$} (a);
\path[->,draw,very thick] (a) to node[fill=white] {$6$} (d);
\path[->,draw,thick] (c) to node[fill=white] {$4$} (b);
\path[->,draw,thick] (c) to node[fill=white] {$12$} (d);
\path[->,draw,thick] (a) to[pos=.7] node[fill=white] {$18$} (c);
\path[->,draw,thick] (d) to[pos=.7] node[fill=white] {$10$} (b);
\end{tikzpicture}
\end{minipage}\begin{minipage}{2.5in}
\begin{tikzpicture}
\node[circle,draw,minimum width=0.25in] at (2,2)  (a) {$a$}; 
\node[circle,draw,minimum width=0.25in] at (-1,2)  (b) {$d$}; 
\node[circle,draw,minimum width=0.25in] at (-1,-2) (c) {$e$}; 
\node[circle,draw,fill = medgreen!50,minimum width=0.25in] at (2,-2) (d) {$c$}; 
\node[circle,draw,minimum width=0.25in] at (4,0)  (e) {$b$};
\path[->,draw,thick] (b) to node[fill=white] {$8$} (a);
\path[->,draw,very thick] (a) to node[fill=white] {$6$} (d);
\path[->,draw,thick] (a) to node[fill=gray!25] {$\mathbf{20}$} (e);
\path[->,draw,thick] (c) to node[fill=white] {$4$} (b);
\path[->,draw,thick] (c) to node[fill=white] {$12$} (d);
\path[->,draw,thick] (d) to node[fill=white] {$14$} (e);
\path[->,draw,thick] (a) to[pos=.7] node[fill=white] {$18$} (c);
\path[->,draw,thick] (d) to[pos=.7] node[fill=white] {$10$} (b);
\path[->,draw,thick] (b) to[pos=.7] node[fill=white] {$2$} (e);
\path[->,draw,thick] (e) to[pos=.7] node[fill=white] {$16$} (c);
\end{tikzpicture}
\end{minipage}
\end{center}
First, we claim that $RP(\mathbf{P}_{-b})=\{a\}$. We lock in edges in the following order: $(a,e)$, $(e,c)$, $(c,d)$; then ignore $(d,a)$, since locking it in would create a cycle ($aecda$); then $(a,c)$, $(e,d)$. The only candidate with no incoming edge locked in is $a$, so indeed $RP(\mathbf{P}_{-b})=\{a\}$. Then since $Margin_\mathbf{P}(a,b)>0$, $a$ is Condorcetian for Ranked Pairs in $\mathbf{P}$.  Next, we claim that $RP(\mathbf{P})=\{c\}$. We lock in edges in the following order: $(a,b)$, $(a,e)$, $(b,e)$, $(c,b)$; then ignore $(e,c)$, since locking it in would create a cycle ($ecbe$); then $(d,c)$, $(d,a)$; then ignore $(a,c)$, since locking it in would create a cycle ($acda$); also ignore $(e,d)$, since locking it in would create a cycle ($daed$); finally, $(d,b)$. The only candidate with no incoming edge locked in is $c$, so indeed $RP(\mathbf{P})=\{c\}$. Therefore, $b$ is a spoiler for $a$ in $\mathbf{P}$. Moreover, no other candidate is Condorcetian for Ranked Pairs in $\mathbf{P}$, since $b\not\in RP(\mathbf{P}_{-e})$, $c\not\in RP(\mathbf{P}_{-b})$, $c\not\in RP(\mathbf{P}_{-d})$, $d\not\in RP(\mathbf{P}_{-a})$,  $d\not\in RP(\mathbf{P}_{-b})$, $e\not\in RP(\mathbf{P}_{-c})$, and $e\not\in RP(\mathbf{P}_{-d})$.  Finally, observe that the largest margin in $\mathbf{P}$ is $Margin_\mathbf{P}(a,b)$ and that $Margin_\mathbf{P}(a,c)>0$.\end{proof}

\begin{remark} For $\mathbf{P}_{-b}$  in the proof of Proposition \ref{RPspoiler}, Split Cycle, Beat Path, Stable Voting, and Minimax all pick $\{a\}$. In $\mathbf{P}$, Split Cycle picks $\{a,c\}$, while Beat Path, Stable Voting, and Minimax all pick~$\{a\}$.\end{remark}

\subsection{Stealers}\label{StealersSection}

As observed in \S~\ref{SpoilerSection}, when $a$ would win without $b$ in the election,  and more voters prefer $a$ to $b$ than prefer $b$ to $a$, yet $a$ loses when $b$ joins, it hardly improves the situation to find that $b$ wins; in this case, although $b$ does not \textit{spoil} the election for $a$---since spoilers are by definition losers---we say that $b$ \textit{steals} the election from $a$. Recall the formal definition of immunity to stealers from Definition \ref{SpoilStealStable}.

\begin{proposition}\label{MinimaxBPstealers} Minimax and Beat Path do not satisfy even partial immunity to stealers. 

In fact, there are profiles $\mathbf{P}$ such that with respect to Minimax: $b$ steals the election from $a$; $a$ is the unique Condorcetian candidate; and $b$ is the Condorcet loser in $\mathbf{P}$. The same holds for Beat Path but without the condition that $b$ is a Condorcet loser.
\end{proposition}
\begin{proof} First, for Minimax, let $\mathbf{P}$ be a profile whose margin graph is displayed in the middle below, so the margin graph of $\mathbf{P}_{-b}$ is displayed on the left:
 \begin{center}
\begin{minipage}{2in}\begin{tikzpicture}

\node[circle,draw, minimum width=0.25in] at (0,0) (b) {$c$}; 
\node[circle,fill = medgreen!50,draw,minimum width=0.25in] at (3,0) (a) {$a$}; 
\node[circle,draw,minimum width=0.25in] at (1.5,1.5) (c) {$d$};  

\path[->,draw,thick] (b) to node[fill=white] {$3$} (a);
\path[->,draw,thick] (c) to node[fill=white] {$5$} (b);
\path[->,draw,thick] (a) to node[fill=white] {$7$} (c);
  \end{tikzpicture}
\end{minipage}\begin{minipage}{2in}\begin{tikzpicture}

\node[circle,draw, minimum width=0.25in] at (0,0) (b) {$c$}; 
\node[circle,draw,minimum width=0.25in] at (3,0) (a) {$a$}; 
\node[circle,draw,minimum width=0.25in] at (1.5,1.5) (c) {$d$}; 
\node[circle,fill = medgreen!50,draw,minimum width=0.25in] at (1.5,-1.5) (d) {$b$}; 

\path[->,draw,thick] (b) to (a);
\path[<-,draw,thick] (d) to (c);
\path[->,draw,thick] (c) to node[fill=white] {$5$} (b);
\path[->,draw,thick] (a) to node[fill=white] {$7$} (c);
\path[->,draw,very thick] (a) to node[fill=white] {$1$} (d);
\path[->,draw,thick] (b) to node[fill=white] {$1$} (d);

\node[fill=white] at (1.5,.5)  {$1$}; 
\node[fill=white] at (2,0)  {$3$}; 

  \end{tikzpicture}
\end{minipage}\begin{minipage}{2in}\begin{tikzpicture}

\node[circle,draw, minimum width=0.25in] at (0,0) (b) {$c$}; 
\node[circle,draw,minimum width=0.25in] at (3,0) (a) {$a$}; 
\node[circle,draw,minimum width=0.25in] at (1.5,1.5) (c) {$d$}; 
\node[circle,fill = medgreen!50,draw,minimum width=0.25in] at (1.5,-1.5) (d) {$b$}; 

\path[->,draw,thick] (b) to (a);
\path[->,draw,thick] (d) to (c);
\path[->,draw,thick] (c) to node[fill=white] {$5$} (b);
\path[->,draw,thick] (a) to node[fill=white] {$7$} (c);
\path[->,draw,very thick] (a) to node[fill=white] {$1$} (d);
\path[->,draw,thick] (b) to node[fill=white] {$1$} (d);

\node[fill=white] at (1.5,.5)  {$3$}; 
\node[fill=white] at (2,0)  {$3$}; 

  \end{tikzpicture}
\end{minipage}\end{center}
Observe that $Minimax(\mathbf{P}_{-b})=\{a\}$ and $Margin_\mathbf{P}(a,b)>0$. Yet $Minimax(\mathbf{P})=\{b\}$, so Minimax violates immunity to stealers. Moreover, $a$ is the only candidate who is Condorcetian for Minimax in $\mathbf{P}$, and the stealer $b$ is the Condorcet loser in $\mathbf{P}$.

For Beat Path, let $\mathbf{P}$ be a profile whose margin graph is displayed on the right above, so again the margin graph of $\mathbf{P}_{-b}$ is displayed on the left. Observe that $BP(\mathbf{P}_{-b})=\{a\}$, that $Margin_\mathbf{P}(a,b)>0$, and yet $BP(\mathbf{P})=\{b\}$, so Beat Path violates immunity to stealers. Moreover, $a$ is the only candidate who is Condorcetian for Beat Path in $\mathbf{P}$.\end{proof}

Immunity to stealers reveals an interesting axiomatic difference between Beat Path and Ranked Pairs.

\begin{proposition}\label{RPStealer} Ranked Pairs satisfies immunity to stealers on uniquely-weighted profiles, though not on all profiles.
\end{proposition}
\begin{proof} Toward a contradiction, suppose there is a uniquely-weighted profile $\mathbf{P}$ and $a,b\in X(\mathbf{P})$ with ${a\in RP(\mathbf{P}_{-b})}$, ${Margin_\mathbf{P}(a,b)>0}$, and $b\in RP(\mathbf{P})$. Since $a\in RP(\mathbf{P}_{-b})$, by Lemma \ref{RPlem2}, there is a stack $az_1\dots z_n$ for $\mathbf{P}_{-b}$. Now let $\mathcal{M}'$ be a margin graph obtained from $\mathcal{M}(\mathbf{P})$ by changing all of $b$'s negative margins (if any) against candidates in $X(\mathbf{P})\setminus\{a\}$ to be positive and such that all margins between distinct candidates in $\mathcal{M}'$ are distinct, while keeping all other margins from $\mathcal{M}(\mathbf{P})$ the same. Let $\mathbf{P}'$ be a profile with $\mathcal{M}(\mathbf{P}')=\mathcal{M}'$, which exists by Theorem \ref{DebordThm}. Since  $b\in RP(\mathbf{P})$, it follows that $b\in RP(\mathbf{P}')$ (\citealt[p.~204]{Tideman1987}) and hence $RP(\mathbf{P}')=\{b\}$ since $\mathbf{P}'$ is uniquely weighted. However, it is easy to see that $abz_1\dots z_n$ is a stack for $\mathbf{P}'$, so $a\in RP(\mathbf{P}')$ by Lemma \ref{RPlem2}, a contradiction. 

To see that Ranked Pairs does not satisfy immunity to stealers on all profiles, let $\mathbf{P}$ be a profile whose margin graph is displayed on the right below, so the margin graph of $\mathbf{P}_{-b}$ is displayed on the left:
\begin{center}
\begin{minipage}{2.75in}
\begin{tikzpicture}
\node[circle,draw,minimum width=0.25in] at (5,2.25)  (c) {$c$};
\node[circle,draw,minimum width=0.25in] at (2,2.25)  (d) {$d$};
\node[circle,fill = medgreen!50,draw,minimum width=0.25in] at (0,0)  (e) {$e$};
\node[circle,draw,minimum width=0.25in] at (2,-2.25)  (f) {$f$};
\node[circle,fill = medgreen!50,draw,minimum width=0.25in] at (5,-2.25)  (a) {$a$};
\path[->,draw,thick] (a) to[pos=.85] node[fill=white] {$24$} (e);
\path[->,draw,thick] (f) to[pos=.7] node[fill=white] {$24$} (a);
\path[->,draw,thick] (c) to[pos=.7] node[fill=white] {$2$} (a);
\path[->,draw,thick] (c) to[pos=.7] node[fill=white] {$12$} (d);
\path[->,draw,thick] (f) to[pos=.7] node[fill=white] {$10$} (c);
\path[->,draw,thick] (a) to[pos=.89] node[fill=white] {$28$} (d);
\path[->,draw,thick] (c) to[pos=.7] node[fill=white] {$12$} (d);
\path[->,draw,thick] (e) to[pos=.7] node[fill=white] {$14$} (d);
\path[->,draw,thick] (c) to[pos=.85] node[fill=white] {$4$} (e);
\path[->,draw,thick] (e) to[pos=.7] node[fill=white] {$14$} (d);
\path[->,draw,thick] (e) to[pos=.7] node[fill=white] {$26$} (f);
\path[->,draw,thick] (f) to[pos=.7] node[fill=white] {$24$} (a);
\path[->,draw,thick] (f) to[pos=.85] node[fill=white] {$6$} (d);
\path[->,draw,thick] (e) to[pos=.7] node[fill=white] {$26$} (f);
\end{tikzpicture}
\end{minipage}\begin{minipage}{3in}
\begin{tikzpicture}
\node[circle,fill = medgreen!50,draw,minimum width=0.25in] at (7,0)  (b) {$b$};
\node[circle,fill = medgreen!50,draw,minimum width=0.25in] at (5,2.25)  (c) {$c$};
\node[circle,draw,minimum width=0.25in] at (2,2.25)  (d) {$d$};
\node[circle,draw,minimum width=0.25in] at (0,0)  (e) {$e$};
\node[circle,draw,minimum width=0.25in] at (2,-2.25)  (f) {$f$};
\node[circle,draw,minimum width=0.25in] at (5,-2.25)  (a) {$a$};
\path[->,draw,very thick] (a) to[pos=.7] node[fill=white] {$16$} (b);
\path[->,draw,thick] (a) to[pos=.85] node[fill=white] {$24$} (e);
\path[->,draw,thick] (f) to[pos=.7] node[fill=white] {$24$} (a);
\path[->,draw,thick] (a) to[pos=.7] node[fill=white] {$16$} (b);
\path[->,draw,thick] (c) to[pos=.5] node[fill=white] {$8$} (b);
\path[->,draw,thick] (b) to[pos=.85] node[fill=white] {$22$} (e);
\path[->,draw,thick] (c) to[pos=.85] node[fill=white] {$2$} (a);
\path[->,draw,thick] (c) to[pos=.7] node[fill=white] {$12$} (d);
\path[->,draw,thick] (f) to[pos=.89] node[fill=white] {$10$} (c);
\path[->,draw,thick] (a) to[pos=.89] node[fill=white] {$28$} (d);
\path[->,draw,thick] (d) to[pos=.85] node[fill=white] {$18$} (b);
\path[->,draw,thick] (c) to[pos=.7] node[fill=white] {$12$} (d);
\path[->,draw,thick] (e) to[pos=.7] node[fill=white] {$14$} (d);
\path[->,draw,thick] (c) to[pos=.85] node[fill=white] {$4$} (e);
\path[->,draw,thick] (e) to[pos=.7] node[fill=white] {$14$} (d);
\path[->,draw,thick] (e) to[pos=.7] node[fill=white] {$26$} (f);
\path[->,draw,thick] (f) to[pos=.7] node[fill=white] {$24$} (a);
\path[->,draw,thick] (f) to[pos=.85] node[fill=white] {$20$} (b);
\path[->,draw,thick] (f) to[pos=.85] node[fill=white] {$6$} (d);
\path[->,draw,thick] (e) to[pos=.7] node[fill=white] {$26$} (f);
\end{tikzpicture}
\end{minipage}
\end{center}
The only tie in margins is between $(a,e)$ and $(f,a)$. In $\mathbf{P}_{-b}$, if we break the tie between $(a,e)$ and $(f,a)$ in favor of $(a,e)$, then $a$ is the Ranked Pairs winner, whereas if we break it in favor of $(f,a)$, then $e$ is winner. Hence $RP(\mathbf{P}_{-b})=\{a,e\}$. In $\mathbf{P}$, if we break the tie in favor of $(a,e)$, then $c$ is the Ranked Pairs winner, whereas if we break it in favor of $(f,a)$, then $b$ is the winner. Hence $RP(\mathbf{P})=\{b,c\}$. Then since $Margin_\mathbf{P}(a,b)>0$, Ranked Pairs violates immunity to stealers.\end{proof}

\subsection{Stability}\label{StabilitySection}

In virtue of violating (partial) immunity to spoilers or (partial) immunity to stealers, Beat Path, Ranked Pairs, and Minimax all violate (partial) stability for winners. By contrast, we will now prove that Split Cycle satisfies stability for winners. In fact, it satisfies the following slightly stronger property. 

\begin{definition} A voting method $F$ satisfies \textit{strong stability for winners} if for all $\mathbf{P}\in\mathrm{dom}(F)$, all candidates who are weakly Condorcetian for $F$ in $\mathbf{P}$ belong to $F(\mathbf{P})$.
\end{definition}

\begin{proposition}\label{SpoilersProp} Split Cycle satisfies strong stability for winners.
\end{proposition}

\begin{proof} If $a\in SC(\mathbf{P}_{-b})$, then for all $c\in X(\mathbf{P}_{-b})$, $Margin_{\mathbf{P}_{-b}} (c,a)\leq Cycle\#_{\mathbf{P}_{-b}}(c,a)$. As $Margin_{\mathbf{P}_{-b}} (c,a)= Margin_{\mathbf{P}}(c,a)$ and $Cycle\#_{\mathbf{P}_{-b}}(c,a)\leq Cycle\#_{\mathbf{P}}(c,a)$, we have that  (i) for all $c\in X(\mathbf{P}_{-b})$, $Margin_{\mathbf{P}} (c,a)\leq Cycle\#_{\mathbf{P}}(c,a)$. By the assumption that $Margin_{\mathbf{P}} (a,b)\geq 0$, we have $Margin_{\mathbf{P}} (b,a)\leq 0$ and hence (ii) $Margin_{\mathbf{P}} (b,a)\leq Cycle\#_{\mathbf{P}}(b,a)$. By (i) and (ii), $a\in SC(\mathbf{P})$.
\end{proof}
\noindent Informally, the explanation of why Split Cycle satisfies strong stability for winners is the following, using two of our main ideas from Section \ref{SplitCycleSection}: since \textit{defeat is direct}, if $Margin_\mathbf{P}(a,b)\geq 0$, then $b$ does not defeat $a$; and since \textit{incoherence raises  the threshold for defeat}, and adding a candidate can increase incoherence\footnote{This is why Split Cycle may allow a candidate $x$ who is not among the winners in $\mathbf{P}_{-b}$ to become a winner in $\mathbf{P}$.} but cannot decrease incoherence in the initial set of candidates, if a candidate $x$ did not defeat $a$ before the addition of $b$, then it does not defeat $a$ after.

As for other voting methods, GOCHA satisfies stability for winners but not strong stability for winners (see the proof of Proposition \ref{SchwartzProp}), while GETCHA satisfies both.\footnote{To see that GETCHA satisfies strong stability for winners, using the definition of GETCHA in Definition \ref{GETCHA}, suppose $a\in GETCHA(\mathbf{P}_{-b})$ and $Margin_\mathbf{P}(a,b)\geq 0$. Further suppose for contradiction that $a\not\in GETCHA(\mathbf{P})$, which with $Margin_\mathbf{P}(a,b)\geq 0$ implies $b\not\in GETCHA(\mathbf{P})$, so $GETCHA(\mathbf{P})\subseteq X(\mathbf{P}_{-b})$. It follows that $GETCHA(\mathbf{P})$ is $\to_{\mathbf{P}_{-b}}$-dominant, so $GETCHA(\mathbf{P}_{-b})\subseteq GETCHA(\mathbf{P})$, which contradicts the facts that $a\in GETCHA(\mathbf{P}_{-b})$ and $a\not\in GETCHA(\mathbf{P})$.} Uncovered Set satisfies stability for winners, as well as strong stability for winners under some definitions (see Appendix \ref{UncoveredAppendix}).

Let us now turn to stability for winners \textit{with tiebreaking}. In virtue of violating immunity to spoilers or immunity to stealers in profiles with a \textit{unique} Condorcetian candidate, Beat Path, Ranked Pairs, and Minimax all violate stability for winners with tiebreaking (recall Fact \ref{AxRels}). In fact, so does Split Cycle, as there are profiles with multiple undefeated candidates, including non-Condorcetian undefeated candidates, and Split Cycle does not perform any further tiebreaking among undefeated candidates.

\begin{proposition}\label{SCSWTB} Split Cycle does not satisfy stability for winners with tiebreaking.
\end{proposition} 
\begin{proof} Let $\mathbf{P}$ be a profile whose margin graph is shown below:
\begin{center}
\begin{tikzpicture}
\node[circle,draw,minimum width=0.25in] at (0,0)      (a) {$c$}; 
\node[circle,draw,fill = medgreen!50,minimum width=0.25in] at (3,0)      (b) {$b$}; 
\node[circle,draw,fill = medgreen!50,minimum width=0.25in] at (1.5,1.5)  (c) {$a$}; 
\node[circle,draw,minimum width=0.25in] at (1.5,-1.5) (d) {$d$};
\path[->,draw,thick] (b) to[pos=.7] node[fill=white] {$2$} (a);
\path[->,draw,thick] (a) to node[fill=white] {$8$} (c);
\path[->,draw,thick] (d) to node[fill=white] {$10$} (a);
\path[->,draw,thick] (c) to node[fill=white] {$4$} (b);
\path[->,draw,thick] (b) to node[fill=white] {$6$} (d);
\path[->,draw,thick] (c) to[pos=.7]  node[fill=white] {$12$} (d);
\end{tikzpicture}
\end{center}
It is easy to check that $SC(\mathbf{P})=\{a,b\}$. Moreover, candidate $a$ is Condorcetian for $SC$ in $\mathbf{P}$, as witnessed by the fact that $a\in SC(\mathbf{P}_{-b})$ and $Margin_\mathbf{P}(a,b)>0$. However, $b$ is not Condorcetian, as $b$ does not win in $\mathbf{P}_{-c}$ or $\mathbf{P}_{-d}$. Thus, stability for winners with tiebreaking requires that $b$ not win in $\mathbf{P}$.\end{proof}

\begin{remark} In the profile used in the proof of Proposition \ref{SCSWTB}, Ranked Pairs and Stable Voting select only $a$, while Beat Path and Minimax select only $b$. Here Beat Path violates partial immunity to stealers, as $a$ is the only Condorcetian candidate for Beat Path and yet $b$ steals the election from $a$.\end{remark}

We think Split Cycle delivers the correct verdict that neither $a$ nor $b$ is \textit{defeated} in the profile used for Proposition \ref{SCSWTB}; for if a candidate  $x$ is undefeated without a candidate $y$ in the election, and more voters prefer $x$ to $y$ than $y$ to $x$, then $x$ should still be undefeated (see \citealt{HP2021} for extensive discussion). Yet we also recognize the practical imperative to break ties, as well as the appeal in some cases of breaking ties deterministically as far as possible. Thus, in response to the profile used for Proposition~\ref{SCSWTB}, one may reasonably maintain that although neither $a$ nor $b$ is defeated, we can break the tie in favor of $a$ on the grounds that $a$ would have won without $b$ in the election, a majority of voters prefer $a$ to $b$, and no other candidate is Condorcetian in this way. This idea leads to the axiom of stability for winners with tiebreaking (Definition \ref{PartialAxioms}.\ref{PartialAxioms4}), which requires that only Condorcetian candidates be selected, if such candidates exist. For example, the Stable Voting method breaks ties among Condorcetian candidates by selecting the Condorcetian candidate(s) $a$ with the largest margin over a candidate who shows $a$ to be Condorcetian. Thus, Stable Voting satisfies stability for winners with tiebreaking, as well as the following strengthening of the axiom.

\begin{definition}A voting method $F$ satisfies \textit{strong stability for winners with tiebreaking} if not only does $F$ satisfy stability for winners with tiebreaking but also for all $\mathbf{P}\in\mathrm{dom}(F)$, if some candidate is weakly Condorcetian for $F$ in  $\mathbf{P}$ but no candidate is Condorcetian for $F$ in $\mathbf{P}$, then all candidates who win in $\mathbf{P}$ are weakly Condorcetian for $F$ in $\mathbf{P}$.\end{definition}

\begin{proposition} Stable Voting satisfies strong  stability for winners with tiebreaking.
\end{proposition}
\noindent For a proof of stability for winners with tiebreaking that easily adapts to the strong version, see \citealt{HP2022}.

Ultimately the question of whether a voting method $F$ should satisfy stability for winners or stability for winners with tiebreaking depends on the distinction from Section \ref{Comparison} of thinking of $F$ as a pre-tiebreaking method for selecting undefeated candidates or as incorporating tiebreaking among undefeated candidates.

Finally, to use the axioms discussed in this section to choose between voting methods, one can go beyond showing that a voting method $F$ satisfies an axiom while another method $F'$ violates it. One can attempt to study the probability that $F'$ will violate the axiom. In fact, we think the better question is: what is the probability that $F'$ will violate the axiom, \textit{conditional on} $F$ and $F'$ picking different sets of winners in the election? After all, when choosing between $F$ and $F'$ based on how well they pick winners---as opposed to, e.g., their computational cost---all that matters is the elections in which the voting methods disagree. By analogy, when choosing between two insurance policies of the same cost, what matters is the different coverage the policies provide in the event of an accident; the fact that an accident is improbable is not a reason to dismiss as unimportant the differences in coverage.

Let us illustrate the methodology of estimating the probability that a voting method will violate an axiom conditional on its picking different winners than another voting method that satisfies the axiom, deferring a systematic treatment to future work (as in \citealt{HP2021b} for the axiom of positive involvement). First, we imagine Beat Path being proposed as a voting method for selecting the undefeated candidates in an election, a problem for which we think the axiom of stability for winners should be satisfied. Table \ref{SSWTable} shows that in linear profiles with four candidates, the least number for which Beat Path can disagree with Split Cycle, when Beat Path does disagree with Split Cycle, Beat Path almost always violates the strong stability for winners axiom that Split Cycle satisfies, according to several standard probability models. 

\begin{table}[h]
\begin{center}
\begin{tabular}{r|rrrrrrrrrr}
&     10/11  &    20/21 &    50/51 &    100/101 &    500/501 &    1,000/1,001 &    5,000/5,001 \\
\hline
  IC &   100\%    & 99\% & 95\% & 93\% & 91\% & 90\% & 90\%  \\
  IAC  &  100\%    & 97\% & 93\% & 91\% & 90\% & 90\% & 90\% \\
  Urn &   100\%    & 97\% & 93\% & 92\% & 91\% & 90\% & 90\%  \\
  Mallows &   100\%    & 99\% & 97\% & 94\% & 91\% & 89\% & 90\%  
\end{tabular}
\end{center}
\caption{Of four-candidate linear profiles in which $BP(\mathbf{P})\neq SC(\mathbf{P})$, the percentage in which  Beat Path violates strong stability for winners, i.e., there are $a,b\in X(\mathbf{P})$ such that $a\in BP(\mathbf{P}_{-b})$, $Margin_\mathbf{P}(a,b)\geq 0$, and $a\not\in BP(\mathbf{P})$. For each column labeled by $n$/$n+1$, we sampled 1,000,000 profiles with $n$ voters and 1,000,000 with $n+1$ voters using the probability model named in the row, defined in Appendix \ref{Quant}.}\label{SSWTable}
\end{table}

\noindent Second, we imagine Beat Path proposed as a voting method for breaking ties among the undefeated candidates as determined by Split Cycle, a problem for which we think the axiom of stability for winners with tiebreaking should hold. Table \ref{SSWTBTable} shows that in linear profiles with four candidates, among those in which Beat Path disagrees with Stable Voting, a non-trivial percentage involve violations by Beat Path of the strong stability for winners with tiebreaking axiom satisfied by Stable Voting, using the same probability~models.

\begin{table}[h]
\begin{center}
\begin{tabular}{r|rrrrrrrrrr}

 &     10/11  &    20/21 &    50/51 &    100/101 &    500/501 &    1,000/1,001 &    5,000/5,001 \\
\hline
  IC &   29\%    & 24\% & 21\% & 20\% & 21\% & 22\% & 29\%  \\
    IAC  &   27\%    & 21\% & 19\% & 20\% & 32\% & 41\% & 59\% \\
  Urn &   29\%    & 23\% & 20\% & 21\% & 31\% & 38\% & 52\%  \\
  Mallows &   17\%    & 13\% & 11\% & 10\% & 11\% & 14\% & 16\%  
\end{tabular}
\end{center}
\caption{Of four-candidate linear profiles in which $BP(\mathbf{P})\neq SV(\mathbf{P})$, the percentage in which  Beat Path violates strong stability for winners with tiebreaking. For each column labeled by $n$/$n+1$, we sampled 1,000,000 profiles with $n$ voters and 1,000,000 with $n+1$ voters using the probability model named in the row, defined in Appendix~\ref{Quant}.}\label{SSWTBTable}
\end{table}

\subsection{Expansion}\label{ExpansionSection}

Stability for winners concerns when a winner $a$ in a profile $\mathbf{P}_{-b}$ remains a winner in a profile $\mathbf{P}$ with one new candidate $b$ added. The profile $\mathbf{P}$ may be viewed as combining the profile $\mathbf{P}_{-b}$ without $b$ and the two-candidate profile $\mathbf{P}_{ab}$ with $X(\mathbf{P}_{ab})=\{a,b\}$ that agrees with $\mathbf{P}$ on the ranking of $a$ vs.~$b$ assigned to each voter. As $Margin_\mathbf{P}(a,b)\geq 0$ is equivalent (for voting methods based on majority voting) to $a$ being a winner in the two-candidate profile $\mathbf{P}_{ab}$, (strong) stability for winners can be restated as follows:  if $a$ is a winner in both $\mathbf{P}_{-b}$ and $\mathbf{P}_{ab}$, then $a$ is a winner in the full profile $\mathbf{P}$. We can generalize this idea to apply not only to profiles of the form $\mathbf{P}_{-b}$ and $\mathbf{P}_{ab}$ but to any two subprofiles of $\mathbf{P}$ such that every candidate in $\mathbf{P}$ appears in one of the subprofiles (``if $a$ can win the battles separately, then $a$ can win the war'').

\begin{definition}\label{ExpandCon} A voting method $F$ satisfies \textit{expansion consistency} if for all $\mathbf{P}\in \mathrm{dom}(F)$ and nonempty $Y,Z\subseteq X(\mathbf{P})$ with $Y\cup Z=X(\mathbf{P})$, we have $F(\mathbf{P}\vert _{ Y})\cap F(\mathbf{P}\vert _{ Z}) \subseteq F(\mathbf{P})$.
\end{definition}

Since expansion consistency implies strong stability for winners for voting methods that agree with majority rule on two-candidate profiles, all the voting methods shown to violate stability for winners in the previous sections also violate expansion consistency.

\begin{proposition} Split Cycle satisfies expansion consistency.
\end{proposition}

\begin{proof} The proof is similar to that of Proposition \ref{SpoilersProp}. If $a\in SC(\mathbf{P}\vert _{ Y})\cap SC(\mathbf{P}\vert _{ Z})$, then for all $c\in Y$,  $Margin_{\mathbf{P}\vert _{ Y}} (c,a)\leq Cycle\#_{\mathbf{P}\vert _{ Y}}(c,a)$, and for all $c\in Z$, $Margin_{\mathbf{P}\vert _{ Z}} (c,a)\leq Cycle\#_{\mathbf{P}\vert _{ Z}}(c,a)$. Clearly for all $c\in Y$, $Margin_{\mathbf{P}} (c,a) = Margin_{\mathbf{P}\vert _{ Y}} (c,a)$, and for all $c\in Z$, $Margin_{\mathbf{P}} (c,a)=Margin_{\mathbf{P}\vert _{ Z}} (c,a)$. Moreover, for all $c\in Y$, $Cycle\#_{\mathbf{P}\vert _{ Y}} (c,a) \leq Cycle\#_{\mathbf{P}} (c,a)$, and for all $c\in Z$, $Cycle\#_{\mathbf{P}\vert _{ Z}} (c,a) \leq Cycle\#_{\mathbf{P}} (c,a)$. Thus, for all $c\in Y\cup Z= X(\mathbf{P})$, $Margin_{\mathbf{P}} (c,a)\leq Cycle\#_{\mathbf{P}} (c,a)$. Hence $a\in SC(\mathbf{P})$.
\end{proof}
\noindent Intuitively, the reason Split Cycle satisfies expansion consistency is the following. First, given $a\in F(\mathbf{P}\vert _{ Y})\cap F(\mathbf{P}\vert _{ Z})$, the margins for each candidate $x$ over $a$ do not change from $\mathbf{P}\vert _{ Y},\mathbf{P}\vert _{ Z}$ to  $\mathbf{P}$. How then could $a$ suddenly be defeated by a candidate $x$ in  $\mathbf{P}$? By our first and third main ideas in Section \ref{SplitCycleSection}, this would require the margin for $x$ over $a$ to meet the threshold for defeat determined by incoherence; but incoherence can only \textit{increase} from $\mathbf{P}\vert _{ Y},\mathbf{P}\vert _{ Z}$ to $\mathbf{P}$, not decrease, so if $x$'s margin over $a$ was not sufficient to defeat $a$ in $\mathbf{P}\vert _{ Y}$ (resp.~$\mathbf{P}\vert _{ Z}$), then it is not sufficient to defeat $a$ in $\mathbf{P}$.

An example of a voting method satisfying stability for winners but not expansion consistency is the Banks voting method, defined as follows.\footnote{We thank Felix Brandt for providing this example.} Say that a \textit{chain} in $M(\mathbf{P})$ is a subset of $X(\mathbf{P})$ linearly ordered by the majority relation of $\mathbf{P}$. Then $a\in Banks(\mathbf{P})$ if $a$ is the maximum element with respect to the majority relation of some maximal chain in $M(\mathbf{P})$.
 
\begin{proposition} Banks satisfies strong stability for winners but not expansion consistency.
\end{proposition}
\begin{proof} For strong stability for winners, suppose $a\in Banks(\mathbf{P}_{-b})$ and $Margin_\mathbf{P}(a,b)\geq 0$. Hence  $a$ is the maximum element of some maximal chain $R$ in $M(\mathbf{P}_{-b})$. Let $R'$ be a maximal chain in $M(\mathbf{P})$ with $R\subseteq R'$. Since $Margin_\mathbf{P}(a,b)\geq 0$, $b$ is not the maximum  of $R'$, so $a$ is the maximum  of $R'$, so $a\in Banks(\mathbf{P})$.

For the violation of expansion consistency, consider a profile $\mathbf{P}$ with the majority graph below:
\begin{center}
\begin{minipage}{2.5in}
\begin{tikzpicture}
\node[circle,draw,minimum width=0.25in] at (-1,2)  (a) {$a$}; 
\node[circle,draw,minimum width=0.25in] at (2,2)  (b) {$b$}; 
\node[circle,draw,minimum width=0.25in] at (4,0)  (c) {$c$};
\node[circle,draw,minimum width=0.25in] at (2,-2) (d) {$d$}; 
\node[circle,draw,minimum width=0.25in] at (-1,-2) (e) {$e$}; 
\node[circle,draw,minimum width=0.25in] at (-3,-1)  (f) {$f$};
\node[circle,draw,minimum width=0.25in] at (-3,1)  (g) {$g$};

\path[->,draw,thick] (a) to node {} (b);
\path[<-,draw,thick] (a) to node {} (c);
\path[<-,draw,thick] (a) to node {} (d);
\path[<-,draw,thick] (a) to node {} (e);
\path[->,draw,thick] (a) to node {} (f);
\path[->,draw,thick] (a) to node {} (g);

\path[<-,draw,thick] (b) to node {} (c);
\path[<-,draw,thick] (b) to node {} (d);
\path[->,draw,thick] (b) to node {} (e);
\path[->,draw,thick] (b) to node {} (f);
\path[<-,draw,thick] (b) to node {} (g);

\path[->,draw,thick] (c) to node {} (d);
\path[->,draw,thick] (c) to node {} (e);
\path[<-,draw,thick] (c) to node {} (f);
\path[->,draw,thick] (c) to node {} (g);

\path[->,draw,thick] (d) to node {} (e);
\path[->,draw,thick] (d) to node {} (f);
\path[<-,draw,thick] (d) to node {} (g);

\path[->,draw,thick] (e) to node {} (f);
\path[->,draw,thick] (e) to node {} (g);

\path[->,draw,thick] (f) to node {} (g);

\end{tikzpicture}
\end{minipage}
\end{center}
The maximal chains of $M(\mathbf{P})$, ordered by the majority relation, are as follows:
\begin{eqnarray*}
&&\mbox{$(c, a, g, b)$\quad $(c, g, d, b)$\quad $(c, e, a, g)$\quad $(c, d, b, e)$\quad $(c, d, e, a)$}\\
&&\mbox{$(f,c,g)$\quad$(e, a, f, g)$\quad $(d, b, e, f)$\quad $(d, a, b, f)$\quad$(d, e, a, f)$.}
\end{eqnarray*}
Then one can check that  $Banks(\mathbf{P}\vert _{ \{a,b,c,e,f\}}) \cap Banks(\mathbf{P}\vert _{ \{a, d,f,g\}}) \ni a \not\in Banks(\mathbf{P})=\{c, d, e, f\}$.\end{proof}

\begin{remark} Expansion consistency as in Definition \ref{ExpandCon} is the analogue for voting methods of Sen's \citeyearpar[p.~314]{Sen1971} expansion consistency (also known as $\gamma$) condition on choice functions. A \textit{choice function} on a nonempty set $X$ is a function $C:\wp(X)\setminus\{\varnothing\}\to\wp(X)\setminus\{\varnothing\}$ such that for all nonempty $S\subseteq X$, $\varnothing\neq C(S)\subseteq S$. Then $C$ satisfies expansion consistency if for all nonempty $S,T\subseteq X$, $C(S)\cap C(T)\subseteq C(S\cup T)$. 

To relate choice consistency conditions to social choice, Sen \citeyearpar{Sen1993} defines, for some fixed nonempty sets $X$ and $V$, a \textit{functional collective choice rule} (FCCR) for $(X,V)$ to be a function $f$ mapping each profile $\mathbf{P}$ with $X(\mathbf{P})=X$ and $V(\mathbf{P})=V$ to a choice function $f(\mathbf{P},\cdot)$ on $X(\mathbf{P})$. Let a \textit{variable-election FCCR} (VFCCR) be a function $f$ mapping each profile $\mathbf{P}$ (i.e., allowing $X(\mathbf{P})$ and $V(\mathbf{P})$ to vary with $\mathbf{P}$) to a choice function  $f(\mathbf{P},\cdot)$ on $X(\mathbf{P})$. For voting theory, we take $X(\mathbf{P})$ to be the set of \textit{candidates who appeared on the ballot} in the election scenario modeled by $\mathbf{P}$---since there is no practical point in considering voting procedures that have access to ranking information not on submitted ballots---but after the ballots are collected, some candidates may withdraw from consideration, be rejected by higher authorities, become incapacitated, etc., leaving us to choose winners from some ``feasible set'' $S\subsetneq X(\mathbf{P})$. 

There are three ways to apply the notion of expansion consistency to a VFCCR $f$:
\begin{itemize}
\item $f$ satisfies \textit{feasible expansion consistency} if for all profiles $\mathbf{P}$, $f(\mathbf{P},\cdot)$ satisfies expansion consistency, i.e., for all $S, T\subseteq X(\mathbf{P})$, $f(\mathbf{P},S)\cap f(\mathbf{P},T)\subseteq f(\mathbf{P},S\cup T)$;
\item $f$ satisfies \textit{profile expansion consistency} if for all profiles $\mathbf{P}$ and $Y,Z\subseteq X(\mathbf{P})$ with $Y\cup Z=X(\mathbf{P})$, $f(\mathbf{P}\vert _{ Y},Y)\cap f(\mathbf{P}\vert _{ Z},Z)\subseteq f(\mathbf{P},X(\mathbf{P}))$;
\item $f$ satisfies \textit{full expansion consistency} if for all profiles $\mathbf{P}$ and $Y,Z\subseteq X(\mathbf{P})$ with $Y\cup Z=X(\mathbf{P})$, $S\subseteq Y$, and $T\subseteq Z$, we have $f(\mathbf{P}\vert _{ Y}, S)\cap f(\mathbf{P}\vert _{ Z}, T)\subseteq f(\mathbf{P}, S\cup T)$.
\end{itemize}
 Thus, profile expansion consistency and full expansion consistency constrain the relation between sets of winners for \textit{different election scenarios} involving rankings of different (but possibly overlapping or extended) sets of candidates, modeled by different profiles; by contrast, feasible expansion consistency constrains the choices of winners from feasible sets based on a fixed profile of rankings. All three notions of expansion consistency are equivalent for VFCCRs satisfying the independence condition that for all profiles $\mathbf{P}$ and nonempty $S\subseteq X(\mathbf{P})$, $f(\mathbf{P},S)=f(\mathbf{P}\vert _{ S},S)$. However, they are not equivalent in general. For example, consider the \textit{global Borda count} VFCCR (cf.~\citealt[p.~71]{Kelly1988}): $f(\mathbf{P},S)$ is the set of all $x\in S$ such that for all $y\in S$, the Borda score of $x$ calculated with respect to the full profile $\mathbf{P}$ is at least that of $y$.\footnote{The \textit{local Borda count} VFCCR (cf.~\citealt[p.~74]{Kelly1988}) takes $f(\mathbf{P},S)$ to be the set of all $x\in S$ such that for all $y\in S$, the Borda score of $x$ calculated with respect to the restricted profile $\mathbf{P}\vert _{ S}$ is at least that of $y$.} Global Borda count satisfies feasible expansion consistency but not profile expansion consistency. By contrast, both local and global VFCCR versions of Split Cycle\footnote{Global Split Cycle takes $f(\mathbf{P},S)$ to be the set of all $x\in S$ such that for all $y\in S$, $y$ does not defeat $x$ according to the defeat relation calculated with respect to $\mathbf{P}$, whereas local Split Cycle uses the defeat relation calculated with respect to $\mathbf{P}\vert _{ S}$.} satisfy full expansion consistency.\footnote{It is more difficult to describe VFCCRs satisfying both feasible and profile but not full expansion consistency. However, the essential point can be seen by considering a binary choice function $C$ that for $\varnothing\neq S\subseteq Y\subseteq X(\mathbf{P})$ returns a nonempty $C(Y,S)\subseteq S$; think of $Y$ as determining the profile $\mathbf{P}\vert _{ Y}$ and $S$ as determining the feasible set $S\subseteq Y$. Let $X(\mathbf{P})=\{a,b,c\}$ and let $C$ be as follows: $C(\{a,b,c\}, \{a,b,c\})= \{a,b\}$, $C(\{a,b,c\}, \{a,b\})=\{b\}$, $C(\{a,b,c\},\{a,c\})=\{a\}$, $C(\{a,b,c\},\{b,c\})=\{b\}$; $C(\{a,b\},\{a,b\})=\{a\}$; $C(\{a,c\},\{a,c\})=\{a\}$; $C(\{b,c\},\{b,c\})=\{b\}$. This violates full expansion consistency because ${a\in C(\{a,b\},\{a,b\})\cap C(\{a,c\}, \{a\})}$ but $a\not\in C(\{a,b,c\}, \{a,b\})$. Yet for each nonempty $Y\subseteq X(\mathbf{P})$, $C(Y,\cdot)$ is a choice function satisfying expansion consistency, so $C$ satisfies feasible expansion consistency, and the choice function $C'$ defined by $C'(S)=C(S,S)$ satisfies expansion consistency, so $C$ satisfies profile expansion consistency.}\end{remark}
 
 \begin{remark}Expansion consistency may remind one of the \textit{reinforcement} criterion (see \citealt{Pivato2013b}), but that criterion concerns combining profiles for disjoint sets of voters voting on the same set of candidates, whereas expansion consistency concerns  profiles for the same set of voters voting on different sets of candidates. Reinforcement states that for any profiles $\mathbf{P}$ and $\mathbf{P}'$ with $V(\mathbf{P})\cap V(\mathbf{P}')=\varnothing$ and $X(\mathbf{P})=X(\mathbf{P}')$, if $F(\mathbf{P})\cap F(\mathbf{P}')\neq\varnothing$, then $F(\mathbf{P}+\mathbf{P}')= F(\mathbf{P})\cap F(\mathbf{P}')$. No Condorcet consistent voting method satisfies reinforcement (see, e.g., \citealt[Proposition 2.5]{Zwicker2016}); we know of no non-trivial voting method that satisfies both expansion consistency and reinforcement;\footnote{Expansion consistency implies a weakening of Condorcet consistency---it implies that if there is a Condorcet winner, that candidate must be \textit{among} the winners. But with no other axioms, expansion consistency and reinforcement are consistent, as they are both satisfied by the trivial voting method that always picks all candidates as winners.} and we do not find reinforcement normatively compelling for all voting contexts.\footnote{In the context where there is a ``true'' ranking of the candidates of which voters have noisy perceptions, reinforcement is satisfied by any voting method that can be rationalized as the maximum likelihood estimator for some noise model with i.i.d. votes (see \citealt{Conitzer2005,Pivato2013}).} Why when the voting method is given the full information in $\mathbf{P}+\mathbf{P}'$ should it be constrained by what it outputs when given only the limited information of $\mathbf{P}$ and only the limited information in $\mathbf{P}'$? For example, for three candidates $a,b,c$, suppose $\mathbf{P}$ is the classic Condorcet paradox profile with 6 voters such that $Margin_\mathbf{P}(a,b)=2$,  $Margin_\mathbf{P}(b,c)=2$, and $Margin_\mathbf{P}(c,a)=2$, so $F(\mathbf{P})=\{a,b,c\}$ by fairness considerations (i.e., for any anonymous and neutral method $F$), while $\mathbf{P}'$ is a profile with $3$ voters such that $Margin_{\mathbf{P}'}(b,a)=1$, $Margin_{\mathbf{P}'}(a,c)=1$, and $Margin_{\mathbf{P}'}(b,c)=1$, so $F(\mathbf{P}')=\{b\}$ by Condorcet consistency. When we look at all the information in $\mathbf{P}+\mathbf{P}'$, we see that $a$ is majority preferred to every other candidate---$a$ is the Condorcet winner, so $F(\mathbf{P}+\mathbf{P}')=\{a\}$. Note that $b$ is only majority preferred to $a$ by a small margin in $\mathbf{P}'$, whereas $a$ is majority preferred to $b$ by a larger margin in $\mathbf{P}$. Due to $c$'s poor performance in $\mathbf{P}'$, there is no cycle in the full profile $\mathbf{P}+\mathbf{P}'$, so although $a$'s margin over $b$ failed to make $a$ the winner in $\mathbf{P}$ due to fairness considerations, $a$'s margin over $b$ in the full cycle-free profile makes $a$ the winner, as it should. Note this also shows it is a mistake to think that given a profile like $\mathbf{P}+\mathbf{P}'$, one can assume that the ``Condorcet component'' $\mathbf{P}$ can be deleted from the profile, as if it contains no information, while not changing the winning set (cf.~the property of \textit{cancelling properly} in \citealt[p.~77]{Balinski2010}). The Condorcet component $\mathbf{P}$ contains the important information that $Margin_\mathbf{P}(a,b)=2$.\end{remark}

\section{Other Criteria}\label{OtherCriteria}

In this section, we consider how Split Cycle fairs with respect to a number of other criteria for voting methods. We organize the criteria into five groups: symmetry criteria (\ref{SymmetrySection}), dominance criteria (\ref{DominanceSection}), independence criteria (\ref{IndependenceSection}), resoluteness criteria (\ref{ResolutenessSection}), and monotonicity criteria (\ref{MonSection}).

\subsection{Symmetry Criteria}\label{SymmetrySection}

\subsubsection{Anonymity and Neutrality}\label{ANSection}

The two most basic symmetry criteria say that permuting voter names does not change the election result (anonymity) and permuting candidate names changes the result according to the permutation (neutrality).

\begin{definition} Let $F$ be a voting method.
\begin{enumerate}
\item $F$ satisfies \textit{anonymity} if for any $\mathbf{P},\mathbf{P}'\in\mathrm{dom}(F)$ and $X(\mathbf{P})=X(\mathbf{P}')$, if there is a permutation $\pi$ of $\mathcal{V}$ such that $\pi[V(\mathbf{P})]=V'(\mathbf{P})$ and $\mathbf{P}(i)=\mathbf{P}'(\pi(i))$ for all $i\in V(\mathbf{P})$, then $F(\mathbf{P})=F(\mathbf{P}')$;
 \item $F$ satisfies \textit{neutrality} if for any $\mathbf{P},\mathbf{P}'\in\mathrm{dom}(F)$ with $V(\mathbf{P})=V(\mathbf{P}')$, if there is a permutation $\tau$ of $\mathcal{X}$ such that $\tau[X(\mathbf{P})]=X(\mathbf{P}')$ and for each $i\in V(\mathbf{P})$ and $x,y\in X(\mathbf{P})$, $(x,y)\in \mathbf{P}(i)$ if and only if $(\tau(x),\tau(y))\in \mathbf{P}'(i)$, then $\tau[F(\mathbf{P})]= F(\mathbf{P}')$.
\end{enumerate}
\end{definition}

The following is obvious from the definition of Split Cycle.

\begin{proposition} Split Cycle satisfies anonymity and neutrality.
\end{proposition}

\subsubsection{Reversal Symmetry}\label{RSSection}

Next we consider a criterion due to Saari (\citealt[\S~3.1.3]{Saari1994}; \citealt[\S~7.1]{Saari1997}) that can be seen as related to neutrality. Neutrality implies that if we swap the places of two candidates $a$ and $b$ on every voter's ballot, then if $a$ won the election before the swap, $b$ should win the election after the swap. Reversal symmetry extends this idea from pairwise swaps to full reversals of voters'~ballots. 

\begin{definition} A voting method $F$ satisfies \textit{reversal symmetry} if for any $\mathbf{P}\in\mathrm{dom}(F)$ with $\vert X(\mathbf{P})\vert >1$, if $F(\mathbf{P})=\{x\}$, then $x\not\in F(\mathbf{P}^{r})$, where $\mathbf{P}^{r}$ is the profile such that $\mathbf{P}^r_i= \{(x,y)\mid (y,x)\in \mathbf{P}_i\}$.
\end{definition}

\begin{proposition} Split Cycle satisfies reversal symmetry.
\end{proposition}

\begin{proof} Suppose $\mathbf{P}$ is such that $\vert X(\mathbf{P})\vert >1$ and $SC(\mathbf{P})=\{x\}$. It follows by Lemma \ref{BeatPathFromWinner} that there is a $y\in X(\mathbf{P})$ such that $x$ defeats $y$ in $\mathbf{P}$, i.e., such that $Margin_{\mathbf{P}}(x,y)> Cycle\#_{\mathbf{P}}(x,y)$. But then since $Margin_{\mathbf{P}}(x,y)=Margin_{\mathbf{P}^{r}}(y,x)$ and $Cycle\#_{\mathbf{P}}(x,y)=Cycle\#_{\mathbf{P}^{r}}(y,x)$, we have  $Margin_{\mathbf{P}^{r}}(y,x)> Cycle\#_{\mathbf{P}^{r}}(x,y)$, so $x\not\in SC(\mathbf{P}^{r})$.\end{proof}

\subsection{Dominance Criteria}\label{DominanceSection}

\subsubsection{Pareto}\label{ParetoSection}

Our first dominance criterion is the well-known Pareto principle (see, e.g., \citealt[Definition 2.6]{Zwicker2016}), stating that Pareto-dominated candidates cannot be elected.

\begin{definition}A voting method $F$ satisfies \textit{Pareto} if for any $\mathbf{P}\in \mathrm{dom}(F)$ and $a,b\in X(\mathbf{P})$, if all voters in $\mathbf{P}$ rank $a$ above $b$, then $b\not\in F(\mathbf{P})$. 
\end{definition}

\begin{proposition} Restricted to the class of acyclic profiles, Split Cycle satisfies Pareto. 
\end{proposition}
\begin{proof} Suppose all voters in $\mathbf{P}$ rank $a$ above $b$, so $Margin_\mathbf{P}(a,b)=\vert V(\mathbf{P})\vert $. Since it is impossible to have a cycle $\rho=a \overset{\vert V(\mathbf{P})\vert }{\longrightarrow} b \overset{\vert V(\mathbf{P})\vert }{\longrightarrow} x_1 \overset{\vert V(\mathbf{P})\vert }{\longrightarrow}\dots \overset{\vert V(\mathbf{P})\vert }{\longrightarrow} x_n \overset{\vert V(\mathbf{P})\vert }{\longrightarrow} a$ in an acyclic profile, one edge of any such cycle must have weight less than $\vert V(\mathbf{P})\vert $, so $a$ defeats $b$ by Lemma \ref{OnlySomeCycles}.\end{proof}

\subsubsection{Condorcet Winner and Loser}\label{CondorcetSection}

The next notions of dominance are based on majority preference rather than unanimity: the \textit{Condorcet} (\textit{winner}) \textit{criterion} states that a candidate who is majority preferred to every other candidate must be the unique winner, while the \textit{Condorcet loser criterion} states that a candidate who is majority dispreferred to every other candidate must not be among the winners.

\begin{definition} For a profile $\mathbf{P}$  and $x\in X(\mathbf{P})$, we say that $x$ is a \textit{Condorcet winner} (resp.~\textit{Condorcet loser}) \textit{in $\mathbf{P}$} if for every $y\in X(\mathbf{P})\setminus\{x\}$, we have $Margin(x,y)>0$ (resp.~$Margin(y,x)>0$ and $X(\mathbf{P})\neq\{x\}$).
\end{definition}

\begin{definition}\label{CondWinLoss}  A voting method $F$ satisfies the \textit{Condorcet} (\textit{winner}) \textit{criterion} (resp.~\textit{Condorcet loser criterion}) if for every $\mathbf{P}\in\mathrm{dom}(F)$ and $x\in X(\mathbf{P})$, if $x$ is the Condorcet winner (resp.~Condorcet loser), then $F(\mathbf{P})=\{x\}$ (resp.~$x\not\in F(\mathbf{P})$). If $F$ satisfies the Condorcet criterion, we say that $F$ is \textit{Condorcet consistent}.
\end{definition}

\begin{proposition} Split Cycle satisfies the Condorcet criterion and the Condorcet loser criterion.
\end{proposition}

\begin{proof} If $x$ is the Condorcet winner (resp.~loser), then for every $y\in X(\mathbf{P})\setminus\{x\}$,  we have $Margin_{\mathbf{P}}(x,y)>0$ (resp.~$Margin_{\mathbf{P}}(y,x)>0$). It follows that $x$ is not involved in any cycles, so for every $y\in X(\mathbf{P})\setminus\{x\}$, we have $Margin_{\mathbf{P}}(x,y)>Cycle\#_{\mathbf{P}}(x,y)=0$ (resp.~$Margin_{\mathbf{P}}(y,x)>Cycle\#_{\mathbf{P}}(y,x)=0$). Hence $x$ defeats every other candidate, so $SC(\mathbf{P})=\{x\}$ (resp.~is defeated by every other candidate, so $x\not\in SC(\mathbf{P}$)).\end{proof}

\subsubsection{Smith and Schwartz Criteria}\label{SmithSchwartz}

A strengthening of the Condorcet criterion is the \textit{Smith criterion} (\citealt{Smith1973}), according to which the set of winners must be a subset of the \textit{Smith set}---the smallest set of candidates such that every candidate inside the set is majority preferred to every candidate outside the set. Following the terminology of Schwartz \citeyearpar{Schwartz1986}, we also call the Smith set the GETCHA set (``GETCHA'' stands for ``generalized top-choice assumption'').

\begin{definition}\label{GETCHA} Let $\mathbf{P}$ be a profile and $S\subseteq X(\mathbf{P})$. Then $S$ is \textit{$\to_\mathbf{P}$-dominant} if $S\neq\varnothing$ and for all $x\in S$ and $y\in X(\mathbf{P})\setminus S$, we have $x\to_\mathbf{P}y$. Define
\[GETCHA(\mathbf{P})=\bigcap \{S\subseteq X(\mathbf{P})\mid S\mbox{ is $\to_\mathbf{P}$-dominant}\}.\]
\end{definition}

\begin{definition} A voting method $F$ satisfies the \textit{Smith criterion} if for any $\mathbf{P}\in\mathrm{dom}(F)$, we have $F(\mathbf{P})\subseteq GETCHA(\mathbf{P})$.
\end{definition}

\begin{proposition}\label{SmithProp} Split Cycle satisfies the Smith criterion.
\end{proposition}

\begin{proof}  Suppose $b\in SC(\mathbf{P})\setminus GETCHA(\mathbf{P})$. Since $b\not \in GETCHA(\mathbf{P})$, there is an $a\in GETCHA(\mathbf{P})$ such that $a\to_\mathbf{P} b$. Then since $b\in SC(\mathbf{P})$, it follows by Lemma \ref{OnlySomeCycles} that there is a simple cycle $\rho$ of the form $a\to b\to x_1\to\dots \to x_n\to a$. Hence one of the edges in $\rho$ goes from a candidate outside $GETCHA(\mathbf{P})$ to a candidate inside $GETCHA(\mathbf{P})$, which is a contradiction.\end{proof}

Next we consider a strengthening of the Smith criterion, based on the idea of the Schwartz set, or in Schwartz's \citeyearpar{Schwartz1986} terminology, the GOCHA set (``GOCHA'' stands for ``generalized optimal choice axiom'').

\begin{definition}\label{GOCHA} Let $\mathbf{P}$ be a profile and $S\subseteq X(\mathbf{P})$. Then $S$ is \textit{$\to_\mathbf{P}$-undominated} if for all $x\in S$ and $y\in X(\mathbf{P})\setminus S$, we have $y\not\to_\mathbf{P}x$. Define \[GOCHA(\mathbf{P})=\bigcup \{S\subseteq X(\mathbf{P})\mid S\mbox{ is $\to_\mathbf{P}$-undominated and no }S'\subsetneq S\mbox{ is $\to_\mathbf{P}$-undominated}\}.\]
\end{definition}

Note that if there are no zero margins between distinct candidates in $\mathbf{P}$, then $GOCHA(\mathbf{P})=GETCHA(\mathbf{P})$. Another useful characterization of the GOCHA set is given by the following lemma.

\begin{lemma}[\citealt{Schwartz1986}, Corollary 6.2.2]\label{GOCHALem}  Let $\mathbf{P}$ be any profile, and let $\to_\mathbf{P}^*$ be the transitive closure of $\to_\mathbf{P}$, i.e., $a\to_\mathbf{P}^*b$ if and only if there are $x_1,\dots , x_n\in X(\mathbf{P})$ with $a=x_1$ and $b=x_n$ such that $x_1\to_\mathbf{P}\dots \to_\mathbf{P}x_n$. Then \[GOCHA(\mathbf{P})=\{x\in X(\mathbf{P})\mid \mbox{there is no }y\in X(\mathbf{P}): y\to_\mathbf{P}^*x\mbox{ and }x\not\to_\mathbf{P}^*y\}.\]
\end{lemma}

Just as the Smith criterion states that the set of winners should always be a subset of the Smith set, the Schwartz criterion states that the set of winners should always be a subset of the Schwartz set. 

\begin{definition} A voting method $F$ satisfies the \textit{Schwartz criterion} if for any $\mathbf{P}\in\mathrm{dom}(F)$, $F(\mathbf{P})\subseteq GOCHA(\mathbf{P})$.
\end{definition}

In contrast to Proposition \ref{SmithProp}, Split Cycle does not satisfy the Schwartz criterion. After the proof, we will explain why we do not find the Schwartz criterion normatively plausible.

\begin{proposition}\label{SchwartzProp} Split Cycle does not satisfy the Schwartz criterion, even when restricted to linear profiles.
\end{proposition}

\begin{proof} By Debord's Theorem, there is a linear profile $\mathbf{P}$ with the following margin graph (simplifying our example for the third idea of Section \ref{ThreeIdeasSection}):

\begin{center}
\begin{tikzpicture}

\node[circle,draw,minimum width=0.25in] at (3,0) (a) {$a$}; 

\node[circle,draw, minimum width=0.25in] at (5,0) (e) {$e$}; 
\node[circle,draw,minimum width=0.25in] at (8,0) (d) {$d$}; 
\node[circle,draw,minimum width=0.25in] at (6.5,1.5) (f) {$f$};  

\path[->,draw,thick] (e) to node[fill=white] {$2$} (d);
\path[->,draw,thick] (f) to node[fill=white] {$2$} (e);
\path[->,draw,thick] (d) to node[fill=white] {$2$} (f);

\path[->,draw,thick, bend left] (a) to node[fill=white] {$2$} (f);

  \end{tikzpicture}
  \end{center}
First, note that $d\in SC(\mathbf{P})$ (indeed, $SC(\mathbf{P})=\{a,d,e\})$, because the only candidate with a positive margin over $d$ is $e$, but $Margin_\mathbf{P}(e,d)\not > Cycle\#_\mathbf{P}(e,d)$. Yet $d\not\in GOCHA(\mathbf{P})$, because $a\to_\mathbf{P}^*d$ and $d\not\to_\mathbf{P}^*a$.\end{proof}

For the reasons explained in Section \ref{ThreeIdeasSection} for the idea that \textit{defeat is direct}, we think that $d$ should not be kicked out of the winning set by $a$ in the profile $\mathbf{P}$ used in the proof of Proposition \ref{SchwartzProp}. Thus, we do not accept the Schwartz criterion. The profile used in the proof of Proposition \ref{SchwartzProp} also shows the following. 

\begin{proposition} There is no voting method on the domain of linear profiles satisfying anonymity, neutrality, strong stability for winners, and the Schwartz criterion.
\end{proposition}

\begin{proof} Where $\mathbf{P}$ is the linear profile used in the proof of Proposition \ref{SchwartzProp}, by anonymity and neutrality, $F(\mathbf{P}_{-a})=\{d,e,f\}$. Then since $Margin_\mathbf{P}(a,d)=0$, it follows by strong stability for winners that $d\in F(\mathbf{P})$, which contradicts the Schwartz criterion as in the proof of Proposition \ref{SchwartzProp}.\end{proof}

\subsection{Independence Criteria}\label{IndependenceSection}

\subsubsection{Independence of Smith-Dominated Alternatives}\label{ISDASection}

The Smith criterion of Section \ref{SmithSchwartz} can be strengthened to the criterion that deletion of candidates outside the Smith set should not change the set of winners.

\begin{definition} A voting method $F$ satisfies \textit{independence of Smith-dominated alternatives} (ISDA) if for any $\mathbf{P}\in\mathrm{dom}(F)$ and $x\in X(\mathbf{P})\setminus GETCHA(\mathbf{P})$, we have $F(\mathbf{P})=F(\mathbf{P}_{-x})$.
\end{definition}

\begin{remark} ISDA implies the Smith criterion, since if $x\in F(\mathbf{P})\setminus  GETCHA(\mathbf{P})$, then $F(\mathbf{P})\neq F(\mathbf{P}_{-x})$.
\end{remark}

\begin{remark}\label{RestrictionRemark} If $F$ satisfies ISDA, then $F$ satisfies the following inter-profile condition: for any profiles $\mathbf{P}$ and $\mathbf{P}'$, where $S=GETCHA(\mathbf{P})$ and $S'=GETCHA(\mathbf{P}')$, if $\mathbf{P}\vert _{ S}= \mathbf{P}'\vert _{ S'}$, then $F(\mathbf{P}) = F(\mathbf{P}')$. This inter-profile condition may be viewed as a weakening of the independence of irrelevant alternatives.
\end{remark}

\begin{proposition} Split Cycle satisfies ISDA.
\end{proposition}

\begin{proof} Suppose $x\in X(\mathbf{P})\setminus GETCHA(\mathbf{P})$. It follows that $GETCHA(\mathbf{P})=GETCHA(\mathbf{P}_{-x})$.\footnote{\label{GETCHAISDA}Clearly  $GETCHA(\mathbf{P})$ is  $\to_{\mathbf{P}_{-x}}$-dominant, so $GETCHA(\mathbf{P}_{-x})\subseteq GETCHA(\mathbf{P})$ by Definition \ref{GETCHA}. To see that $GETCHA(\mathbf{P}_{-x})$ is $\to_{\mathbf{P}}$-dominant, consider an $a\in GETCHA(\mathbf{P}_{-x})$ and $b\in X(\mathbf{P})\setminus GETCHA(\mathbf{P}_{-x})$. If $b\neq x$, then $b\in X(\mathbf{P}_{-x})\setminus GETCHA(\mathbf{P}_{-x})$ and hence $a\to_{\mathbf{P}_{-x}} b$ because $GETCHA(\mathbf{P}_{-x})$ is $\to_{\mathbf{P}_{-x}}$-dominant, which implies $a\to_{\mathbf{P}} b$. If $b=x$, then since $GETCHA(\mathbf{P}_{-x})\subseteq GETCHA(\mathbf{P})$ and $x\in X(\mathbf{P})\setminus GETCHA(\mathbf{P})$, again we have $a\to_{\mathbf{P}} b$. Thus, $GETCHA(\mathbf{P}_{-x})$ is $\to_{\mathbf{P}}$-dominant, so $GETCHA(\mathbf{P})\subseteq GETCHA(\mathbf{P}_{-x})$ by Definition \ref{GETCHA}.} Toward showing that $SC(\mathbf{P})=SC(\mathbf{P}_{-x})$, let $\mathbf{Q},\mathbf{Q}'\in \{\mathbf{P},\mathbf{P}_{-x}\}$ with $\mathbf{Q}\neq\mathbf{Q}'$. Suppose $y\in SC(\mathbf{Q})$. We will show $y\in SC(\mathbf{Q}')$. For any $z\in X(\mathbf{P}_{-x})$, we have \begin{equation} Margin_\mathbf{P}(z,y)=Margin_{\mathbf{P}_{-x}}(z,y).\label{EqualMarginsEq}\end{equation} Hence if $Margin_\mathbf{Q}(z,y) = 0$, then $Margin_{\mathbf{Q}'}(z,y) = 0$, so $z$ does not defeat $y$ in $\mathbf{Q}'$. Suppose instead that  $Margin_\mathbf{Q}(z,y) > 0$. Since $y\in SC(\mathbf{Q})$, $z$ does not defeat $y$ in $\mathbf{Q}$, so we have 
\begin{equation}Margin_\mathbf{Q}(z,y)< Cycle\#_\mathbf{Q}(z,y).\label{MarginLessCycle}
\end{equation}
Now we claim that \begin{equation}Cycle\#_\mathbf{P}(z,y)=Cycle\#_{\mathbf{P}_{-x}}(z,y).\label{NoCycleChange}\end{equation}  Since $y\in SC(\mathbf{Q})$, $y\in GETCHA(\mathbf{Q})$ by  Proposition \ref{SmithProp}. Then from $Margin_\mathbf{Q}(z,y) > 0$, it follows that $z\in GETCHA(\mathbf{Q})$. Since $x \in X(\mathbf{P})\setminus GETCHA(\mathbf{P})$ and $z\in GETCHA(\mathbf{Q})=GETCHA(\mathbf{P})$, there is no simple cycle in $\mathcal{M}(\mathbf{P})$ of the form $z\to y\to w_1\to \dots \to w_n\to z$ with $x\in \{w_1,\dots,w_n\}$, since there is no path from a candidate outside  $GETCHA(\mathbf{P})$, like $x$, to a candidate inside $GETCHA(\mathbf{P})$, like $z$. This establishes (\ref{NoCycleChange}). Then together (\ref{EqualMarginsEq}), (\ref{MarginLessCycle}), and (\ref{NoCycleChange}) entail $Margin_{\mathbf{Q}'}(z,y)< Cycle\#_{\mathbf{Q}'}(z,y)$. Hence $z$ does not defeat $y$ in $\mathbf{Q}'$. Finally, if $\mathbf{Q}'=\mathbf{P}$, then $x$ does not defeat $y$ in $\mathbf{Q}'$, since $y\in GETCHA(\mathbf{Q}')$ while $x\not\in GETCHA(\mathbf{Q}')$. Thus, no candidate defeats $y$ in $\mathbf{Q}'$, so $y\in SC(\mathbf{Q}')$.\end{proof}

\subsubsection{Independence of Clones}\label{ClonesSection}

In Section \ref{SpoilerSection}, we informally discussed the anti-vote-splitting axiom of \textit{independence of clones} (\citealt{Tideman1987}). Recall that a set $C$ of two or more candidates is a set of clones if no candidate outside of $C$ appears in between two candidates from $C$ on any voter's ballot.

\begin{definition}\label{CloneDef} Given a profile $\mathbf{P}$, a set $C\subseteq X(\mathbf{P})$ is a \textit{set of clones in $\mathbf{P}$} if  $2\leq \vert C\vert  < \vert X(\mathbf{P})\vert $ and for all $c,c'\in C$, $x\in X(\mathbf{P})\setminus C$, and $i\in V(\mathbf{P})$, if $c\mathbf{P}_ix$ then $c'\mathbf{P}_ix$, and if $x\mathbf{P}_ic$ then $x\mathbf{P}_ic'$. 
\end{definition}

The independence of clones criterion states that (i) removing a clone from a profile does not change which non-clones belong to the winning set, and (ii) a clone wins in a profile if and only if after removing that clone from the profile, one of the other clones wins in the resulting profile. 

\begin{definition}\label{IndependenceOfClones} A voting method $F$ is such that \textit{non-clone choice is independent of clones} if for every $\mathbf{P}\in\mathrm{dom}(F)$, set $C$ of clones in $\mathbf{P}$, $c\in C$, and $a\in X(\mathbf{P})\setminus C$, \[\mbox{$a\in F(\mathbf{P})$ if and only if $a\in F(\mathbf{P}_{-c})$.}\] $F$ is such that \textit{clone choice is independent of clones} if for every $\mathbf{P}\in\mathrm{dom}(F)$, set $C$ of clones in $\mathbf{P}$, and $c\in C$,
\[C\cap F(\mathbf{P})\neq\varnothing\mbox{ if and only if }C\setminus\{c\}\cap F(\mathbf{P}_{-c})\neq\varnothing.\]
Finally, $F$ satisfies \textit{independence of clones} if $F$ is such that non-clone choice is independent of clones and clone choice is independent of clones.\end{definition}

We prove the following in Appendix \ref{ClonesAppendix}.

\begin{theorem} Split Cycle satisfies independence of clones.
\end{theorem}

\begin{remark}\label{RPnote}Tideman \citeyearpar{Tideman1987} shows that the version of Ranked Pairs defined in Section \ref{Comparison} and Appendix~\ref{RankedPairsAppendix} satisfies the condition of independence of clones for all profiles $\mathbf{P}$ such that for all $a,b,x,y\in X(\mathbf{P})$, $Margin_\mathbf{P}(a,b)= 0$ only if $a=b$, and $Margin_\mathbf{P}(a,b)=Margin_\mathbf{P}(x,y)$ only if (i) $a=x$ or $a$ and $x$ belong to a set of clones, and (ii) $b=y$ or $b$ and $y$ belong to a set of clones. Zavist and Tideman \citeyearpar{ZavistTideman1989} show that the same version of Ranked Pairs does not satisfy independence of clones for all profiles. They propose a modified version of Ranked Pairs that satisfies independence of clones at the expense of violating~anonymity. However, as suggested by Tideman (p.c.), one can obtain an anonymous and fully clone-independent version of Ranked Pairs for linear profiles by declaring a candidate $x$ a winner in a linear profile $\mathbf{P}$ if there exists a voter $i\in V(\mathbf{P})$ such that the Zavist and Tideman version of Ranked Pairs declares $x$ a winner in $\mathbf{P}$ when $i$ is the designated voter used to generate their tiebreaking ranking of candidates (TBRC).\footnote{This formulation relies on the assumption that voters submit linear orders of the candidates. Zavist and Tideman allow ties in voter ballots and use a randomizing device to generate the linear TBRC from the designated voter's ballot, in case it contains ties. Thus, Zavist and Tideman define a probabilistic voting method.}
\end{remark}

\subsection{Resoluteness Criteria}\label{ResolutenessSection}

In this section, we discuss criteria concerning the ability of a voting method to narrow down the set of winners. Of course, any anonymous and neutral voting method will select \textit{all} candidates as winners in a profile in which all candidates are tied. But such profiles are highly unlikely. To rule out such cases, we can consider uniquely-weighted profiles as in Definition \ref{MarginGraphDef}. Still, some highly unlikely uniquely-weighted profiles will produce large winning sets for Split Cycle, as shown by part \ref{Exclude2b} of the following.

\begin{proposition}\label{Exclude2} $\,$
\begin{enumerate}
\item\label{Exclude2a} For any uniquely-weighted profile $\mathbf{P}$  such that $\vert X(\mathbf{P})\vert \geq 3$, we have $\vert SC(\mathbf{P})\vert \leq \vert X(\mathbf{P})\vert -2$.
\item\label{Exclude2b} For any $n\geq 4$, there is a uniquely-weighted profile $\mathbf{P}$ with $\vert X(\mathbf{P})\vert =n$ and $\vert SC(\mathbf{P})\vert = \vert X(\mathbf{P})\vert -2$.
\end{enumerate}
\end{proposition}
\begin{proof} For part \ref{Exclude2a}, pick $a,b\in X(\mathbf{P})$ such that $Margin_\mathbf{P}(a,b)$ is the largest margin of any edge in $\mathcal{M}(\mathbf{P})$, which implies that $Margin_\mathbf{P}(a,b)>0$. Then clearly $a$ defeats $b$ in $\mathbf{P}$. Now pick $c,d\in X(\mathbf{P})$ with $d\neq b$ such that $Margin_\mathbf{P}(c,d)$ is the largest margin in $\mathcal{M}(\mathbf{P})$ of any edge not going to $b$. Suppose for contradiction that $c$ does not defeat $d$ in $\mathbf{P}$. Then there is a simple cycle $\rho$ containing $c$ and $d$ such that $Margin_\mathbf{P}(c,d)$ is strictly less than the other margins along the cycle, at least one of which is a margin of an edge not going to $b$. But this contradicts the fact that $Margin_\mathbf{P}(c,d)$ is the largest margin in $\mathcal{M}(\mathbf{P})$ of any edge not going to $b$. Thus, $c$ defeats $d$. Hence $SC(\mathbf{P})$ contains neither $b$ nor $d$, so $\vert SC(\mathbf{P})\vert \leq \vert X(\mathbf{P})\vert -2$.

For part \ref{Exclude2b}, consider the sequence of margin graphs of the following form (for a definition and code to generate the sequence, see \href{https://github.com/epacuit/splitcycle}{https://github.com/epacuit/splitcycle}): 
\begin{center}
\begin{minipage}{1.5in}

\begin{tikzpicture}

\node[circle,draw,minimum width=0.1in] at (0,0) (a) {$$}; 
\node[circle,draw,minimum width=0.1in] at (2.5,0) (b) {$$}; 
\node[circle,draw,fill=medgreen!50,minimum width=0.1in] at (1.25,1.25) (c) {$$}; 
\node[circle,draw,fill=medgreen!50,minimum width=0.1in] at (1.25,2.5) (d) {$$}; 

\path[->,draw,thick, bend left=15] (c) to node[fill=white] {$10$} (b);
\path[->,draw,thick] (b) to node[fill=white] {$12$} (a);
\path[->,draw,thick, bend left=15] (a) to node[fill=white] {$4$} (c);
\path[->,draw,thick] (d) to node[fill=white] {$2$} (c);
\path[->,draw,thick, bend left] (a) to node[fill=white] {$6$} (d);
\path[->,draw,thick, bend left] (d) to node[fill=white] {$8$} (b);

\end{tikzpicture}
\end{minipage}\begin{minipage}{1.5in}

\begin{tikzpicture}

\node[circle,draw, minimum width=0.1in] at (0,0) (a) {$$}; 
\node[circle,draw,minimum width=0.1in] at (2.5,0) (b) {$$}; 
\node[circle,draw,fill=medgreen!50,minimum width=0.1in] at (1.25,1.25) (c) {$$}; 
\node[circle,draw,fill=medgreen!50,minimum width=0.1in] at (1.25,2.5) (d) {$$}; 
\node[circle,draw,fill=medgreen!50,minimum width=0.1in] at (1.25,3.75) (e) {$$}; 

\path[->,draw,thick,bend left=15] (c) to node[fill=white] {$18$} (b);
\path[->,draw,thick] (b) to node[fill=white] {$20$} (a);
\path[->,draw,thick,bend left=15] (a) to node[fill=white] {$8$} (c);
\path[->,draw,thick] (d) to node[fill=white] {$2$} (c);
\path[->,draw,thick] (e) to node[fill=white] {$6$} (d);
\path[->,draw,thick, bend left=20] (a) to node[fill=white] {$10$} (d);
\path[->,draw,thick, bend left=20] (d) to node[fill=white] {$16$} (b);
\path[->,draw,thick, bend left] (a) to node[fill=white] {$12$} (e);
\path[->,draw,thick, bend left] (e) to node[fill=white] {$14$} (b);
\path[->,draw,thick, bend left=42] (e) to node[fill=white] {$4$} (c);

\end{tikzpicture}
\end{minipage}\begin{minipage}{2in}

\begin{tikzpicture}

\node[circle,draw, minimum width=0.1in] at (0,0) (a) {$$}; 
\node[circle,draw,minimum width=0.1in] at (2.5,0) (b) {$$}; 
\node[circle,fill=medgreen!50,draw,minimum width=0.1in] at (1.25,1.25) (c) {$$}; 
\node[circle,fill=medgreen!50,draw,minimum width=0.1in] at (1.25,2.5) (d) {$$}; 
\node[circle,fill=medgreen!50,draw,minimum width=0.1in] at (1.25,3.75) (e) {$$}; 
\node[circle,fill=medgreen!50,draw,minimum width=0.1in] at (1.25,5) (f) {$$}; 

\path[->,draw,thick,bend left=15] (c) to node[fill=white] {$28$} (b);
\path[->,draw,thick] (b) to node[fill=white] {$30$} (a);
\path[->,draw,thick,bend left=15] (a) to node[fill=white] {$14$} (c);
\path[->,draw,thick] (d) to node[fill=white] {$2$} (c);
\path[->,draw,thick] (e) to node[fill=white] {$8$} (d);
\path[->,draw,thick, bend left=20] (a) to node[fill=white] {$16$} (d);
\path[->,draw,thick, bend left=20] (d) to node[fill=white] {$26$} (b);
\path[->,draw,thick, bend left] (a) to node[fill=white] {$18$} (e);
\path[->,draw,thick, bend left] (e) to node[fill=white] {$24$} (b);
\path[->,draw,thick, bend left=43] (e) to node[fill=white] {$4$} (c);
\path[->,draw,thick] (f) to node[fill=white] {$12$} (e);
\path[->,draw,thick, bend left=44] (a) to node[fill=white] {$20$} (f);
\path[->,draw,thick, bend left=44] (f) to node[fill=white] {$22$} (b);
\path[->,draw,thick, bend left=57.5] (f) to node[fill=white] {$6$} (c);
\path[->,draw,thick, bend left=43] (f) to node[fill=white] {$10$} (d);

\end{tikzpicture}
\end{minipage}\begin{minipage}{0in}$\cdots$
\end{minipage}
\end{center}
By Theorem \ref{DebordThm}, there are (linear) profiles realizing each margin graph in the sequence. In each margin graph, the arrow from the bottom right candidate to the bottom left candidate is a defeat; all arrows pointing to the bottom right candidate are defeats; but no other arrows are defeats, since each has the weakest margin in a cycle. Thus, all candidates are undefeated except the bottom two candidates.\end{proof}

Fortunately, the margin graphs used in the proof of Proposition \ref{Exclude2}.\ref{Exclude2b} are realized by an extremely small proportion of profiles, as we will see with some data in Section \ref{AsymResolve}  (Table \ref{NarrowingTable}). As worst-case winning set sizes are not the best measure of the general ability of a voting method to narrow down the set of winners, we will consider alternative resoluteness criteria in the next two subsections.

\subsubsection{Rejectability}\label{Rejectability}
 
The next criterion we propose concerns winnowing a set of winners down to a single winner. The \textit{rejectability} criterion states that if in a profile $\mathbf{P}$, candidate $x$ is \textit{among} the winners, then we should be able to make $x$ the \textit{unique} winner in a profile $\mathbf{P}^+$ obtained from $\mathbf{P}$ by adding voters who sufficiently strengthen the rejection of other candidates, i.e., sufficiently increase what were already non-negative margins against other candidates, so as to defeat the others (recall our idea in Section \ref{SplitCycleSection} that incoherence does not raise the threshold for defeat \textit{infinitely}). Thus, if candidate $a$ is majority preferred to $b$ in $\mathbf{P}$, then this still holds in $\mathbf{P}^+$ with a margin that is at least as large and possibly larger than in $\mathbf{P}$. No majority preferences are reversed from $\mathbf{P}$ to $\mathbf{P}^+$, for if we were to allow that, then we could simply make $x$ the Condorcet winner in $\mathbf{P}^+$, trivializing the criterion. 
 
 \begin{definition} A voting method $F$ satisfies \textit{rejectability} if for any $\mathbf{P}\in\mathrm{dom}(F)$ such that $\vert  F(\mathbf{P}) \vert  >1$ and $x\in F(\mathbf{P})$, there is a profile $\mathbf{P}^+\in\mathrm{dom}(F)$ with $X(\mathbf{P})=X(\mathbf{P}^+)$ and $V(\mathbf{P})\subseteq V(\mathbf{P}^+)$ such that for all $a,b\in X(\mathbf{P})$, if $Margin_{\mathbf{P}}(a,b)> 0$, then $Margin_{\mathbf{P}^+}(a,b)\geq Margin_{\mathbf{P}}(a,b)$, and $F(\mathbf{P}^+)=\{x\}$.
 \end{definition}
 
\noindent Thus, if a method \textit{fails} rejectability, then for some $\mathbf{P}$ and $x\in F(\mathbf{P})$, no matter how extremely we turn majority preferences against other candidates into enormous landslides, we cannot make $x$ the unique winner. 

Rejectability is a strong criterion insofar as it rules out all irresolute C1 voting methods (as does the resolvability criterion of Section \ref{Resolvability}). Recall that a voting method $F$ is C1 (\citealt{Fishburn1977}) if for any profiles $\mathbf{P}$ and $\mathbf{P}'$, if their majority graphs (Definition \ref{MarginGraphDef}) are the same---$M(\mathbf{P})=M(\mathbf{P}')$---then their winners are also the same---$F(\mathbf{P})=F(\mathbf{P}')$. Copeland, GETCHA/GOCHA, and Uncovered Set are all C1.
 
 \begin{proposition}\label{GETCHA_GOCHA_Rejectability} No anonymous and neutral C1 voting method (whose domain contains all linear profiles with three candidates) satisfies rejectability. 
 \end{proposition}
 
\begin{proof} Given a profile $\mathbf{P}$ with $X(\mathbf{P})=\{a,b,c\}$ and whose margin graph contains the cycle $a\to b\to c\to a$, no matter the margins, an anonymous and neutral C1 method $F$ must have $F(\mathbf{P})=\{a,b,c\}$; hence we can never increase any margins in such a way that one candidate becomes the unique winner.
\end{proof}
An example of a non-C1 method violating rejectability is the Weighted Covering method (\citealt{Dutta1999}, \citealt{Fernandez2018}), according to which $x\in WC(\mathbf{P})$ if there is no $y\in X(\mathbf{P})$ such that $Margin_\mathbf{P}(y,x)>0$ and for all $z\in X(\mathbf{P})$, $Margin_\mathbf{P}(y,z)\geq Margin_\mathbf{P}(x,z)$. Weighted Covering also selects all candidates in the profile $\mathbf{P}$ in the proof of Proposition \ref{GETCHA_GOCHA_Rejectability}.

In our proof that Split Cycle satisfies rejectability, we use the following lemma.

\begin{lemma}\label{OverwhelmingMaj} Split Cycle satisfies the \textit{overwhelming majority}\footnote{This is the terminology from \citealt{Myerson1995}. Cf.~Smith's \citeyearpar{Smith1973} ``Archimedean property'' and Young's \citeyearpar{Young1975} ``continuity.''} criterion: for all profiles $\mathbf{P}$ and $\mathbf{P}'$ with $X(\mathbf{P})=X(\mathbf{P}')$ and $V(\mathbf{P})\cap V(\mathbf{P}')=\varnothing$, there is an $n\in\mathbb{N}$ such that for all $m\in\mathbb{N}$ with $m\geq n$, we have $SC(\mathbf{P}+m\mathbf{P}')\subseteq SC(\mathbf{P}')$, where $m\mathbf{P}'=\gamma_1(\mathbf{P}')+\dots +\gamma_m(\mathbf{P}')$ with $\gamma_1(\mathbf{P}'),\dots,\gamma_m(\mathbf{P}')$ being copies of $\mathbf{P}'$ with pairwise disjoint sets of voters (recall Definition \ref{DisjointUnion}).
\end{lemma}

\begin{proof} Let $n = 2\vert V(\mathbf{P})\vert +1$. To show that $SC(\mathbf{P}+m\mathbf{P}')\subseteq SC(\mathbf{P}')$, it suffices to show that for any $a,b\in X(\mathbf{P}')$, if $a$ defeats $b$ in $\mathbf{P}'$, then $a$ defeats $b$ in $\mathbf{P}+m\mathbf{P}'$. Assume $a$ defeats $b$ in $\mathbf{P}'$, so $Margin_{\mathbf{P}'}(a,b)> Cycle\#_{\mathbf{P}'}(a,b)$, so $Margin_{\mathbf{P}'}(a,b)- Cycle\#_{\mathbf{P}'}(a,b)\geq 1$.  Then for all $m\geq n$, since \[\mbox{$Margin_{m\mathbf{P}'}(a,b)=m\times Margin_{\mathbf{P}'}(a,b)$ and $Cycle\#_{m\mathbf{P}'}(a,b)= m\times Cycle\#_{\mathbf{P}'}(a,b)$,}\] we have $Margin_{m\mathbf{P}'}(a,b)- Cycle\#_{m\mathbf{P}'}(a,b)\geq m\geq n= 2\vert V(\mathbf{P})\vert +1$. Also note that \[\mbox{$Margin_{\mathbf{P}+m\mathbf{P}'}(a,b) \geq Margin_{m\mathbf{P}'}(a,b) -\vert V(\mathbf{P})\vert $ and 
$Cycle\#_{\mathbf{P}+m\mathbf{P}'}(a,b) \leq  Cycle\#_{m\mathbf{P}'}(a,b)+\vert V(\mathbf{P})\vert $.}\]
 It follows that $Margin_{\mathbf{P}+m\mathbf{P}'}(a,b)> Cycle\#_{\mathbf{P}+m\mathbf{P}'}(a,b)$, so $a$ defeats $b$ in $\mathbf{P}+m\mathbf{P}'$.\end{proof}
 
 \begin{proposition}\label{RejectProp} Split Cycle satisfies rejectability.
 \end{proposition}
 
 \begin{proof} We claim that to establish rejectability, it suffices to show that for any profile $\mathbf{P}$ such that $\vert  SC(\mathbf{P}) \vert  >1$ and $x\in SC(\mathbf{P})$, there is a profile $\mathbf{P}'$ with $X(\mathbf{P})=X(\mathbf{P}')$ such that for all $a,b\in X(\mathbf{P})$, if $Margin_{\mathbf{P}}(a,b)> 0$, then $Margin_{\mathbf{P}'}(a,b)\geq Margin_{\mathbf{P}}(a,b)$, and $SC(\mathbf{P}')=\{x\}$. For then by Lemma \ref{OverwhelmingMaj}, there is an $m\in\mathbb{N}$ such that $SC(\mathbf{P}+m\mathbf{P}')=\{x\}$, and for all $a,b\in X(\mathbf{P})$, if $Margin_{\mathbf{P}}(a,b)> 0$, then $Margin_{\mathbf{P}+m\mathbf{P}'}(a,b)\geq Margin_{\mathbf{P}}(a,b)$. As $V(\mathbf{P})\subseteq V(\mathbf{P}+m\mathbf{P}')$, we may take $\mathbf{P}^+=\mathbf{P}+m\mathbf{P}'$ for rejectability.
 
 Suppose $\vert  SC(\mathbf{P}) \vert  >1$ and $x\in SC(\mathbf{P})$. We show how to modify the margin graph $\mathcal{M}(\mathbf{P})$ to a margin graph $\mathcal{M}'$ on $X(\mathbf{P})$ such that (i) all edges between nodes are preserved from $\mathcal{M}(\mathbf{P})$ to $\mathcal{M}'$, (ii) no weights on edges decrease from $\mathcal{M}(\mathbf{P})$ to $\mathcal{M}'$,  and (iii) $SC(\mathcal{M}')=\{x\}$ (recall Remark \ref{SCofMarginGraph}). Then Debord's Theorem yields a profile $\mathbf{P}'$ whose margin graph is $\mathcal{M}'$. By (i)-(ii), we have that for all $a,b\in X(\mathbf{P})$, if $Margin_{\mathbf{P}}(a,b)> 0$, then $Margin_{\mathbf{P}'}(a,b)\geq Margin_{\mathbf{P}}(a,b)$. By (iii), $SC(\mathbf{P}')=\{x\}$.
 
 Let the set of edges in $\mathcal{M}'$ be the set of all edges in $\mathcal{M}(\mathbf{P})$ plus an edge from $x$ to any $y$ such that $Margin_\mathbf{P}(x,y)=0$. Let $n$ be the largest margin in $\mathcal{M}(\mathbf{P})$. Each edge $(a,b)$ in $\mathcal{M}'$ has weight either $n+1$ or $n+3$ according to the following rules (we use $Margin_{\mathcal{M}'}$ and $Cycle\#_{\mathcal{M}'}$ with their obvious meanings):
 \begin{enumerate}
 \item if the edge $(a,b)$ occurs on a shortest simple path\footnote{A \textit{simple path} in a graph is a sequence $\langle x_1,\dots,x_n\rangle$ of \textit{distinct} nodes with $x_i\to x_{i+1}$ for each $i\in \{1,\dots,n-1\}$. The \textit{length} of a path is the number of nodes in the path minus one.} from $x$ to $b$ in $\mathcal{M}'$, set $Margin_{\mathcal{M}'}(a,b)=n+3$;
 \item otherwise, set $Margin_{\mathcal{M}'}(a,b)=n+1$.
 \end{enumerate}
We now claim that every $y\in X(\mathbf{P})\setminus\{x\}$ is defeated in $\mathcal{M}'$. 
 
 Case 1: $Margin_\mathbf{P}(x,y)\geq 0$. Then $Margin_{\mathcal{M}'}(x,y)=n+3$ by rule 1. Moreover, for any simple cycle $\rho$ of the form $x\to y \to z_1\to\dots\to z_k\to x$ in $\mathcal{M}'$, we have $Margin_{\mathbf{P}}(z_k,x)>0$ by the construction of $\mathcal{M}'$ from $\mathcal{M}(\mathbf{P})$ and hence $Margin_{\mathcal{M}'}(z_k,x)=n+1$ by rule 2, so $Split\#(\rho)=n+1$. Hence $Cycle\#_{\mathcal{M}'}(x,y)=n+1$. Thus, $Margin_{\mathcal{M}'}(x,y)>Cycle\#_{\mathcal{M}'}(x,y)$, so $x$ defeats $y$ in $\mathcal{M}'$.
 
 Case 2: $Margin_\mathbf{P}(y,x)> 0$. Then since $x\in SC(\mathbf{P})$, it follows by Lemma \ref{OnlySomeCycles} that there is a simple cycle of the form $y\to x\to z_1\to\dots\to z_k\to y$ in $\mathbf{P}$ where $x\to z_1\to\dots\to z_k\to y$ is a shortest simple path from $x$ to $y$. Hence $Margin_{\mathcal{M}'}(z_k,y)=n+3$ by rule 1. We claim that $z_k$ defeats $y$ in $\mathcal{M}'$. If there is no simple cycle of the form $z_k\to w_1\to \dots \to w_\ell$ with $w_1=y$ and $w_\ell = z_k$ in $\mathcal{M}'$, then $z_k$ defeats $y$ in $\mathcal{M}'$. If there is such a simple cycle $\rho$, then we claim that one of the edges  $w_i\to w_{i+1}$ in $\rho$ has weight $n+1$. If there is no simple path from $x$ to any of $w_2,\dots,w_\ell$, this follows from rule 2 above. So suppose there is a simple path from $x$ to one of $w_2,\dots,w_\ell$. Then there is a $w_i$ such that (i) the shortest path $p$ from $x$ to $w_i$ is no longer than the shortest path from $x$ to any $w_j$. This setup is shown in Figure \ref{RejectDiagram}. Now we claim that the edge $w_{i-1}\to w_i$ in $\rho$ has weight $n+1$; for it to have weight $n+3$, the edge $w_{i-1}\to w_i$ must occur on a shortest path from $x$ to $w_i$, which is impossible. For suppose $p'$ is a path from $x$ to $w_i$ including the edge $w_{i-1}\to w_i$. By (i), the initial segment of $p'$ from $x$ to $w_{i-1}$ has length at least that of $p$, by our choice of $p$; so the length of $p'$ is at least the length of $p$ plus one; hence $p'$ is not a shortest path from $x$ to $w_i$. Thus, we have proved that one of the edges  $w_i\to w_{i+1}$ has weight $n+1$. Thus, $Split\#(\rho)=n+1$. It follows that $Cycle\#_{\mathcal{M}'}(z_k,y)=n+1$, which with $Margin_{\mathcal{M}'}(z_k,y)=n+3$  implies that $z_k$ defeats $y$ in $\mathcal{M}'$.\end{proof}
 
 \begin{corollary}\label{BPreject} Beat Path and Ranked Pairs satisfy rejectability.
 \end{corollary}
 \begin{proof} Let $F\in \{BP,RP\}$ and $\mathbf{P}$ be a profile such that $\vert F(\mathbf{P})\vert >1$ and $x\in F(\mathbf{P})$. Then by Lemmas \ref{SubsetLem} and \ref{SubsetLem2}, $\vert SC(\mathbf{P})\vert >1$ and $x\in SC(\mathbf{P})$. Hence by Proposition \ref{RejectProp}, there is a $\mathbf{P}'$ as in the definition of rejectability such that $SC(\mathbf{P}')=\{x\}$, which implies $F(\mathbf{P}')=\{x\}$ given Lemmas \ref{SubsetLem} and \ref{SubsetLem2} and $F(\mathbf{P}')\neq\varnothing$.\end{proof}

 \begin{figure}
 \begin{center}
  \begin{minipage}{2in}\begin{tikzpicture}[semithick]
  
  \draw[ 
        decoration={markings, mark=at position 0.135 with {\arrow{<}}},
        postaction={decorate}
        ]
        [ 
        decoration={markings, mark=at position 0.02 with {\arrow{<}}},
        postaction={decorate}
        ]
        [ 
        decoration={markings, mark=at position 0.4 with {\arrow{<}}},
        postaction={decorate}
        ]
            [ 
        decoration={markings, mark=at position 0.5175 with {\arrow{<}}},
        postaction={decorate}
        ]
            [ 
        decoration={markings, mark=at position 0.635 with {\arrow{<}}},
        postaction={decorate}
        ]
  
        (0,0) circle (1.5);

\node at (0,2.25) (x) {$x$};

\node at (-1.5,0) {$\bullet$};
\node at (-1.75,0) (y) {$y$}; 
\node at (-1.1,-1) {$\bullet$};
\node at (-.3,-1) (wl)  {$w_\ell=z_k$};

\node at (-.8,-.5) {$n+3$};

\node at (-1.1,1) {$\bullet$};
\node at (-.7,1) (w2) {$w_2$}; 

\node at (1.1,1) (wi1bullet) {$\bullet$};
\node at (.5,1) (wi1) {$w_{i-1}$}; 

\node at (.8,.5) {$n+1$};

\node at (1.5,0) (wibullet) {$\bullet$};
\node at (1.8,0) (wi) {$w_i$}; 

\path[->,draw,thick, bend left] (y) to  (x);
\path[->,draw,thick, dotted, bend left] (x) to  (wi);

  \end{tikzpicture}
\end{minipage}
\end{center}
\caption{diagram for the proof of Proposition \ref{RejectProp}.}\label{RejectDiagram}
\end{figure}
 
\begin{example} If we pick any candidate $x$ in the majority graph shown on the left below, the proof of Proposition \ref{RejectProp} give us an algorithm to weight the edges of the majority graph such that in the resulting margin graph $x$ is the unique Split Cycle winner. For example, we can make $a$ the unique winner with the weighting on the middle graph and $d$ the unique winner with the weighting on the right graph.
  \begin{center}
 \begin{minipage}{2in}\begin{tikzpicture}

\node[circle,draw, minimum width=0.25in] at (0,0) (b) {$b$}; 
\node[circle,draw,minimum width=0.25in] at (3,0) (a) {$a$}; 
\node[circle,draw,minimum width=0.25in] at (1.5,1.5) (c) {$c$}; 
\node[circle,draw,minimum width=0.25in] at (1.5,-1.5) (d) {$d$}; 

\path[->,draw,thick] (b) to (a);
\path[->,draw,thick] (d) to (c);
\path[->,draw,thick] (c) to  (b);
\path[->,draw,thick] (a) to  (c);
\path[->,draw,thick] (d) to  (a);
\path[->,draw,thick] (b) to  (d);

\end{tikzpicture}
\end{minipage}\begin{minipage}{2in}\begin{tikzpicture}

\node[circle,draw, minimum width=0.25in] at (0,0) (b) {$b$}; 
\node[circle,draw,minimum width=0.25in] at (3,0) (a) {$a$}; 
\node[circle,draw,minimum width=0.25in] at (1.5,1.5) (c) {$c$}; 
\node[circle,draw,minimum width=0.25in] at (1.5,-1.5) (d) {$d$}; 

\path[->,draw,thick] (b) to (a);
\path[->,draw,thick] (d) to (c);
\path[->,draw,thick] (c) to node[fill=white] {$3$} (b);
\path[->,draw,thick] (a) to node[fill=white] {$3$} (c);
\path[->,draw,thick] (d) to node[fill=white] {$1$} (a);
\path[->,draw,thick] (b) to node[fill=white] {$3$} (d);

\node[fill=white] at (1.5,.5)  {$1$}; 
\node[fill=white] at (2,0)  {$1$}; 

  \end{tikzpicture}
\end{minipage} \begin{minipage}{2in}\begin{tikzpicture}

\node[circle,draw, minimum width=0.25in] at (0,0) (b) {$b$}; 
\node[circle,draw,minimum width=0.25in] at (3,0) (a) {$a$}; 
\node[circle,draw,minimum width=0.25in] at (1.5,1.5) (c) {$c$}; 
\node[circle,draw,minimum width=0.25in] at (1.5,-1.5) (d) {$d$}; 

\path[->,draw,thick] (b) to (a);
\path[->,draw,thick] (d) to (c);
\path[->,draw,thick] (c) to node[fill=white] {$3$} (b);
\path[->,draw,thick] (a) to node[fill=white] {$1$} (c);
\path[->,draw,thick] (d) to node[fill=white] {$3$} (a);
\path[->,draw,thick] (b) to node[fill=white] {$1$} (d);

\node[fill=white] at (1.5,.5)  {$3$}; 
\node[fill=white] at (2,0)  {$1$}; 

  \end{tikzpicture}
\end{minipage}
\end{center}

\end{example}

In fact, from the proof of Proposition \ref{RejectProp} we can extract a proof of the following proposition about when it is possible, starting from an arbitrary graph, to turn the graph into a margin graph in which a given candidate is a (unique) winner for Split Cycle.

\begin{proposition} For any asymmetric directed graph $\mathcal{G}=(G,\to)$ and $a\in G$, the following are equivalent:
\begin{enumerate}
\item there is a margin graph $\mathcal{M}$ based on $\mathcal{G}$ such that $a\in SC(\mathcal{M})$ (recall Remark \ref{SCofMarginGraph});
\item there is a margin graph $\mathcal{M}$ based on $\mathcal{G}$ such that $SC(\mathcal{M})=\{a\}$;
\item for all $x\in G\setminus\{a\}$, if $x\to a$, then there is a simple cycle of the form $x\to a\to y_1\to \dots \to y_n \to x$ in~$\mathcal{G}$.
\end{enumerate}
\end{proposition}

\subsubsection{Resolvability}\label{Resolvability}

\setcounter{secnumdepth}{4}

Like the rejectability criterion of Section \ref{Rejectability}, the criteria considered in this section concern winnowing a set of winners down to a unique winner. 

\paragraph{Single-voter resolvability} The first criterion, \textit{single-voter resolvability}, says that any tied winner can be made the unique winner by adding \textit{just one new voter}. We see no justification for requiring that one voter is always sufficient, and as far as we know, no arguments for the normative necessity of this criterion are given in the literature. Tideman \citeyearpar{Tideman1987} uses single-voter resolvability to rule out the GOCHA method, but this can be accomplished by rejectability instead. Indeed, we suspect that some intuitions about winnowing sets of winners to a unique winner are better captured by rejectability than by single-voter~resolvability. 

\begin{definition}\label{TidemanResolve} Given a voting method $F$ and $\mathscr{D}\subseteq\mathrm{dom}(F)$, we say that $F$ satisfies \textit{single-voter resolvability with respect to $\mathscr{D}$} if for any $\mathbf{P}\in\mathscr{D}$, if $\vert F(\mathbf{P})\vert >1$, then for any $x\in F(\mathbf{P})$, there is a profile $\mathbf{P}'$ with $V(\mathbf{P})\cap V(\mathbf{P}')=\varnothing$ and $\vert V(\mathbf{P}')\vert =1$ such that $F(\mathbf{P}+\mathbf{P}')=\{x\}$.
\end{definition}

\begin{proposition} Split Cycle does not satisfy single-voter resolvability even with respect to linear profiles.
\end{proposition}

\begin{proof} Recall the margin graph of the linear profile $\mathbf{P}$ from the proof of Proposition \ref{MinimaxBPstealers} showing that Minimax and Beat Path do not satisfy immunity to stealers:
\begin{center}
\begin{tikzpicture}

\node[circle,draw, minimum width=0.25in] at (0,0) (b) {$b$}; 
\node[circle,draw,minimum width=0.25in] at (3,0) (a) {$a$}; 
\node[circle,draw,minimum width=0.25in] at (1.5,1.5) (c) {$c$}; 
\node[circle,draw,minimum width=0.25in] at (1.5,-1.5) (d) {$d$}; 

\path[->,draw,thick] (b) to (a);
\path[->,draw,thick] (d) to (c);
\path[->,draw,thick] (c) to node[fill=white] {$3$} (b);
\path[->,draw,thick] (a) to node[fill=white] {$3$} (c);
\path[->,draw,thick] (a) to node[fill=white] {$1$} (d);
\path[->,draw,thick] (b) to node[fill=white] {$1$} (d);

\node[fill=white] at (1.5,.5)  {$3$}; 
\node[fill=white] at (2,0)  {$3$}; 

  \end{tikzpicture}\end{center}
Here $SC(\mathbf{P})=\{a,b,d\}$, but there is no one-voter profile $\mathbf{P}'$ with $SC(\mathbf{P}+\mathbf{P}')=\{a\}$ or ${SC(\mathbf{P}+\mathbf{P}')=\{b\}}$, since however each margin changes by at most one from $\mathbf{P}$ to $\mathbf{P}+\mathbf{P}'$, the margins of $a$ over $d$ and of $b$ over $d$ will still be the weakest in a cycle in $\mathcal{M}(\mathbf{P}+\mathbf{P}')$.\end{proof}
\noindent Below we will show a deep tension between single-voter resolvability and stability for winners.

Resolvability and rejectability can be related using the following additional criterion from \citealt{Smith1973}.

\begin{definition} A voting method $F$ satisfies \textit{homogeneity} if for any $\mathbf{P}\in\mathrm{dom}(F)$, if $\mathbf{P}^*$ is a copy of $\mathbf{P}$ with a disjoint set of voters (recall Definition \ref{DisjointUnion}), then $F(\mathbf{P})=F(\mathbf{P}+\mathbf{P}^*)$.
\end{definition}

\begin{lemma}\label{ResolveReject} If a voting method $F$ satisfies homogeneity and single-voter resolvability with respect to $\mathrm{dom}(F)$, then it satisfies rejectability.
\end{lemma}
\begin{proof} Let $\mathbf{P}\in\mathrm{dom}(F)$ be such that $\vert  F(\mathbf{P}) \vert  >1$ and $x\in F(\mathbf{P})$. Let $\mathbf{P}^*$ be a copy of $\mathbf{P}$ with a disjoint set of voters. Then by homogeneity, $F(\mathbf{P})=F(\mathbf{P}+\mathbf{P}^*)$. It follows by resolvability that there is a single voter profile $\mathbf{Q}$ such that $F(\mathbf{P}+\mathbf{P}^*+\mathbf{Q})=\{x\}$. Since for any $a,b\in X(\mathbf{P})$ with $Margin_\mathbf{P}(a,b)> 0$, we have $Margin_{\mathbf{P}+\mathbf{P}^*+\mathbf{Q}}(a,b)\geq Margin_\mathbf{P}(a,b)$, the profile $\mathbf{P}+\mathbf{P}^*+\mathbf{Q}$ is the desired profile $\mathbf{P}^+$ for rejectability.\end{proof}

\paragraph{Asymptotic resolvability}\label{AsymResolve} Another use of the term `resolvability'  (see \citealt[\S~4.2.1]{Schulze2011}) concerns the proportion of profiles with multiple winners as the number of voters goes to infinity. 

\begin{definition}\label{AsymptoticResolve} For $k\in \mathbb{N}$, a voting method $F$ satisfies \textit{asymptotic resolvability for $k$ candidates} if the proportion of profiles $\mathbf{P}\in\mathrm{dom}(F)$ with $\vert X(\mathbf{P})\vert =k$ and $\vert V(\mathbf{P})\vert =n$ for which $\vert F(\mathbf{P})\vert >1$ approaches 0 as $n$ approaches infinity.
\end{definition}
For comparison, recall the quasi-resoluteness condition from Section \ref{Comparison}, according to which $F$ picks a unique winner in any uniquely-weighted profile. Since the proportion of profiles that are uniquely weighted goes to $1$ as the number of voters goes to infinity, quasi-resoluteness implies  asymptotic resolvability. However, the converse implication does not hold. For example, the Borda  method (for a definition, see \citealt[\S~2.1]{Pacuit2019}) is asymptotically resolvable but not quasi-resolute (e.g., consider a three-candidate election in which $Margin_\mathbf{P}(a,b)=2$, $Margin_\mathbf{P}(b,c)=4$, and $Margin_\mathbf{P}(c,a)=6$, in which case Borda picks $b$ and $c$).

In Section \ref{NewCriteria}, we discussed the tradeoff between a voting method being quasi-resolute and satisfying stability for winners. The next result illustrates this tradeoff in the case of resolvability. We impose an assumption that is satisfied by all voting methods based on majority margins that we know of---not only Condorcet methods but also, e.g., Borda (see \citealt[p.~28]{Zwicker2016} for a formulation of Borda as a margin-based method). Say that a voting method $F$ satisfies the \textit{triangle property} if for any uniquely-weighted linear profile $\mathbf{P}$ with a majority cycle, if $x$ has the largest margin of victory and smallest margin of loss, then $x\in F(\mathbf{P})$ (a property also used in Example \ref{TwoCondorcetian}). 
  
The proof of Theorem \ref{ResolveStability} makes essential use of a theorem of Harrison-Trainor \citeyearpar{HT2020} that answers one of our conjectures. 

\begin{theorem}\label{ResolveStability} Suppose $F$ is a voting method on the domain of linear profiles satisfying stability for winners and the triangle property. Then $F$ does not satisfy single-voter resolvability with respect to its domain, and $F$ does not satisfy asymptotic resolvability for any $k>3$.\end{theorem}

\begin{proof} In this proof, all profiles are assumed to be linear. 

We will use the fact that stability for winners implies the following: for any profile $\mathbf{P}$, defining $\mathbf{P}^G=\mathbf{P}\vert _{ GETCHA(\mathbf{P})}$ (recall Section \ref{SmithSchwartz}), we have $F(\mathbf{P}^G)\subseteq F(\mathbf{P})$. To see this, let $X(\mathbf{P})\setminus GETCHA(\mathbf{P})=\{b_1,\dots,b_n\}$. Suppose $a\in F(\mathbf{P}^G)$, so $a\in GETCHA(\mathbf{P})$. Then $a\to b_1$, so stability for winners implies $a\in F(\mathbf{P}\vert _{ GETCHA(\mathbf{P})\cup\{b_1\}})$. Then since $a\to b_2$, stability for winners implies $a\in F(\mathbf{P}\vert _{ GETCHA(\mathbf{P})\cup\{b_1,b_2\}})$, and so on, until we obtain $a\in F(\mathbf{P})$.

We will also use the notion of a \textit{qualitative margin graph}, which is a pair $\mathbb{M}=(M,\prec)$ where $M$ is an asymmetric directed graph and $\prec$ is a strict weak order on the set of edges of $M$. We say that $\mathbb{M}$ is \textit{uniquely weighted} if $\prec$ is a strict linear order. Given a profile $\mathbf{P}$, let the \textit{qualitative margin graph} $\mathbb{M}(\mathbf{P})$ \textit{of} $\mathbf{P}$ be the pair $(M(\mathbf{P}),\prec_\mathbf{P})$ where $M(\mathbf{P})$ is the majority graph of $\mathbf{P}$, and $\prec_\mathbf{P}$ is the relation on the set of edges of $M(\mathbf{P})$ defined by $(a,b)\prec_\mathbf{P}(c,d)$ if $Margin_\mathbf{P}(a,b)<Margin_\mathbf{P}(c,d)$. It follows from Debord's Theorem that every qualitative margin graph is realized by some profile. Harrison-Trainor \citeyearpar{HT2020} proves that for any $k\geq 1$ and uniquely-weighted qualitative margin graph $\mathbb{M}$ with $k$ candidates, the proportion of profiles with $k$ candidates and $n$ voters realizing $\mathbb{M}$ does not go to 0 as $n$ goes to infinity. Thus, asymptotic resolvability for $k$ candidates implies the following condition $(\star)$: there is no uniquely-weighted qualitative margin graph $\mathbb{M}$ with $k$ candidates such that for every profile $\mathbf{P}$ realizing $\mathbb{M}$,  $\vert F(\mathbf{P})\vert >1$. This also follows from single-voter resolvability: for if there exists a uniquely-weighted $\mathbb{M}$ such that every $\mathbf{P}$ realizing $\mathbb{M}$ has $\vert F(\mathbf{P})\vert >1$, then we can pick a profile $\mathbf{P}$ realizing $\mathbb{M}$ with sufficiently many voters (note that if $\mathbf{P}$ realizes $\mathbb{M}$, so does $\mathbf{P}+\mathbf{P}^*$ where $\mathbf{P}^*$ is a copy of $\mathbf{P}$ with a disjoint set of voters) such that for any single-voter profile $\mathbf{P}'$, $\mathbf{P}+\mathbf{P}'$ still realizes $\mathbb{M}$ (since the differences between distinct margins are too large in $\mathbf{P}$ for one voter to change the qualitative margin graph), so that $\vert F(\mathbf{P}+\mathbf{P}')\vert >1$, in violation of single-voter resolvability.

Now consider any profile $\mathbf{P}$ with $X(\mathbf{P})>3$ realizing a qualitative margin graph $\mathbb{M}$ that \textit{when restricted to $GETCHA(\mathbf{P})$} has the following form, where $\alpha \prec  \gamma \prec \beta$ and $\gamma \prec \varphi\prec \psi$:
\begin{center}
\begin{minipage}{2in}\begin{tikzpicture}

\node[circle,draw, minimum width=0.25in] at (0,0) (b) {$x_1$}; 
\node[circle,draw,minimum width=0.25in] at (3,0) (a) {$x_3$}; 
\node[circle,draw,minimum width=0.25in] at (1.5,1.5) (c) {$x_2$}; 
\node[circle,draw,minimum width=0.25in] at (1.5,-1.5) (d) {$x_4$}; 

\path[->,draw,thick] (b) to (a);
\path[->,draw,thick] (d) to (c);
\path[->,draw,thick] (c) to node[fill=white] {$\gamma$} (b);
\path[->,draw,thick] (a) to node[fill=white] {$\varphi$} (c);
\path[->,draw,thick] (d) to node[fill=white] {$\chi$} (a);
\path[->,draw,thick] (b) to node[fill=white] {$\alpha$} (d);

\node[fill=white] at (1.5,.5)  {$\beta$}; 
\node[fill=white] at (2,0)  {$\psi$}; 

  \end{tikzpicture}
\end{minipage}
\end{center}
Since $\alpha \prec  \gamma \prec \beta$, by the triangle property we have $x_4\in F((\mathbf{P}^G)_{-x_3})$. Then given $x_4\to x_3$, from stability for winners we have $x_4\in F(\mathbf{P}^G)$ and hence $x_4\in F(\mathbf{P})$ by the first paragraph of the proof. Since  $\gamma \prec \varphi\prec \psi$, by the triangle property we have $x_1\in F((\mathbf{P}^G)_{-x_4})$. Then given $x_1\to x_4$, from stability for winners we have $x_1\in F(\mathbf{P}^G)$ and hence $x_1\in F(\mathbf{P})$ by the first paragraph of the proof. Thus, $\vert F(\mathbf{P})\vert >1$. Since this holds for every $\mathbf{P}$ realizing $\mathbb{M}$, condition $(\star)$ above does not hold, so neither version of resolvability holds either.\end{proof}

It is easy to see that Split Cycle satisfies asymptotic resolvability for $k=2$ and $k=3$ (for $k=3$, this follows from Proposition \ref{Exclude2}.\ref{Exclude2a}). For $k>3$, since Split Cycle satisfies the triangle property and stability for winners, Theorem \ref{ResolveStability} yields the following.

\begin{corollary}\label{NotAsymp} For $k>3$, Split Cycle does not satisfy asymptotic resolvability for $k$ candidates.
\end{corollary} 

Table \ref{NarrowingTable} shows estimates for the average sizes of winning sets in the limit as the number of voters goes to infinity for several voting methods that are not asymptotically resolvable. Estimates were obtained using the Monte Carlo simulation technique described in \citealt[\S~12]{HT2020}, sampling 10,000,000 profiles for each number of candidates.

\begin{table}[h]
\begin{center}
\begin{tabular}{l|rrrrrrrrrr}

  &    3 &    4 &    5 &    6 &    7 &    8 &    9 &   10 &    20 &    30 \\
\hline
 Split Cycle     & 1    & 1.01 & 1.03 & 1.06 & 1.08 & 1.11 & 1.14 & 1.16 &  1.42 &  1.62 \\
 Copeland        & 1.18 & 1.26 & 1.29 & 1.3  & 1.31 & 1.31 & 1.31 & 1.31 &  1.28 &  1.25 \\
 Uncovered Set   & 1.18 & 1.35 & 1.53 & 1.71 & 1.90  & 2.08 & 2.27 & 2.47 &  4.55 &  6.84 \\
 GETCHA          & 1.18 & 1.44 & 1.79  & 2.21 & 2.72 & 3.30 & 3.96 & 4.68 & 13.52 & 22.9 \\
\end{tabular}
\end{center}
\caption{Estimated average sizes of winning sets for profiles with a given number of candidates (top row) in the limit as the number of voters goes to infinity.}\label{NarrowingTable}
\end{table}

While it is certainly of theoretical interest to know whether the proportion of profiles with multiple winners goes to 0 as the number of voters goes to infinity, for real world applications, what matters is the proportion of profiles with multiple winners for realistic numbers of voters. In  Appendix \ref{Quant}, we provide a quantitative  analysis.   For instance, our results show that when there are 7 candidates and up to a few thousand voters,  Split Cycle produces multiple winners on only about 1\% more of such profiles than Beat Path, which satisfies resolvability in both forms above.  Our results also show that this difference in frequency of multiple winners decreases as the number of candidates decreases.  In addition, our analysis shows that Split Cycle is substantially more resolute than GETCHA. 

\setcounter{secnumdepth}{3}

\subsection{Monotonicity Criteria}\label{MonSection}

\subsubsection{Non-Negative Reponsiveness}\label{MonotonicitySection}

The term `monotonicity' has many meanings in voting theory. One of the standard meanings is given by the criterion of \textit{non-negative responsiveness} (\citealt{Tideman1987}): lifting the position of a winner $x$ on voters' ballots cannot result in $x$ becoming a loser.

\begin{definition}\label{SimpleLift} For any profiles $\mathbf{P}$ and $\mathbf{P}'$ with $V(\mathbf{P})=V(\mathbf{P}')$ and $x\in X(\mathbf{P})=X(\mathbf{P}')$, we say that \textit{$\mathbf{P}'$ is obtained from $\mathbf{P}$ by a simple lift of $x$} if the following conditions hold:
\begin{enumerate}
\item\label{SimpleLift1} for all $a,b\in X(\mathbf{P})\setminus\{x\}$ and $i\in V(\mathbf{P})$, $a\mathbf{P}_i b$ if and only if $a\mathbf{P}'_i b$;
\item\label{SimpleLift2} for all $a\in X(\mathbf{P})$ and $i\in V(\mathbf{P})$, if $x\mathbf{P}_ia$ then $x\mathbf{P}_i' a$;
\item\label{SimpleLift2} for all $a\in X(\mathbf{P})$ and $i\in V(\mathbf{P})$, if $a\mathbf{P}_i'x$ then $a\mathbf{P}_i x$.
\end{enumerate}
\end{definition}

\begin{definition} A voting method $F$ satisfies \textit{non-negative responsiveness} if for every $\mathbf{P}\in\mathrm{dom}(F)$ and $x\in X(\mathbf{P})$, if $x\in F(\mathbf{P})$ and $\mathbf{P}'\in\mathrm{dom}(F)$ is obtained from $\mathbf{P}$ by a simple lift of $x$, then $x\in F(\mathbf{P}')$.
\end{definition}

\begin{proposition} Split Cycle satisfies non-negative responsiveness.
\end{proposition}

\begin{proof} Suppose $x\in SC(\mathbf{P})$ and $\mathbf{P}'$ is obtained from $\mathbf{P}$ by a simple lift of $x$. Since $x\in SC(\mathbf{P})$, for all $y\in X(\mathbf{P})$, $y$ does not defeat $x$ in $\mathbf{P}$, so $Margin_\mathbf{P}(y,x)\leq Cycle\#_\mathbf{P}(y,x)$. We claim that $y$ does not defeat $x$ in $\mathbf{P}'$ either. Since $\mathbf{P}'$ is obtained from $\mathbf{P}$ by a simple lift of $x$, we have $Margin_{\mathbf{P}'}(y,x)\leq Margin_\mathbf{P}(y,x)$. If $Margin_{\mathbf{P}'}(y,x)\leq 0$, then $y$ does not defeat $x$ in $\mathbf{P}'$, so suppose  $Margin_{\mathbf{P}'}(y,x)>0$. We claim that 
\begin{equation}Cycle\#_{\mathbf{P}'}(y,x)\geq Cycle\#_\mathbf{P}(y,x) - (Margin_\mathbf{P}(y,x) - Margin_{\mathbf{P}'}(y,x)).\label{LiftCycleNumber}\end{equation}
For given any simple cycle $\rho = y \overset{\alpha}{\rightarrow} x\overset{\beta}{\rightarrow} z_1\overset{\gamma_1}{\rightarrow} \dots\overset{\gamma_{n-1}}{\rightarrow} z_n\overset{\gamma_n}{\rightarrow} y$ in $\mathcal{M}(\mathbf{P})$, by Definition \ref{SimpleLift} we have that $\rho' = y \overset{\alpha'}{\rightarrow} x\overset{\beta'}{\rightarrow} z_1\overset{\gamma_1}{\rightarrow} \dots\overset{\gamma_{n-1}}{\rightarrow} z_n\overset{\gamma_n}{\rightarrow} y$ is a simple cycle in $\mathcal{M}(\mathbf{P}')$ where $\alpha'\leq\alpha$ and  $\beta\leq\beta'$. Hence $Split\#(\rho')\geq Split\#(\rho) - (Margin_\mathbf{P}(y,x) - Margin_{\mathbf{P}'}(y,x))$. This proves (\ref{LiftCycleNumber}), which with $Margin_\mathbf{P}(y,x)\leq Cycle\#_\mathbf{P}(y,x)$ implies  $Margin_{\mathbf{P}'}(y,x)\leq Cycle\#_{\mathbf{P}'}(y,x)$. Hence $y$ does not defeat $x$ in $\mathbf{P}'$. Since $y$ was arbitrary, we conclude that $x\in SC(\mathbf{P}')$.\end{proof}

\subsubsection{Positive and Negative Involvement}\label{InvolvementSection}

Like rejectability and resolvability, the next two criteria we consider---positive and negative involvement---also concern adding voters to an election. In this case, the concern is about perverse changes to the set of winners in light of who the new voters rank as their favorite (resp.~least favorite) candidate. Recall our discussion in Section \ref{NoShowSection} of violations of positive or negative involvement as ``strong no show paradoxes.'' The criterion of positive (resp.~negative) involvement ensures that if $x$ is among the winners (resp.~losers) and we add a voter who ranks $x$ as their favorite (resp.~least favorite), then $x$ will still be a winner (resp.~loser).

\begin{definition} $F$ satisfies \textit{positive involvement} if for any profiles $\mathbf{P}\in\mathrm{dom}(F)$ and $\mathbf{P}'$ with $X(\mathbf{P})=X(\mathbf{P}')$, ${V(\mathbf{P})\cap V(\mathbf{P}')=\varnothing}$, and $\vert V(\mathbf{P}')\vert =1$, if $x\in F(\mathbf{P})$, $\mathbf{P}+\mathbf{P}'\in\mathrm{dom}(F)$, and for $i\in V(\mathbf{P}')$, $x \mathbf{P}_i' y$ for all $y\in X(\mathbf{P}')\setminus \{x\}$, then $x\in F(\mathbf{P}+\mathbf{P}')$.

$F$ satisfies \textit{negative involvement} if for any profiles $\mathbf{P}\in\mathrm{dom}(F)$ and $\mathbf{P}'$ with $X(\mathbf{P})=X(\mathbf{P}')$, $V(\mathbf{P})\cap V(\mathbf{P}')=\varnothing$, and $\vert V(\mathbf{P}')\vert =1$, if $x\not\in F(\mathbf{P})$, $\mathbf{P}+\mathbf{P}'\in\mathrm{dom}(F)$, and for $i\in V(\mathbf{P}')$, $y \mathbf{P}_i' x$ for all $y\in X(\mathbf{P}')\setminus \{x\}$, then $x\not\in F(\mathbf{P}+\mathbf{P}')$.
\end{definition}

\begin{lemma}\label{MoreThanOne} If $F$ satisfies positive involvement (resp.~negative involvement), then it satisfies the analogous coalitional properties that drop the restriction that $\vert V(\mathbf{P}')\vert =1$.
\end{lemma}
\begin{proof} To prove the properties for a coalition of more than one voter, add each voter in the coalition one at a time, applying positive (resp.~negative involvement) at each step. This can be iterated because the property of $x$ belonging to (resp.~not belonging to) the winning set is preserved at each step.\end{proof}

\begin{remark} It is important to distinguish positive and negative involvement from the \textit{participation} criterion (recall Section \ref{NoShowSection}), which we discuss further in Appendix \ref{ParticipationAppendix}. It is also important that positive (resp.~negative) involvement applies only when adding a voter for whom $x$ is their \textit{unique} favorite (resp.~least favorite) candidate. One may consider a related criterion concerning voters for whom $x$ is \textit{among} their favorite (resp.~least favorite) candidates (see \citealt{Duddy2014}). But we see no problem with the addition of voters who rank $x$ and $y$ as tied changing the winner of an election with majority cycles from $x$ to $y$, given how the new voters change $x$'s and $y$'s pairwise performance against other candidates.\end{remark}

None of Beat Path, Ranked Pairs, Copeland, GETCHA/GOCHA, or Uncovered Set satisfies positive or negative involvement. The failure of positive and negative involvement has been called ``a common flaw in Condorcet voting correspondences'' (\citealt{Perez2001}). However, Split Cycle does not have this flaw.

\begin{proposition} Split Cycle satisfies positive and negative involvement.
\end{proposition}
\begin{proof}
First, consider positive involvement. We prove the contrapositive. Suppose $x\not\in SC(\mathbf{P}+\mathbf{P}')$. Hence there is a $z\in X(\mathbf{P})$ that defeats $x$ in $\mathbf{P}+\mathbf{P}'$, i.e., such that 
\begin{equation}Margin_{\mathbf{P}+\mathbf{P}'}(z,x)>Cycle\#_{\mathbf{P}+\mathbf{P}'}(z,x).\label{ParticipationEq1}\end{equation}
Since $\vert V(\mathbf{P}')\vert =1$, we have 
\begin{equation}Cycle\#_{\mathbf{P}+\mathbf{P}'}(z,x)\geq Cycle\#_{\mathbf{P}}(z,x)-1,\end{equation} and since
$x\mathbf{P}_i'z$, we have 
\begin{equation}Margin_{\mathbf{P}+\mathbf{P}'}(z,x)=Margin_{\mathbf{P}}(z,x)-1.\label{ParticipationEq3}\end{equation}
It follows from (\ref{ParticipationEq1})--(\ref{ParticipationEq3}) that
\[Margin_{\mathbf{P}}(z,x)>Cycle\#_{\mathbf{P}}(z,x),\]
so $x\not\in SC(\mathbf{P})$.

Next, consider negative involvement. Suppose $x\not\in F(\mathbf{P})$. Hence there is a $z\in X(\mathbf{P})$ that defeats $x$ in $\mathbf{P}$, i.e., such that 
\begin{equation}Margin_{\mathbf{P}}(z,x)>Cycle\#_{\mathbf{P}}(z,x).\label{ParticipationEq1'}\end{equation}
Since $\vert V(\mathbf{P}')\vert =1$, we have 
\begin{equation} Cycle\#_{\mathbf{P}+\mathbf{P}'}(z,x)\leq Cycle\#_{\mathbf{P}}(z,x)+1,\end{equation} and since
$z\mathbf{P}_i'x$, we have 
\begin{equation}Margin_{\mathbf{P}+\mathbf{P}'}(z,x)=Margin_{\mathbf{P}}(z,x)+1.\label{ParticipationEq3'}\end{equation}
It follows from (\ref{ParticipationEq1'})--(\ref{ParticipationEq3'}) that
\[Margin_{\mathbf{P}+\mathbf{P}'}(z,x)>Cycle\#_{\mathbf{P}+\mathbf{P}'}(z,x),\]
so $x\not\in SC(\mathbf{P}+\mathbf{P}')$.\end{proof}

Thus, with Split Cycle the strong no show paradox discussed in Section \ref{NoShowSection} is impossible.

\section{Conclusion}\label{Conclusion}

In this paper, we have proposed the Split Cycle voting method, which can be distinguished from all methods we know of in any of the following three ways:
\begin{itemize}
\item Only Split Cycle satisfies independence of clones, positive involvement, and at least one of  Condorcet consistency, monotonicity, and immunity to spoilers.
\item Only Split Cycle satisfies independence of clones and negative involvement.
\item Only Split Cycle satisfies independence of clones, immunity to spoilers, and rejectability.
\end{itemize}
Moreover, Split Cycle can be motivated by the three key ideas of Section \ref{SplitCycleSection}:
\begin{enumerate}
\item Group incoherence raises the threshold for defeat, but not infinitely.
\item Incoherence can be localized.
\item Defeat is direct.
\end{enumerate}
We think the third idea is especially important for justifying election outcomes to supporters of a candidate who was not among the winners of the election. To try to explain to supporters of a candidate $x$ that the reason $x$ is not among the winners is that another candidate $y$ ``defeated'' $x$ \textit{even though a majority of voters prefer $x$ to $y$} (as is possible with the Beat Path voting method, for example) seems a recipe for complaints of illegitimacy and resulting social instability.

There are several natural next steps for future research. For theoretical purposes, it would be desirable to have a set of axioms that single out Split Cycle not only from known voting methods but from all possible voting methods, providing a complete axiomatic characterization of Split Cycle (see \citealt{HP2021} and \citealt{Ding2022}). For both theoretical and applied purposes, it would be desirable to have a more detailed quantitative analysis of how Split Cycle performs on profiles with various numbers of candidates and voters. The code we are making available  (recall Remark \ref{CodeRemark}) allows any researcher to perform such analyses. Ultimately, of course, the best test of Split Cycle will come from its continued use in~practice.

\subsection*{Acknowledgements}

For helpful comments, we wish to thank Felix Brandt, Yifeng Ding, Mikayla Kelley, John Moser, Chase Norman, Dominik Peters, Markus Schulze, Warren D. Smith, Nicolaus Tideman, John~Weymark, students in the Fall 2019 seminar on Voting and Democracy at UC Berkeley,  students in the Spring 2020 seminar on Preference and Judgment Aggregation at the University of Maryland, the referees for \textit{Public Choice}, and the handling editor, Marek Kaminski. We are also grateful for feedback received from the audiences at presentations of this work in the Logic Seminar at Stanford University on November 20,~2019, the 2nd Games, Agents, and Incentives Workshop (GAIW@AAMAS 2020) on May 10, 2020, and the Online Social Choice and Welfare Seminar Series on February 2, 2021.

\appendix

\section{Independence of Clones}\label{ClonesAppendix}

In this appendix, we prove that Split Cycle satisfies independence of clones. In the following, fix a profile $\mathbf{P}$ with a set $C$ of clones and $c\in C$.  Then obviously we have the following.

\begin{lemma}\label{SameMargin} 
\begin{enumerate}
\item\label{SameMargin1} For any $a,b\in X(\mathbf{P})\setminus \{c\}$, $Margin_{\mathbf{P}}(a,b)=Margin_{\mathbf{P}_{-c}}(a,b)$.
\item\label{SameMargin2} For any $b\in X(\mathbf{P})\setminus C$ and $e\in C\setminus\{c\}$, $Margin_{\mathbf{P}}(c,b)=Margin_{\mathbf{P}_{-c}}(e,b)$.
\end{enumerate}
\end{lemma}

Next we show that certain cycle numbers do not change from $\mathbf{P}$ to $\mathbf{P}_{-c}$.  For this we use the following key lemma.

\begin{lemma}\label{KeyLem} For any $c_1,c_2\in C$ with $c_1\neq c_2$ and simple cycle $\rho$ in $\mathcal{M}(\mathbf{P})$ that contains $c_1$ and some non-clone, the sequence $\rho'$ obtained from $\rho$ by replacing all clones in $\rho$ by $c_2$ and then replacing any subsequence $c_2,\dots,c_2$ by $c_2$ is a simple cycle in $\mathcal{M}(\mathbf{P})$ such that $Split\#(\rho')\geq Split\#(\rho)$.\end{lemma}

\begin{proof} For any $a\in X(\mathbf{P})\setminus C$ and $d\in C$, if $a\to d$ (resp.~$d\to a$) occurs in $\rho$ with margin $\alpha$ in $\mathcal{M}(\mathbf{P})$, then by the definition of a set of clones (Definition \ref{CloneDef}), we have $a\to c_2$ (resp.~$c_2\to a$) with margin $\alpha$ in $\mathcal{M}(\mathbf{P})$. It follows that $\rho'$ is a simple cycle in $\mathcal{M}(\mathbf{P})$ and that the margins between successive candidates in $\rho'$ already occurred as margins between successive candidates in $\rho$, which implies $Split\#(\rho')\geq Split\#(\rho)$.\end{proof}

\begin{lemma}\label{SameCycleNumber1} For any $a\in X(\mathbf{P})\setminus\{c\}$ and $b\in X(\mathbf{P})\setminus C $, we have $Cycle\#_\mathbf{P}(a,b)=Cycle\#_{\mathbf{P}_{-c}}(a,b)$.
\end{lemma} 

\begin{proof} First, observe that any simple cycle in $\mathcal{M}(\mathbf{P}_{-c})$ is also a simple cycle with the same margins in $\mathcal{M}(\mathbf{P})$. Hence $Cycle\#_\mathbf{P}(a,b)\geq Cycle\#_{\mathbf{P}_{-c}}(a,b)$.

Second, to show $Cycle\#_\mathbf{P}(a,b)\leq Cycle\#_{\mathbf{P}_{-c}}(a,b)$, it suffices to show that for every simple cycle $\rho$ in $\mathcal{M}(\mathbf{P})$ extending $a\to b$, there is a simple cycle $\rho'$ in $\mathcal{M}(\mathbf{P}_{-c})$ extending $a\to b$ such that $Split\#(\rho ')\geq Split\#(\rho)$. If $\rho$ does not contain $c$, then take $\rho'=\rho$. Suppose $\rho$ does contain $c$. Case 1: $a\in C$. Then apply Lemma \ref{KeyLem} with $c_1:=c$ and $c_2:=a$ to obtain a simple cycle $\rho'$ in $\mathcal{M}(\mathbf{P})$ extending $a\to b$, but not containing $c$, such that $Split\#(\rho ')\geq Split\#(\rho)$; since $\rho'$ does not contain $c$, it is also a simple cycle in $\mathcal{M}(\mathbf{P}_{-c})$ with the desired properties. Case 2: $a\not\in C$. Then apply Lemma \ref{KeyLem} with $c_1:=c$, $c_2\in C\setminus\{c\}$ and reason as in~Case~1.\end{proof}

\begin{lemma}\label{SameCycleNumber2} Let $d\in C$ and $e\in C\setminus\{c\}$.
\begin{enumerate}
\item\label{SameCycleNumber21} For any $b\in X(\mathbf{P})\setminus C$, $Cycle\#_\mathbf{P}(d,b)=Cycle\#_{\mathbf{P}_{-c}}(e,b)$;
\item\label{SameCycleNumber22} For any $a\in X(\mathbf{P})\setminus C$, $Cycle\#_\mathbf{P}(a,d)=Cycle\#_{\mathbf{P}_{-c}}(a,e)$.
\end{enumerate}
\end{lemma}

\begin{proof}For part \ref{SameCycleNumber21}, for any simple cycle $\rho$ in $\mathcal{M}(\mathbf{P})$ extending $d\to b$, by Lemma \ref{KeyLem} with $c_1:=d$ and $c_2:=e$, there is a simple cycle $\rho'$ in $\mathcal{M}(\mathbf{P})$ extending $e\to b$, but not containing $c$ (since $e\in C\setminus\{c\}$), such that $Split\#(\rho')\geq Split\#(\rho)$. Since $\rho'$ does not contain $c$, it is also a simple cycle in $\mathcal{M}(\mathbf{P}_{-c})$ extending $e\to b$ with the same margins. Hence $Cycle\#_\mathbf{P}(d,b)\leq Cycle\#_{\mathbf{P}_{-c}}(e,b)$. Next, suppose $\rho$ is a simple cycle in $\mathcal{M}(\mathbf{P}_{-c})$ extending $e \to b$. Then $\rho$ is also a simple cycle in $\mathcal{M}(\mathbf{P})$ extending $e\to b$ with the same margins. Thus, by Lemma \ref{KeyLem} with $c_1:=e$ and $c_2:=d$, there is a simple cycle $\rho'$ in $\mathcal{M}(\mathbf{P})$ extending $d\to b$  such that $Split\#(\rho')\geq Split\#(\rho)$. Hence $Cycle\#_\mathbf{P}(d,b)\geq Cycle\#_{\mathbf{P}_{-c}}(e,b)$.

The proof of part \ref{SameCycleNumber22} is analogous.\end{proof}

\begin{proposition}\label{NonCloneIndependence} For any $b\in X(\mathbf{P})\setminus C$, we have $b\in SC(\mathbf{P})$ if and only if $b\in SC(\mathbf{P}_{-c})$. Hence Split Cycle is such that non-clone choice is independent of clones. 
\end{proposition}

\begin{proof} Suppose $b\not\in SC(\mathbf{P}_{-c})$, so there is an $a\in X(\mathbf{P})\setminus \{c\}$ such that $a$ defeats $b$ in $\mathbf{P}_{-c}$. Then by Lemmas \ref{SameMargin}.\ref{SameMargin1} and \ref{SameCycleNumber1}, $a$ defeats $b$ in $\mathbf{P}$. Now suppose $b\not\in SC(\mathbf{P})$, so there is an $a\in X(\mathbf{P})$ such that $a$ defeats $b$ in $\mathbf{P}$. Case 1: $a\neq c$. Then by Lemmas \ref{SameMargin}.\ref{SameMargin1} and \ref{SameCycleNumber1} again, $a$ defeats $b$ in $\mathbf{P}_{-c}$. Case 2: $a=c$. Then by Lemmas \ref{SameMargin}.\ref{SameMargin2} and \ref{SameCycleNumber2}.\ref{SameCycleNumber21} with $d:=c$, each $e\in C\setminus\{c\}$ defeats $b$ in~$\mathbf{P}_{-c}$.
\end{proof}

\begin{proposition}\label{CloneIndependence} $C\cap SC(\mathbf{P})\neq\varnothing\mbox{ if and only if }C\setminus\{c\}\cap SC(\mathbf{P}_{-c})\neq\varnothing$. Hence Split Cycle is such that clone choice is independent of clones.
\end{proposition}

\begin{proof} Suppose $C\cap SC(\mathbf{P})=\varnothing$. Hence every clone in $C$ is defeated in $\mathbf{P}$. Since the defeat graph for $\mathbf{P}$ contains no cycles (Lemma \ref{NoCycles}), it follows that there is some $a\in X(\mathbf{P})\setminus C$ that defeats some $d\in C$ in $\mathbf{P}$. It then follows by the definition of a set of clones (Definition \ref{CloneDef}), Lemma \ref{SameMargin}.\ref{SameMargin1}, and Lemma \ref{SameCycleNumber2}.\ref{SameCycleNumber22} that $a$ defeats every $e\in C\setminus\{c\}$ in $\mathbf{P}_{-c}$. Hence $C\setminus\{c\}\cap SC(\mathbf{P}_{-c})=\varnothing$. Similarly, if $C\setminus\{c\}\cap SC(\mathbf{P}_{-c})=\varnothing$, then there is some $a\in X(\mathbf{P}_{-c})\setminus C$ that defeats some $e\in C\setminus\{c\}$ in $\mathbf{P}_{-c}$. It then follows by Definition \ref{CloneDef}, Lemma \ref{SameMargin}.\ref{SameMargin1}, and Lemma \ref{SameCycleNumber2}.\ref{SameCycleNumber22} that $a$ defeats every $d\in C$ in $\mathbf{P}$. Hence $C\cap SC(\mathbf{P})=\varnothing$.\end{proof}

\begin{theorem} Split Cycle satisfies independence of clones.
\end{theorem}
\begin{proof} By Propositions \ref{NonCloneIndependence} and \ref{CloneIndependence}.
\end{proof}

\section{Participation}\label{ParticipationAppendix}

In Sections \ref{NoShowSection} and \ref{InvolvementSection} on positive and negative involvement, we mentioned the related \textit{participation} criterion. Participation is usually stated for resolute voting methods: if $x$ is the winner in a profile, and we add to the profile a new voter who strictly prefers $x$ to $y$, then $y$ is not the winner in the resulting profile. (Note that there is no requirement that $x$ be at the top of the new voter's ballot or that $y$ be at the bottom, a point to which we return below.) When applied to irresolute voting methods, we call this ``resolute'' participation.\footnote{Several authors have investigated what could be called ``irresolute'' participation-like criteria (recall Footnote \ref{IrresNote}), where one changes the initial assumption from $F(\mathbf{P})=\{x\}$ to $x\in F(\mathbf{P})$. For example, Perez \citeyearpar{Perez2001} considers the following axiom, called \textit{VC-participation}: if $x\in F(\mathbf{P})$ and $\mathbf{P}'$ is a one-voter profile with a new voter $i$ having $x \mathbf{P}_i' y$, then $y\in F(\mathbf{P}+\mathbf{P}')$ implies $x\in F(\mathbf{P}+\mathbf{P}')$. He then observes that no Condorcet consistent voting method satisfies VC-participation. However, it is not clear that this criterion is a plausible normative requirement on a voting method. Suppose, for example, that $i$'s ranking is $zxyw$, and $i$'s joining the election results in a change from $F(\mathbf{P})=\{x,w\}$ to $F(\mathbf{P}+\mathbf{P}')=\{z,y\}$. It is not clear that we should impose a criterion that prohibits such a change, which seems to be a strict improvement from $i$'s point of view.}

\begin{definition} A voting method $F$ satisfies \textit{resolute participation} if for any  $\mathbf{P}\in\mathrm{dom}(F)$ and $\mathbf{P}'$ with $X(\mathbf{P})=X(\mathbf{P}')$, $V(\mathbf{P})\cap V(\mathbf{P}')=\varnothing$,  $\vert V(\mathbf{P}')\vert =1$, and $\mathbf{P}+\mathbf{P}'\in\mathrm{dom}(F)$, and any $x,y\in X(\mathbf{P})$, if $F(\mathbf{P})=\{x\}$ and $x\mathbf{P}_i'y$ for $i\in V(\mathbf{P}')$, then $F(\mathbf{P}+\mathbf{P}')\neq\{y\}$.
\end{definition}

It turns out that for linear profiles, Split Cycle satisfies resolute participation---but for a reason unrelated to the main idea of participation, namely that Split Cycle satisfies the following stronger property.

\begin{definition} A voting method $F$ satisfies \textit{winner continuity} if for any $\mathbf{P}\in\mathrm{dom}(F)$ and $\mathbf{P}'$ with $X(\mathbf{P})=X(\mathbf{P}')$, $V(\mathbf{P})\cap V(\mathbf{P}')=\varnothing$, and $\vert V(\mathbf{P}')\vert =1$, if $F(\mathbf{P})=\{x\}$ and $\mathbf{P}+\mathbf{P}'\in\mathrm{dom}(F)$, then $x\in F(\mathbf{P}+\mathbf{P}')$.
\end{definition}

Note, for example, that Plurality satisfies winner continuity, while the Borda voting method does not.

\begin{proposition}\label{WinnerCon} Restricted to linear profiles, Split Cycle satisfies winner continuity.
\end{proposition}

\begin{proof} Suppose $SC(\mathbf{P})=\{x\}$. Further suppose $x\not\in SC(\mathbf{P}+\mathbf{P}')$, so there is some $z\in X(\mathbf{P})$ such that 
\begin{equation}Margin_{\mathbf{P}+\mathbf{P}'}(z,x)>Cycle\#_{\mathbf{P}+\mathbf{P}'}(z,x).\label{PickUpz}\end{equation}
Since $z\not\in SC(\mathbf{P})$,  by Lemma \ref{BeatPathFromWinner} there are distinct $y_1,\dots,y_n$ with $y_1=x$ and $y_n=z$ such that $y_1Dy_{2}D\dots Dy_{n-1} D y_n$ in the defeat graph of $\mathbf{P}$.

Since $\vert V(\mathbf{P}')\vert =1$, it follows from (\ref{PickUpz}) that $Margin_{\mathbf{P}}(z,x)=0\mbox{ or }Margin_{\mathbf{P}}(z,x)>0$.

Case 1: $Margin_{\mathbf{P}}(z,x)=0$. Then since $\mathbf{P}$ is a linear profile, for each $i\in \{1,\dots,n-1\}$, $Margin_{\mathbf{P}}(y_i,y_{i+1})$ is even, and since $y_iDy_{i+1}$, it is greater than 0, so $Margin_{\mathbf{P}}(y_i,y_{i+1})\geq  2$. Since $\vert V(\mathbf{P}')\vert =1$, together $Margin_{\mathbf{P}}(z,x)=0$ and (\ref{PickUpz}) imply $Margin_{\mathbf{P}+\mathbf{P}'}(z,x)=1$. In addition, since $\vert V(\mathbf{P}')\vert =1$, from $Margin_{\mathbf{P}}(y_i,y_{i+1})\geq  2$ we have $Margin_{\mathbf{P}+\mathbf{P}'}(y_i,y_{i+1})\geq  1$. Thus, we have a simple cycle \[\rho=y_1 \overset{\gamma_1}{\longrightarrow} y_{2} \overset{\gamma_{2}}{\longrightarrow} \dots \overset{\gamma_{n-1}}{\longrightarrow} y_n \overset{\delta}{\longrightarrow} y_1\]
in the margin graph of $\mathbf{P}+\mathbf{P}'$ in which $\delta$, i.e., $Margin_{\mathbf{P}+\mathbf{P}'}(z,x)$, is not greater than any $\gamma_i$. But this contradicts (\ref{PickUpz}).

Case 2: $Margin_{\mathbf{P}}(z,x)>0$. Together with $y_1Dy_{2}D\dots Dy_{n-1} D y_n$, this means there is a simple cycle  \[\rho=y_1 \overset{\alpha_1}{\longrightarrow} y_{2} \overset{\alpha_{2}}{\longrightarrow} \dots \overset{\alpha_{n-1}}{\longrightarrow} y_n \overset{\beta}{\longrightarrow} y_1\]
in the margin graph of $\mathbf{P}$. Moreover, from $y_1Dy_{2}D\dots Dy_{n-1} D y_n$, it follows that for each $i\in \{1,\dots,n-1\}$, $\alpha_i$ is greater than the splitting number of $\rho$; hence $\beta$, i.e.,  $Margin_{\mathbf{P}}(z,x)$, is the splitting number of $\rho$. Thus, for each $i\in \{1,\dots,n-1\}$, we have $Margin_{\mathbf{P}}(y_i,y_{i+1})\geq  Margin_{\mathbf{P}}(z,x)+2$ since the parity of all margins must be the same, given that $\mathbf{P}$ is a linear profile. Since $\vert V(\mathbf{P}')\vert =1$, it follows that there is a simple cycle
 \[\rho^\star=y_1 \overset{\alpha_1^\star}{\longrightarrow} y_{2} \overset{\alpha_{2}^\star}{\longrightarrow} \dots \overset{\alpha_{n-1}^\star}{\longrightarrow} y_n \overset{\beta^\star}{\longrightarrow} y_1\]
 in the margin graph of $\mathbf{P}+\mathbf{P}'$ in which $\beta^\star$, i.e., $Margin_{\mathbf{P}+\mathbf{P}'}(z,x)$, is not greater than any $\alpha_i^\star$. But this contradicts (\ref{PickUpz}).\end{proof}

\begin{corollary} Restricted to linear profiles, Split Cycle satisfies resolute participation.
\end{corollary}
\begin{proof} Immediate from Proposition \ref{WinnerCon}.
\end{proof}

While positive and negative involvement entail the analogous coalitional properties (recall Lemma \ref{MoreThanOne}), resolute participation does not entail the analogous coalitional property.

\begin{definition} A voting method $F$ satisfies \textit{resolute coalitional participation} if for any  $\mathbf{P}\in\mathrm{dom}(F)$ and $\mathbf{P}'$ with $X(\mathbf{P})=X(\mathbf{P}')$, $V(\mathbf{P})\cap V(\mathbf{P}')=\varnothing$, and $\mathbf{P}+\mathbf{P}'\in\mathrm{dom}(F)$, and any $x,y\in X(\mathbf{P})$, if $F(\mathbf{P})=\{x\}$ and $x\mathbf{P}_i'y$ for all $i\in V(\mathbf{P}')$, then $F(\mathbf{P}+\mathbf{P}')\neq\{y\}$.
\end{definition}

\begin{proposition}\label{ResCoalPart} $\,$
\begin{enumerate}
\item\label{ResCoalPart1} Split Cycle does not satisfy resolute participation on strict weak order profiles.
\item\label{ResCoalPart2} Split Cycle does not satisfy resolute coalitional participation even on linear profiles.
\end{enumerate}
\end{proposition}

\begin{proof} For part \ref{ResCoalPart1}, by Debord's Theorem, there is a strict weak order profile $\mathbf{P}$ whose margin graph is shown on the left below:

\begin{center}
\begin{minipage}{2in}\begin{tikzpicture}

\node[circle,draw, minimum width=0.25in] at (0,0) (c) {$c$}; 
\node[circle,draw,minimum width=0.25in] at (3,0) (a) {$a$}; 
\node[circle,draw,minimum width=0.25in] at (1.5,1.5) (b) {$b$}; 
\node[circle,draw,minimum width=0.25in] at (1.5,-1.5) (d) {$d$}; 

\path[->,draw,thick] (b) to node[fill=white] {$2$} (a);
\path[->,draw,thick] (d) to node[fill=white] {$3$} (c);
\path[->,draw,thick] (c) to node[fill=white] {$1$} (b);

\path[->,draw,thick] (a) to node[fill=white] {$2$} (d);

  \end{tikzpicture}
\end{minipage}\begin{minipage}{2in}\begin{tikzpicture}

\node[circle,draw, minimum width=0.25in] at (0,0) (c) {$c$}; 
\node[circle,draw,minimum width=0.25in] at (3,0) (a) {$a$}; 
\node[circle,draw,minimum width=0.25in] at (1.5,1.5) (b) {$b$}; 
\node[circle,draw,minimum width=0.25in] at (1.5,-1.5) (d) {$d$}; 

\path[->,draw,thick] (b) to node[fill=white] {$3$} (a);
\path[->,draw,thick] (d) to node[fill=white] {$2$} (c);
\path[->,draw,thick] (c) to node[fill=white] {$2$} (b);
\path[->,draw,thick] (c) to (a);
\path[->,draw,thick] (a) to node[fill=white] {$1$} (d);
\path[->,draw,thick] (b) to (d);

\node[fill=white] at (1.5,-.5)  {$1$}; 
\node[fill=white] at (2,0)  {$1$}; 

  \end{tikzpicture}
\end{minipage}
\end{center}
On the right, we show the margin graph of the profile $\mathbf{P}+\mathbf{P}'$ where $\mathbf{P}'$ is a one-voter profile whose voter has $c\mathbf{P}_i'b\mathbf{P}_i'd\mathbf{P}_i'a$. Although $b\mathbf{P}'_id$, we go from $SC(\mathbf{P})=\{b\}$ to $SC(\mathbf{P}+\mathbf{P}')=\{d\}$.

For part \ref{ResCoalPart2}, by Debord's Theorem, there is a linear profile $\mathbf{P}$ whose margin graph is shown on the left below:

\begin{center}
\begin{minipage}{2in}\begin{tikzpicture}

\node[circle,draw, minimum width=0.25in] at (0,0) (c) {$c$}; 
\node[circle,draw,minimum width=0.25in] at (3,0) (a) {$a$}; 
\node[circle,draw,minimum width=0.25in] at (1.5,1.5) (b) {$b$}; 
\node[circle,draw,minimum width=0.25in] at (1.5,-1.5) (d) {$d$}; 

\path[->,draw,thick] (b) to node[fill=white] {$3$} (a);
\path[->,draw,thick] (d) to node[fill=white] {$5$} (c);
\path[->,draw,thick] (c) to node[fill=white] {$1$} (b);
\path[->,draw,thick] (a) to (c);
\path[->,draw,thick] (a) to node[fill=white] {$3$} (d);
\path[->,draw,thick] (d) to (b);

\node[fill=white] at (1.5,.5)  {$1$}; 
\node[fill=white] at (1,0)  {$1$}; 

  \end{tikzpicture}
\end{minipage}\begin{minipage}{2in}\begin{tikzpicture}

\node[circle,draw, minimum width=0.25in] at (0,0) (c) {$c$}; 
\node[circle,draw,minimum width=0.25in] at (3,0) (a) {$a$}; 
\node[circle,draw,minimum width=0.25in] at (1.5,1.5) (b) {$b$}; 
\node[circle,draw,minimum width=0.25in] at (1.5,-1.5) (d) {$d$}; 

\path[->,draw,thick] (b) to node[fill=white] {$5$} (a);
\path[->,draw,thick] (d) to node[fill=white] {$3$} (c);
\path[->,draw,thick] (c) to node[fill=white] {$3$} (b);
\path[->,draw,thick] (c) to (a);
\path[->,draw,thick] (a) to node[fill=white] {$1$} (d);
\path[->,draw,thick] (b) to (d);

\node[fill=white] at (1.5,-.5)  {$1$}; 
\node[fill=white] at (2,0)  {$1$}; 

  \end{tikzpicture}
\end{minipage}
\end{center}
On the right, we show the margin graph of the profile $\mathbf{P}+\mathbf{P}'$ where $\mathbf{P}'$ is a two-voter profile whose two voters both have $c\mathbf{P}_i'b\mathbf{P}_i'd\mathbf{P}_i'a$. Although both voters have $b\mathbf{P}'_id$, we go from $SC(\mathbf{P})=\{b\}$ to $SC(\mathbf{P}+\mathbf{P}')=\{d\}$.\end{proof}

In our view, the examples in the proof of Proposition \ref{ResCoalPart} show that participation is too strong to require. Its violation can be rationalized as follows. In $\mathbf{P}$, $d$ is defeated by $a$; yet with the new voter(s) having $d\mathbf{P}_i'a$, $d$ is no longer defeated by $a$ (or anyone else) in $\mathbf{P}+\mathbf{P}'$. In $\mathbf{P}$, $b$ is not defeated by $c$ (or anyone else); yet with the new voter(s) having $c\mathbf{P}_i'b$, $b$ becomes defeated by $c$ in $\mathbf{P}+\mathbf{P}'$. In short, the new voters help $d$ against its main threat, $a$, and hurt $b$ against its main threat, $c$, resulting in the change of the winning set from $\{b\}$ to $\{d\}$. It does not matter, in this case, that the new voters help $b$ against $d$, because $b$ and $d$ do not threaten to defeat each other in the presence of the cycles. 

In the example used in the proof of Proposition \ref{ResCoalPart}.\ref{ResCoalPart2}, there is a powerful symmetry argument: if $b$ is the unique winner for the margin graph on the left above, then $d$ must be the unique winner for the margin graph on the right above, assuming a neutrality property for margin graphs---that the names assigned to nodes do not matter---satisfied by Split Cycle (and the methods in Appendices~\ref{RankedPairsAppendix}-\ref{UncoveredAppendix}).

\begin{definition} A voting method $F$ satisfies \textit{margin graph neutrality} if for any profiles $\mathbf{P}$ and $\mathbf{P}'$, if there is a weighted directed graph isomorphism $h: \mathcal{M}(\mathbf{P})\to \mathcal{M}(\mathbf{P}')$, then $F(\mathbf{P}')=h[F(\mathbf{P})]$.
\end{definition}
\noindent For example, the map $c\mapsto a$, $a\mapsto c$, $b\mapsto d$, $d\mapsto b$ is a weighted directed graph isomorphism from the margin graph on the left in the proof of Proposition \ref{ResCoalPart}.\ref{ResCoalPart2} to the margin graph on the right (imagine turning the margin graph on the left $180^\circ$---then it matches the margin graph on the right except for the names of nodes). Despite the fact that the right margin graph is obtained from the left margin graph by adding two voters who rank $b$ over $d$, candidate $b$ on the left and candidate $d$ on the right are in isomorphic situations. Thus, if $b$ is the winner on the left, $d$ must be the winner on the right by margin graph neutrality.

Finally, note that the phenomenon with $b$ and $d$ above can happen only when $b$ and $d$ are in a cycle. Indeed, we have the following version of participation when the two relevant candidates are cycle free.

\begin{proposition}  For any profiles $\mathbf{P}$ and $\mathbf{P}'$ with $X(\mathbf{P})=X(\mathbf{P}')$ and $V(\mathbf{P})\cap V(\mathbf{P}')=\varnothing$ and any $x,y\in X$, if $x\in SC(\mathbf{P})$ and $x\mathbf{P}_i'y$ for all $i\in V(\mathbf{P}')$, and there is no cycle in $\mathcal{M}(\mathbf{P})$ or $\mathcal{M}(\mathbf{P}+\mathbf{P}')$ containing $x$ and $y$, then $y\not \in SC(\mathbf{P}+\mathbf{P}')$.
\end{proposition}
\begin{proof} Since $x\in SC(\mathbf{P})$, $y$ does not defeat $x$ in $\mathbf{P}$. Since there is no cycle in $\mathcal{M}(\mathbf{P})$ involving $x$ and $y$, it follows that $Margin_\mathbf{P}(x,y)\geq 0$. Hence $Margin_{\mathbf{P}+\mathbf{P}'}(x,y)> 0$, and by hypothesis, there is no cycle involving in $\mathcal{M}(\mathbf{P}+\mathbf{P}')$. Hence $x$ defeats $y$ in $\mathbf{P}+\mathbf{P}'$, so $y\not\in SC(\mathbf{P}+\mathbf{P}')$.\end{proof}

 \section{Other Methods}\label{OtherMethodsAppendix}

In this appendix, we give definitions of the voting methods in Figure \ref{AxiomTable} other than Split Cycle, as well citations or proofs for the claims about their properties in the figure. Typically it is assumed that the domain of these voting methods is the domain of linear profiles, though many properties continue to hold for the domain of strict weak order profiles. To avoid repetition, we note the following: anonymity and neutrality for each method are immediate from the definitions;  Pareto is obvious for all methods except GETCHA/GOCHA, for which we give an example violation; expansion consistency implies strong stability for winners, and stability for winners implies immunity to spoilers and immunity to stealers. Finally, additional examples of violations of positive involvement can be found in \citealt[Appendix A]{HP2021b}.

\subsection{Ranked Pairs (\citealt{Tideman1987})}\label{RankedPairsAppendix}
Let $\mathbf{P}$ be a profile and $T$ a linear order on the set $X(\mathbf{P})\times X(\mathbf{P})$ of pairs of candidates (the tiebreaking ordering). We say that a pair $(x,y)$ of candidates has a \textit{higher priority} than a pair $(x',y')$ of candidates using the tiebreaking ordering $T$ when either  $Margin_\mathbf{P}(x,y) > Margin_\mathbf{P}(x',y')$ or $Margin_\mathbf{P}(x,y) = Margin_\mathbf{P}(x',y')$ and $(x,y)\mathrel{T} (x',y')$.  Given a profile $\mathbf{P}$ and a tiebreaking ordering $T$ of $X(\mathbf{P})\times X(\mathbf{P})$, we construct a \textit{Ranked Pairs ranking} $\succ_{\mathbf{P},T}$ of $X(\mathbf{P})$ according to the following procedure: 
\begin{enumerate}
\item Initialize $\succ_{\mathbf{P},T}$ to $\varnothing$.
\item If all pairs $(x,y)$ with $x\neq y$ and $Margin_\mathbf{P}(x,y)\geq 0$ have been considered, then return $\succ_{\mathbf{P},T}$.  Otherwise let $(a,b)$ be the pair with the highest priority  among those with $a\neq b$ and $Margin_\mathbf{P}(a,b)\geq 0$ that have not been considered so far.  
\item If $\succ_{\mathbf{P},T}\cup\, \{(a,b)\}$ is acyclic, then add $(a,b)$ to $\succ_{\mathbf{P},T}$; otherwise, add $(b,a)$ to $\succ_{\mathbf{P},T}$.   Go to step 2. 
\end{enumerate}
When the procedure terminates, $\succ_{\mathbf{P},T}$ is a linear order.\footnote{This is a standard algorithm for Ranked Pairs, but if our goal is only to select winners rather than to produce a linear order of the whole set of candidates, then to save some steps we can run the procedure above for only those pairs $(x,y)$ with $Margin_\mathbf{P}(x,y)> 0$, in line with our description of Ranked Pairs in Section \ref{Comparison}. Then a winner according to $\succ_{\mathbf{P},T}$ is a \textit{maximal} element of $\succ_{\mathbf{P},T}$, i.e., an $x$ for which there is no $y\succ_{\mathbf{P},T} x$, and $x$ is a Ranked Pairs winner in $\mathbf{P}$ if $x$ is a winner according to $\succ_{\mathbf{P},T}$ for some tiebreaking ordering $T$.} A linear order $L$ on $X(\mathbf{P})$ is a \textit{Ranked Pairs ranking for $\mathbf{P}$} if $L=\,\succ_{\mathbf{P},T}$ for some tiebreaking ordering $T$ of $X(\mathbf{P})\times X(\mathbf{P})$.  Then the set $RP(\mathbf{P})$ of Ranked Pairs winners is the set of all $x\in X(\mathbf{P})$ such that $x$ is the maximum of some Ranked Pairs ranking for $\mathbf{P}$. 

See \citealt{Lamboray2008}, \citealt{BrillFischer2012}, and \citealt{Wangetal2019} for  discussion of the  axiomatic and computational properties of Ranked Pairs. 

\begin{paragraph}{Stability criteria} See Propositions \ref{RPspoiler} and \ref{RPStealer}.

 \end{paragraph}
\begin{paragraph}{Other criteria} Proofs that Ranked Pairs satisfies the Condorcet winner and loser criteria,  single-voter resolvability, and non-negative responsiveness can be found in \citealt{Tideman1987}.  For the satisfaction of reversal symmetry, Smith, and ISDA, see \citealt[Table 2]{Schulze2011}.  See Remark \ref{RPnote} on the status of independence of clones. For the satisfaction of rejectability, see Corollary \ref{BPreject}. Asymptotic resolvability follows from the fact that the proportion of profiles that are uniquely weighted goes to 1 as the number of voters goes to infinity, and Ranked Pairs selects a unique winner in any uniquely-weighted profile.  For the failure of positive involvement and negative involvement, see \citealt[p.~612]{Perez2001}.
 \end{paragraph}
  
\subsection{Beat Path  (\citealt{Schulze2011})}\label{BeatpathAppendix}

Let $\mathcal{M}$ be a margin graph.   A (\textit{simple}) \textit{path from $x$ to $y$ in $\mathcal{M}$}  is a sequence $\langle x_1, \ldots, x_k\rangle$ of distinct nodes in $\mathcal{M}$ where $x_1=x$, $x_k=y$, and for $i\in\{1,\ldots, k-1\}$, $x_i\overset{\alpha_i}{\to} x_{i+1}$.  The strength of a path $\langle x_1, \ldots, x_k\rangle$   in $\mathcal{M}$ is \[S_\mathcal{M}(\langle x_1, \ldots, x_k\rangle)=\min \{\alpha_i \mid x_i\overset{\alpha_i}{\to} x_{i+1},1\leq i\leq k-1\}.\] Given a profile $\mathbf{P}$, let  $Path_\mathbf{P}(x,y)$ be the set of all paths from $x$ to $y$ in $\mathcal{M}(\mathbf{P})$.    The strength of $x$ over $y$ in $\mathbf{P}$ is 
$$Strength_\mathbf{P}(x,y)=\begin{cases} \max\{S_{\mathcal{M}(\mathbf{P})}(p) \mid p\in Path_\mathbf{P}(x,y)\} & Path_\mathbf{P}(x,y)\ne\varnothing\\
0 & \mbox{otherwise.}
\end{cases}$$
 Then the set $BP(\mathbf{P})$ of Beat Path winners  is the set of all $x\in X(\mathbf{P})$ such that there is no  $y\in X(\mathbf{P})$ such that $Strength_\mathbf{P}(y,x)> Strength_\mathbf{P}(x,y)$.    

\begin{paragraph}{Stability criteria} See Propositions \ref{BPspoiler} and \ref{MinimaxBPstealers}.\end{paragraph}

\begin{paragraph}{Other criteria}For proofs that Beat Path satisfies reversal symmetry, the Condorcet winner and loser criteria, Smith, ISDA, independence of clones, single-voter and asymptotic resolvability, and non-negative responsiveness, see \citealt{Schulze2022}. For the satisfaction of rejectability, see Corollary \ref{BPreject}. For an example of a simultaneous failure of positive and negative involvement, where adding two voters with the ranking $aefcbd$ changes the unique Beat Path winner from $a$ to $d$, see Example 7 of \citealt{Schulze2022}.\end{paragraph}

\subsection{Minimax (\citealt{Simpson1969}, \citealt{Kramer1977})}\label{MinimaxSection} The set of winners for Minimax, also known as the Simpson-Kramer method, are the candidates whose largest majority loss is the smallest, i.e., for a profile $\mathbf{P}$, \[{Minimax}(\mathbf{P})= \mathrm{argmin}_{x\in X(\mathbf{P})} \mathrm{max}(\{Margin_\mathbf{P}(y,x) \mid y\in X(\mathbf{P})\}).\]

\begin{paragraph}{Stability criteria} For the satisfaction of immunity to spoilers, if $a\in Minimax(\mathbf{P}_{-b})$, $Margin_\mathbf{P}(a,b)>0$, and $b\not\in Minimax(\mathbf{P})$, then $a$ must still be among the candidates in $\mathbf{P}$ whose largest majority loss is smallest, so $a\in Minimax(\mathbf{P})$. For the violation of partial immunity to stealers, see Proposition \ref{MinimaxBPstealers}.\end{paragraph}
 
\begin{paragraph}{Other criteria} See \citealt{Felsenthal2012} for violations of reversal symmetry (under `preference inversion') and the Condorcet loser criterion (also shown in the proof of Proposition \ref{MinimaxBPstealers}), as well  as proofs of the Condorcet winner and non-negative reponsiveness criteria. Violation of independence of clones is discussed in \citealt{Tideman1987}. For violations of Smith and hence ISDA, see \citealt[p.~10]{Darlington2016}. For the satisfaction of single-voter resolvability, see \citealt{Tideman1987}, and for asymptotic resolvability, the argument is the same as given for Ranked Pairs in Appendix \ref{RankedPairsAppendix}. 

\begin{fact} Minimax satisfies rejectability.
\end{fact}
\begin{proof} Given $x\in Minimax(\mathbf{P})$, modify the margin graph $\mathcal{M}(\mathbf{P})$ to $\mathcal{M}'$ such that for all $y\in X(\mathbf{P})\setminus\{x\}$, (i) if there is no edge from $x$ to $y$ in $\mathcal{M}(\mathbf{P})$, add an edge from $x$ to $y$ in $\mathcal{M}'$, and (ii) increase the weights of all incoming edges to $y$ to be larger than the largest majority loss of $x$ in $\mathbf{P}$, such that all weights in $\mathcal{M}'$ have the same parity as the weights in $\mathcal{M}(\mathbf{P})$. Then $x$ is clearly the unique Minimax winner in $\mathcal{M}'$, and $\mathcal{M}'$ is the margin graph of a profile $\mathbf{P}'$ (which is linear if $\mathbf{P}$ is) by Debord's Theorem. Finally, since Minimax clearly satisfies the overwhelming majority criterion (recall Lemma \ref{OverwhelmingMaj}), $\mathbf{P}'$ may be used to obtain the $\mathbf{P}^+$ required for rejectability as in the proof of Proposition \ref{RejectProp}. \end{proof}
For the satisfaction of positive and negative involvement, see \citealt[p.~613]{Perez2001}.\end{paragraph}

\subsection{Copeland (\citealt{Copeland1951})}\label{CopelandAppendix}

The Copeland score of a candidate $x$ is the number of candidates to whom $x$ is majority preferred minus the number majority preferred to $x$.  The Copeland winners are the candidates with maximal Copeland score:
\[{Copeland}(\mathbf{P})= \mathrm{argmax}_{x\in X(\mathbf{P})}\, \vert \{y\in X(\mathbf{P})\mid Margin_\mathbf{P}(x,y)>0\}\vert -\vert \{y\in X(\mathbf{P})\mid Margin_\mathbf{P}(y,x)>0\}\vert .\]

\begin{paragraph}{Stability criteria} 
\begin{fact} Copeland satisfies immunity to spoilers.
\end{fact}
\begin{proof} If $a\in Copeland (\mathbf{P}_{-b})$, so $a$'s Copeland score is  maximal  in $\mathbf{P}_{-b}$, and $Margin_\mathbf{P}(a,b)>0$, then $a$'s Copeland score in $\mathbf{P}$ is maximal \textit{among the original candidates in $X(\mathbf{P}_{-b})$}; if in addition $b\not\in Copeland(\mathbf{P})$, then $a$'s Copeland score in $\mathbf{P}$ is maximal among all candidates in $X(\mathbf{P})$, so $a\in Copeland(\mathbf{P})$.
\end{proof} 
\noindent However, if we do not assume $b\not\in Copeland(\mathbf{P})$, then it is easy to construct a profile $\mathbf{P}$ in which $b$ has a higher Copeland score in $\mathbf{P}$ than  $a$ does (this requires $\vert X(\mathbf{P})\vert \geq 6$ if $M(\mathbf{P})$ is a tournament and $\vert X(\mathbf{P})\vert \geq 5$  otherwise), so that $a\not\in Copeland(\mathbf{P})$. Thus, Copeland violates immunity to stealers. 

\begin{fact} Copeland satisfies partial immunity to stealers.
\end{fact}

\begin{proof} Suppose $a$ is the unique Condorcetian candidate in $\mathbf{P}$, but  $b$ steals the election from $a$ in $\mathbf{P}$, so (i)~$a\in Copeland(\mathbf{P}_{-b})$, (ii) $a\to_\mathbf{P} b$, (iii) $a\not\in Copeland(\mathbf{P})$, and (iv) $b\in Copeland(\mathbf{P})$. Since $a$'s Copeland score is maximal in $\mathbf{P}_{-b}$ by (i) and increases by 1 from  $\mathbf{P}_{-b}$ to $\mathbf{P}$ by (ii), together (iii) and (iv) imply that $b$ has the maximum Copeland score in $\mathbf{P}$, so $Copeland(\mathbf{P})=\{b\}$ and there is some $c\in X(\mathbf{P})$ with $b\to_\mathbf{P}c$. Now since $b$ is not Condorcetian, for any $c\in X(\mathbf{P})$ such that $b\to_\mathbf{P}c$, we have $b\not\in Copeland(\mathbf{P}_{-c})$; then since $b$'s Copeland score decreases by only 1 from $\mathbf{P}$ to $\mathbf{P}_{-c}$, it follows from $b\not\in Copeland(\mathbf{P}_{-c})$ that there is a $d\in F(\mathbf{P}_{-c})$ whose Copeland score in $\mathbf{P}$ is one less than that of $b$ in $\mathbf{P}$. We claim that $d\to_\mathbf{P}b$. For if $d\not\to_\mathbf{P}b$, then since $d$'s Copeland score in $\mathbf{P}$, which is one less than that of $b$, is at least that of $a$ by (iii)-(iv), and $d$'s Copeland score does not decrease from $\mathbf{P}$ to $\mathbf{P}_{-b}$ given $d\not\to_\mathbf{P}b$, whereas $a$'s Copeland score does decrease from $\mathbf{P}$ to $\mathbf{P}_{-b}$ by (ii), it follows that $d$'s Copeland score is greater than that of $a$ in $\mathbf{P}_{-b}$, contradicting~(i). Thus, $d\to_\mathbf{P}b$. But then since $d$'s Copeland score is at least that of $a$ in $\mathbf{P}$, it follows by (i)-(ii) that $d\in Copeland(\mathbf{P}_{-b})$, which with $d\to_\mathbf{P}b$ implies that $d$ is Condorcetian, contradicting the assumption that $a$ is the unique Condorcetian candidate in $\mathbf{P}$.\end{proof}

\noindent However, Copeland does not satisfy stability for winners with tiebreaking, as shown by the following.

\begin{example} Consider a profile $\mathbf{P}$ with $X(\mathbf{P})=\{a,b,c,d\}$ whose majority graph $M(\mathbf{P})$ has $a\to b$, $b\to c$, and $b\to d$, but no other edges; then $a$ is the unique Condorcetian candidate, but $Copeland(\mathbf{P})=\{a,b\}$.\end{example}\end{paragraph}

\begin{paragraph}{Other criteria} It is easy to see that Copeland satisfies reversal symmetry, the Condorcet winner and loser criteria, Smith, ISDA, and non-negative responsiveness. For the failure of independence of clones, see \citealt{Tideman1987}. For the failure of rejectability and single-voter resolvability, see Proposition \ref{GETCHA_GOCHA_Rejectability}.  For the failure of asymptotic resolvability for $k\geq 3$, consider a majority graph with three candidates in a top cycle followed by a linear order of the remaining candidates, so the top three candidates are Copeland winners; the proportion of profiles realizing such a majority graph does not go to 0 as the number of voters goes to infinity (\citealt{HT2020}).  For the failure of positive and negative involvement, see \citealt[\S~4.1]{Perez2001}.
\end{paragraph}

\subsection{GETCHA (\citealt{Smith1973})}\label{GETCHAAppendix}
For the definition of GETCHA, see Definition \ref{GETCHA} in Section \ref{SmithSchwartz}.

\begin{paragraph}{Stability criteria}   It is easy to see that GETCHA satisfies expansion consistency.\end{paragraph}

\begin{paragraph}{Other criteria} To see that GETCHA fails Pareto, consider the following.

\begin{example} In the following profile, all voters prefer $a$ to $x$, but $x$ is among the GETCHA winners:
\begin{center}
\begin{minipage}{2in}\begin{tabular}{ccc}
$1$ & $1$ & $1$   \\\hline
$a$ & $b$ &  $c$ \\
$x$ &  $c$ & $a$ \\
$b$ &  $a$ &  $x$ \\
$c$ &  $x$ &  $b$ \\
\end{tabular}\end{minipage} \begin{minipage}{2in}\begin{tikzpicture}

\node[circle,draw, minimum width=0.25in] at (0,1.5) (a) {$a$}; 
\node[circle,draw,minimum width=0.25in] at (0,-1.5) (c) {$c$}; 
\node[circle,draw,minimum width=0.25in] at (2,1.5) (b) {$b$}; 
\node[circle,draw,minimum width=0.25in] at (2,-1.5) (x) {$x$}; 

\path[->,draw,thick] (b) to[pos=.7] node[fill=white] {$1$}  (c);
\path[->,draw,thick] (c) to node[fill=white] {$1$} (a);
\path[->,draw,thick] (a) to node[fill=white] {$1$} (b);

\path[->,draw,thick] (a) to[pos=.7] node[fill=white] {$3$}  (x);
\path[->,draw,thick] (x) to node[fill=white] {$1$} (b);
\path[->,draw,thick] (c) to node[fill=white] {$1$} (x);

 \end{tikzpicture}
\end{minipage}
\end{center}
\end{example}
It is easy to see that  GETCHA satisfies the Condorcet winner and loser criteria, as well as non-negative responsiveness. For reversal symmetry, note that if $\mathbf{P}$ is a profile with $GETCHA(\mathbf{P})=\{x\}$, then $x$ is a Condorcet winner.   Thus, $x$ is a Condorcet loser in $\mathbf{P}^r$, so $x\not\in GETCHA(\mathbf{P}^r)$.   It is also not difficult to see that GETCHA satisfies independence of clones, using an alternative characterization of GETCHA.   Given a profile $\mathbf{P}$, let  $a \leadsto _\mathbf{P} b$ mean that $Margin_\mathbf{P}(a,b)\geq 0$. Let $\leadsto _\mathbf{P}^*$ be the transitive closure of $\leadsto _\mathbf{P}$.  
\begin{lemma}[\citealt{Schwartz1986}, Corollary 6.2.2]\label{GETCHAlem} For any profile $\mathbf{P}$, \[GETCHA(\mathbf{P})=\{x\in X(\mathbf{P})\mid \mbox{for all }y\in X(\mathbf{P}): x\leadsto _\mathbf{P}^*y\}.\]
\end{lemma}

By definition, GETCHA satisfies the Smith criterion, and we proved ISDA in Footnote~\ref{GETCHAISDA}.  For the failure of rejectability and single-voter resolvability, see Proposition \ref{GETCHA_GOCHA_Rejectability}. The failure of asymptotic resolvability for $k\geq 3$ follows from the fact that GETCHA selects a unique winner only if there is a Condorcet winner, and for 3 or more candidates, the proportion of profiles with a Condorcet winner does not go to 1 as the number of voters goes to infinity (\citealt{DeMeyer1970}). To see that GETCHA fails positive and negative involvement, consider the following examples.

\begin{example} In a three-candidate, two-voter profile $\mathbf{P}$ with $a\mathbf{P}_ib\mathbf{P}_ic$ and $c\mathbf{P}_ja\mathbf{P}_jb$,  $GETCHA(\mathbf{P})=\{a,b,c\}$, yet adding one voter $k$ such that $b\mathbf{P}_k' a \mathbf{P}_k' c$ results in a profile $\mathbf{P}'$ in which $a$ is the Condorcet winner, so $GETCHA(\mathbf{P}')=\{a\}$.
\end{example}

\begin{example} Consider any three-candidate profile $\mathbf{P}$ with a Condorcet winner $a$, so $GETCHA(\mathbf{P})=\{a\}$, in which $Margin_\mathbf{P}(a,c)=1$; then adding a voter with the ranking $c\mathbf{P}_i' a \mathbf{P}_i'  b$ results in a profile $\mathbf{P}'$ in which $Margin_{\mathbf{P}'}(a,b)>0$ but $Margin_{\mathbf{P}'}(a,c)=0$, which implies $GETCHA(\mathbf{P}')=\{a,b,c\}$.\end{example}\end{paragraph}

\subsection{GOCHA (\citealt{Schwartz1986})}\label{GOCHAAppendix}

For the definition of GOCHA, see Definition \ref{GOCHA} in Section \ref{SmithSchwartz}. 

\begin{paragraph}{Stability criteria}  It is easy to see that GOCHA satisfies stability for winners using Lemma \ref{GOCHALem}.  The proof of Proposition \ref{SchwartzProp} shows that GOCHA does not satisfy strong stability for winners.\end{paragraph}

\begin{paragraph}{Other criteria}  For the violation of Pareto, the same example given for GETCHA in Section \ref{GETCHAAppendix} works for GOCHA. It is easy to see that GOCHA satisfies reversal symmetry,  the  Condorcet winner and loser criteria, and non-negative responsiveness (see \citealt{Felsenthal2012}). It is well known that GOCHA satisfies the Smith criterion, and ISDA can be proved using Lemma \ref{GOCHALem}.\footnote{Suppose $a\not\in GOCHA(\mathbf{P}_{-x})$, so there is a $b\in X(\mathbf{P}_{-x})$ such that $b\to_{\mathbf{P}_{-x}}^*a$ but $a\not\to_{\mathbf{P}_{-x}}^*b$. Then $b\to_{\mathbf{P}}^*a$. If $a\not\to_{\mathbf{P}}^*b$, then we are done: $a\not\in GOCHA(\mathbf{P})$. If $a\to_{\mathbf{P}}^*b$, then given $a\not\to_{\mathbf{P}_{-x}}^*b$, it follows that $x$ is on the path witnessing $a\to_{\mathbf{P}}^*b$, which together with $b\to_{\mathbf{P}}^*a$ implies that there is a path from $x$ to $a$ in $\mathcal{M}(\mathbf{P})$. Then since $x\not\in GETCHA(\mathbf{P})$ and there can be no path from a candidate outside $GETCHA(\mathbf{P})$ to one inside $GETCHA(\mathbf{P})$, it follows that $a\not\in GETCHA(\mathbf{P})$. Hence $a\not\in GOCHA(\mathbf{P})$, since GOCHA satisfies the Smith criterion. Conversely, suppose $a\not\in GOCHA(\mathbf{P})$, so there is a $b\in X(\mathbf{P})$ such that $b\to_\mathbf{P}^*a$ but $a\not\to_\mathbf{P}^*b$. If $x$ is not on the path witnessing $b\to_\mathbf{P}^*a$, then $b\to_{\mathbf{P}_{-x}}^*a$ but $a\not\to_{\mathbf{P}_{-x}}^*b$, so we are done: $a\not \in GOCHA(\mathbf{P}_{-x})$. If $x$ is on the path witnessing $b\to_\mathbf{P}^*a$, then since $x\not\in GETCHA(\mathbf{P})$ and there can be no path from a candidate outside $GETCHA(\mathbf{P})$ to one inside $GETCHA(\mathbf{P})$, it follows that $a\not\in GETCHA(\mathbf{P})$. Hence by ISDA for GETCHA, $a\not\in GETCHA(\mathbf{P}_{-x})$ and hence $a\not\in GOCHA(\mathbf{P}_{-x})$, since GOCHA satisfies the Smith criterion.}  For the satisfaction of independence of clones, see \citealt{Tideman1987}.    GOCHA fails rejectability, single-voter resolvability, and asymptotic resolvability for $k\geq 3$ by the same reasoning as for GETCHA, using the fact that in the limit as the number of voters goes infinity, GOCHA is equivalent to GETCHA. For the failure of positive involvement and negative involvement, see \citealt[\S~4.1]{Perez2001}, where GOCHA is called ``Top Cycle'', or \citealt{Felsenthal2016}, where GOCHA is called ``Schwartz''. The following is a simple example of the failure of positive involvement.

\begin{example}Consider a three-cycle with $y\overset{5}{\to} x\overset{3}{\to} z\overset{1}{\to} y$, so $GOCHA(\mathbf{P})=\{x,y,z\}$, and add a new voter with $x\mathbf{P}_i'y\mathbf{P}_i'z$ to obtain $y\overset{4}{\to} x\overset{4}{\to} z $ with $y$ and $z$ tied, so $GOCHA(\mathbf{P}')=\{y\}$.\end{example}
\end{paragraph}

\subsection{Uncovered Set (\citealt{Fishburn1977}, \citealt{Miller1980})}\label{UncoveredAppendix}

The Uncovered Set in voting is usually attributed to Fishburn \citeyearpar{Fishburn1977} and Miller \citeyearpar{Miller1980}, though the covering relation appears in earlier game-theoretic work of Gillies \citeyearpar{Gillies1959}. Fishburn defined his version of the Uncovered Set for arbitrary margin graphs, whereas Miller defined his only for \textit{tournaments}, i.e., directed graphs in which the edge relation $\to$ is not only asymmetric but also \textit{weakly complete}: for all distinct nodes $x,y$, either $x\to y$ or $y\to x$. Fishburn and Miller's definitions are equivalent for tournaments but not for margin graphs that are not weakly complete, which may arise from profiles with an even number of voters or non-linear ballots. Several non-equivalent definitions of the Uncovered Set for arbitrary margin graphs appear in the literature (see \citealt{Bordes1983,Peris1999,Penn2006,Duggan2013}), and some of these versions differ in their axiomatic properties. As examples, we will consider the versions due to Fishburn and Gillies.

Given a margin graph $\mathcal{M}$ and nodes $x,y$ in $\mathcal{M}$, say that $y$ \textit{left-covers $x$ in $\mathcal{M}$} if for all nodes $z$ in $\mathcal{M}$, if $z\to y$, then $z\to x$.\footnote{Miller's \citeyearpar{Miller1980} definition uses the right-sided version: $y$ \textit{right-covers} $x$ in $\mathcal{M}$ if for all $z$, if $x\to z$, then $y\to z$. If $\to $ is weakly complete, then left-covering and right-covering are equivalent.} Then the Fishburn and Gillies versions of the Uncovered Set are defined by:
\begin{eqnarray*}
UC_{Fish}(\mathbf{P})&=&\{x\in X(\mathbf{P})\mid \mbox{there is no $y\in X(\mathbf{P})$: $y$ left-covers $x$ but $x$ does not left-cover $y$ in $\mathcal{M}(\mathbf{P})$}\};\\
UC_{Gill}(\mathbf{P})&=&\{x\in X(\mathbf{P})\mid \mbox{there is no $y\in X(\mathbf{P})$: $y\to x$ and $y$ left-covers $x$ in $\mathcal{M}(\mathbf{P})$}\}.
\end{eqnarray*}
Note that $UC_{Fish}(\mathbf{P})\subseteq UC_{Gill}(\mathbf{P})$. A useful alternative characterization of $UC_{Gill}$ is given by the following ``two-step'' principle (see, e.g., \citealt[Proposition 12(ii)]{Duggan2013}): $x\in UC_{Gill}(\mathbf{P})$ if and only if for all $y\in X(\mathbf{P})\setminus\{x\}$, $Margin_\mathbf{P}(x,y)\geq 0$ or there is a $z\in X(\mathbf{P})$ such that $Margin_\mathbf{P}(x,z)\geq 0$  and $Margin_\mathbf{P}(z,y)> 0$.

\begin{paragraph}{Stability criteria} The method $UC_{Fish}$ satisfies stability for winners and hence immunity to spoilers, since if $Margin_\mathbf{P}(a,b)>0$, then $b$ does not left-cover $a$. However, it violates strong stability for winners. 

\begin{example}\label{UncoveredEx} Consider the profile $\mathbf{P}$ below where $UC_{Fish}(\mathbf{P}_{-b})=\{a,c,d\}$, $Margin_{\mathbf{P}}(a,b)\geq 0$, but $UC_{Fish}(\mathbf{P})=\{b\}$:

\begin{center}
\begin{minipage}{2in}\begin{tabular}{cccccc}
$1$ & $1$ & $1$ & $1$ & $1$ & $1$\\\hline 
$a$ & $d$ & $c$ & $b$ & $b$ & $b$\\ 
$d$ & $c$ & $a$ & $a$ & $d$ & $c$\\ 
$c$ & $a$ & $d$ & $d$ & $c$ & $a$\\ 
$b$ & $b$ & $b$ & $c$ & $a$ & $d$
\end{tabular}\end{minipage} \begin{minipage}{2in}
\begin{tikzpicture}

\node[circle,draw,minimum width=0.25in] at (3,0) (b) {$b$}; 

\node[circle,draw, minimum width=0.25in] at (5,0) (e) {$a$}; 
\node[circle,draw,minimum width=0.25in] at (8,0) (d) {$d$}; 
\node[circle,draw,minimum width=0.25in] at (6.5,1.5) (f) {$c$};  

\path[->,draw,thick] (e) to node[fill=white] {$2$} (d);
\path[->,draw,thick] (f) to node[fill=white] {$2$} (e);
\path[->,draw,thick] (d) to node[fill=white] {$2$} (f);

  \end{tikzpicture}
  \end{minipage}
  \end{center}
\end{example}

\noindent  Hence $UC_{Fish}$ also violates expansion consistency. By contrast, $UC_{Gill}$ satisfies expansion consistency, as one can easily see using the two-step characterization above.\end{paragraph}

\begin{paragraph}{Other criteria} For a proof that the Uncovered Set satisfies Pareto, see \citealt[Proposition 52]{Duggan2013}. For reversal symmetry, the argument is the same as we gave for GETCHA in Appendix \ref{GETCHAAppendix}. For the satisfaction of the Condorcet criterion under various definitions of the Uncovered Set, see \citealt[Propositions 4, 5, and 13]{Duggan2013}. The Condorcet loser and non-negative responsiveness criteria are also straightforward to check. For the satisfaction of the Smith criterion under all standard definitions of the Uncovered Set, see \citealt[Propositions 4, 5, and 14]{Duggan2013}. 

\begin{fact} $UC_{Fish}$ and $UC_{Gill}$ satisfy ISDA.
\end{fact}

\begin{proof} Suppose $x\in X(\mathbf{P})\setminus GETCHA(\mathbf{P})$ and $UC\in \{UC_{Fish},UC_{Gill}\}$. To see that $UC(\mathbf{P}_{-x})\subseteq UC(\mathbf{P})$, suppose $a\in UC(\mathbf{P}_{-x})$. By the Smith criterion for $UC$ and ISDA for GETCHA, $UC(\mathbf{P}_{-x})\subseteq GETCHA(\mathbf{P}_{-x})=GETCHA(\mathbf{P})$, so together $a\in UC(\mathbf{P}_{-x})$ and $x\in X(\mathbf{P})\setminus GETCHA(\mathbf{P})$ imply $Margin_\mathbf{P}(a,x)>0$, so $a\in UC(\mathbf{P})$ by stability for winners. Now we claim that $UC(\mathbf{P})\subseteq UC(\mathbf{P}_{-x})$. By definition,  $a\in UC_{Gill}(\mathbf{P})$ (resp.~$a\in UC_{Fish}(\mathbf{P})$) if and only if for all $b\in X(\mathbf{P})\setminus\{a\}$, we have $a\leadsto b$ (resp.~$a$ left-covers $b$) in $\mathbf{P}$ or there is a $c\in X(\mathbf{P})$ such that $a\leadsto  c\to b$. Now suppose $a\in UC(\mathbf{P})$. To show $a\in UC_{Gill}(\mathbf{P}_{-x})$ (resp.~$a\in UC_{Fish}(\mathbf{P}_{-x})$), we must show that for any $b\in X(\mathbf{P}_{-x})\setminus\{a\}$, we have $a\leadsto b$ (resp.~$a$ left-covers~$b$) in $\mathbf{P}_{-x}$ or there is a $c\in X(\mathbf{P}_{-x})$ such that $a\leadsto c\to b$. Since $a\in UC_{Gill}(\mathbf{P})$ (resp.~$a\in UC_{Fish}(\mathbf{P})$), we have $a\leadsto b$ (resp.~$a$ left-covers $b$) in $\mathbf{P}$ or there is a $c\in X(\mathbf{P})$ such that $a\leadsto c\to b$. If $a\leadsto b$ (resp.~$a$ left-covers $b$) in $\mathbf{P}$, then this holds in $\mathbf{P}_{-x}$, so we are done. So suppose it is not the case that $a\leadsto b$ (resp.~$a$ left-covers $b$) in $\mathbf{P}$, but instead there is a $c\in X(\mathbf{P})$ such that $a\leadsto c\to b$. For $UC_{Gill}$, since it is not the case that $a\leadsto b$, we have $b\to a$.  Then since $a\in UC_{Gill}(\mathbf{P})\subseteq GETCHA(\mathbf{P})$, it follows from $c\to b\to a$ that $c\in GETCHA(\mathbf{P})$, so $c\neq x$.  For $UC_{Fish}$, since $a\in UC_{Fish}(\mathbf{P})\subseteq GETCHA(\mathbf{P})$, $a$ left-covers every candidate  in $\mathbf{P}$ outside $GETCHA(\mathbf{P})$. Hence from our assumption that $a$ does not left-cover $b$ in $\mathbf{P}$, it follows that $b\in GETCHA(\mathbf{P})$. Then since $c\to b$, we have $c\in GETCHA(\mathbf{P})$, so $c\neq x$. Thus, in either case, $c\in X(\mathbf{P}_{-x})$ and $a\leadsto c\to b$, so we are~done.\end{proof}

That $UC_{Gill}$ satisfies independence of clones can be seen using the two-step characterization above. However, $UC_{Fish}$ does not, as shown by the following.

\begin{example} In the profile $\mathbf{P}$ from Example \ref{UncoveredEx},  observe that $\{a,c,d\}$ is a set of clones in $\mathbf{P}$, $UC_{Fish}(\mathbf{P})=\{b\}$, but $UC_{Fish}(\mathbf{P}_{-c})=\{a,b\}$, so $UC_{Fish}$ does not satisfy the condition that clone choice is independent of clones (recall Definition \ref{IndependenceOfClones}). 
\end{example}

Uncovered Set violates rejectability, single-voter resolvability, and asymptotic resolvability for $k\geq 3$ by the same reasoning as for GETCHA (i.e., Uncovered Set selects a unique winner only if there is a Condorcet winner). For the failure of positive and negative involvement under all standard definitions of the Uncovered Set, see \citealt[\S~4.1]{Perez2001}.\end{paragraph}

\subsection{Instant Runoff}\label{RankedChoiceAppendix}

Instant Runoff is usually defined only for profiles $\mathbf{P}$ in which each ballot $\mathbf{P}_i$ is a linear order of some subset of $X(\mathbf{P})$. Instant Runoff iteratively removes the candidate with the least number of first-place votes, until there is a  candidate  with a majority of the first-place votes. The question then arises of what to do if there are two or more candidates with the least number of first-place votes. One version of Instant Runoff (\citealt[p.~7]{Taylor2008}) removes all such candidates;  and if, at some stage of the removal process, all remaining candidates have the same number of first-place votes (so all candidates would be removed), then all remaining candidates  are selected as winners. Alternatively, the ``parallel universe'' version of Instant Runoff (cf.~\citealt[\S~3]{Freeman2015}) says that $a$ wins in $\mathbf{P}$ if there is a candidate $b$ with the least number of first-place votes in $\mathbf{P}$ such that $a$ wins according to the parallel universe version of Instant Runoff in $\mathbf{P}_{-b}$. See \citealt{Freemanetal2014} and \citealt{Wangetal2019}  for axiomatic and computational properties of Instant Runoff.  
 
\begin{paragraph}{Stability criteria}  Example \ref{IRVExample} shows that Instant Runoff violates (even partial) immunity to spoilers. If we consider removing the Progressive from the election instead of the Republican, then it also shows that Instant Runoff violates  (even partial) immunity to stealers. \end{paragraph}

\begin{paragraph}{Other criteria} For the failure of reversal symmetry, see \citealt{Felsenthal2012} (under ``preference inversion'' for ``Alternative vote''). It is well known that Instant Runoff violates the Condorcet winner  and non-negative responsiveness criteria but satisfies the Condorcet loser criterion (\citealt{Felsenthal2012}). The parallel universe version of Instant Runoff satisfies independence of clones (\citealt{Tideman1987}), while the simultaneous removal version does not: if there are 4 voters with the ranking $abc$, 3 with $bca$, and 3 with $cba$, then the clones $b$ and $c$ are both eliminated in the first round, so $a$ wins, whereas in the election with just $a$ and $b$, $b$ wins.\footnote{Thanks to Dominik Peters for this example.} The failure of  Smith and ISDA  follows from the failure of the Condorcet winner criterion.  For the satisfaction of single-voter resolvability, see \citealt{Tideman1987} (where Instant Runoff is called ``Alternative vote''). That Instant Runoff satisfies rejectability follows from Lemma \ref{ResolveReject} given that it satisfies singer-voter resolvability and clearly homogeneity. Asymptotic resolvability follows from the fact that for any number of candidates, the proportion of profiles in which there is a tie in the number of first place votes for two candidates goes to 0.  For positive involvement, suppose $x\in IRV(\mathbf{P})$, so $x$ is not eliminated at any stage of the iteration procedure starting from $\mathbf{P}$. Then where $\mathbf{P}'$ is a one-voter profile whose voter ranks $x$ in first place, clearly $x$ is not eliminated at any stage of the iteration procedure starting from $\mathbf{P}+\mathbf{P}'$, so $x\in IRV(\mathbf{P}+\mathbf{P}')$. For the failure of negative involvement, see \citealt{Fishburn1983} (under the ``no show paradox'').\end{paragraph}

\subsection{Plurality}

The Plurality score of a candidate is the number of voters who rank that candidate uniquely in first place. The Plurality voting method selects as winners all candidates whose Plurality scores are maximal. The problems with Plurality  are well known (see \citealt{Laslier2012}). 

\begin{paragraph}{Stability criteria}  Example \ref{BushNaderGore} shows that Plurality violates immunity to spoilers. If we consider removing Bush from the election instead of Nader, then it also shows that Plurality violates immunity to stealers.\end{paragraph}

\begin{paragraph}{Other criteria} For the failure of reversal symmetry, see \citealt{Felsenthal2012}.   It is well known that Plurality violates the Condorcet winner and  Condorcet loser criteria (again see \citealt{Felsenthal2012}).   It is clear that Plurality satisfies  non-negative responsiveness.   The failure of the Smith and ISDA criteria follows from the failure of the Condorcet winner criterion. Example \ref{BushNaderGore} shows that Plurality does not satisfy independence of clones. The satisfaction of single-voter and asymptotic resolvability and positive and negative involvement is obvious.  That Plurality satisfies rejectability follows from Lemma \ref{ResolveReject} given that it satisfies single-voter resolvability and clearly homogeneity. \end{paragraph}

\section{Frequency of Irresoluteness}\label{Quant}

\renewcommand{\thefigure}{D.\arabic{figure}}
\setcounter{figure}{0} 

The graphs in this appendix show the frequency with which Split Cycle and several other voting methods select more than one winner, as well as the sizes of the winning sets conditional on there being more than one winner, before tiebreaking. Thus, they show how often a tiebreaking procedure must be applied. 

We evaluated the voting methods on profiles with different numbers of candidates and voters, using different probability models to generate profiles. Probability models for generating linear profiles are more common, so in this section all profiles are linear. Below we explain the probability models, the different types of graphs, and some conclusions drawn from the data.

\begin{paragraph}{Probability models}We considered several probability models for generating profiles with $n$ candidates and $m$ voters. According to the impartial culture (IC) model, each such profile is equally likely. Equivalently, each voter chooses a linear order of the $n$ candidates at random, and the voters' choices are independent. 

In the P\'{o}lya-Eggenberger urn model (\citealt{Berg1985}), to generate a profile given a parameter $\alpha\in [0,\infty)$, each voter in turn randomly draws a linear order from an urn. Initially the urn is the set of all linear orders of the $n$ candidates. If a voter randomly chooses $L$ from the urn, we return $L$ to the urn plus $\alpha n!$ copies of $L$. IC is the special case where $\alpha=0$. The Impartial Anonymous Culture (IAC) is the special case where $\alpha=1/n!$. Following \citealt{Boehmer2021}, for each generated profile, we chose $\alpha$ according to a Gamma distribution with shape parameter $k=0.8$ and scale parameter $\theta=1$ for the model we call ``urn.''

In the Mallow's model (see \citealt{Mallow1957,Marden1995}), to generate a profile, the main idea is to fix a reference  linear ordering of the candidates and to assign to each voter a ranking that is ``close'' to this reference ranking.   Closeness to the reference ranking is defined using the Kendall-tau distance between rankings, depending on a {\em dispersion} parameter $\phi$.   Setting $\phi= 0$ means that every voter is assigned the reference ranking, and setting $\phi=1$ is equivalent to the IC model. Formally, to generate a profile given a reference ranking $L_0$ of the set $X$ of candidates and $\phi\in (0,1]$, the probability that a voter's ballot is the linear order $L$ of $X$ is $Pr_{L_0,\phi}(L)=\phi^{\tau(L,L_0)}/C$ where $\tau(L,L_0)= {{\vert X\vert }\choose{2}} - \vert L\cap L_0\vert $ is the Kendell-tau distance of $L$ to $L_0$, and $C$ is a normalization constant. For each profile, we chose $\phi$ by first randomly selecting what Boehmer et al.~\citeyearpar{Boehmer2021} call a rel-$\phi$ value, which together with the number $m$ of candidates determines $\phi$. See \citealt{Boehmer2021} for details on this parameterization of the Mallow's model in terms of rel-$\phi$ values. 

In addition to generating profiles using a single reference ranking $L_0$, we considered generating profiles using two reference rankings, which are the reverse of each other. E.g., $L_0$ ranks candidates from more liberal to more conservative, while $L_0^{-1}$ ranks candidates in the opposite order. The set of voters is divided into two groups, each associated with one of the reference rankings. Each voter is equally likely to be assigned to either of the two groups. Formally, the probability that a voter's ballot is $L$ is $\frac{1}{2} Pr_{L_0,\phi}(L)+\frac{1}{2}Pr_{L_0^{-1},\phi}(L)$.
\end{paragraph}

\begin{paragraph}{Types of graphs} We include three types of graphs with data from simulated elections. For all three types, for each data point in a graph, we sampled 25,000 profiles with an even number $n$ of voters and 25,000 profiles with the next odd number $n+1$ of voters, displayed in the graph below the even number, in order to have a mix of even and odd-sized electorates.

The first type of graph concerns the  \textit{frequency of irresoluteness}. The graphs on the left of Figures \ref{icgraphs}, \ref{urngraphs}, \ref{MallowsSingleRef}, and \ref{MallowsDoubleRef} show the frequency of multiple winners for several voting methods as the number of candidates ranges from 5 to 30 and the number of voters range from 4 to 5,001. On the right of the figures, we use the boxen plot or letter-valued plot (\citealt{Hofmann2017}) representation of the quantiles of the sizes of the winning sets \textit{for profiles with multiple winners}. The black dots outside of the boxes are the ``outliers.''

The second type of graph concerns the \textit{frequency of different winners}. This can be understood in two ways. First, how often do two methods produce different sets of tied winners? Second, if we assume that given a set of tied winners, the ultimate winner will be chosen randomly according to a uniform probability distribution, how often do two methods not only produce different sets of tied winners but  also produce different ultimate winners after random tiebreaking? Figure~\ref{PercentDiffBPvsSC} shows the frequency of different winners for Split Cycle vs.~Beat Path in these two senses according to several probability models. This is related to irresoluteness because in either of the two senses of having different winners, Split Cycle can differ from Beat Path only when Split Cycle outputs multiple winners before tiebreaking.\end{paragraph}

\begin{paragraph}{Choice of methods} We selected the voting methods with which to compare Split Cycle as follows. In light of Section \ref{Comparison}, a key comparison is to the most well-known refinements of Split Cycle, namely Beat Path and Ranked Pairs. However, we did not include Ranked Pairs due to the computational difficulty of determining the Ranked Pairs winners for elections with small numbers of voters (see \citealt{BrillFischer2012}); but for those combinations of voters and candidates for which we were able to compute Ranked Pairs for all sampled profiles, Ranked Pairs was similar to Beat Path in irresoluteness. We chose to include Copeland because like Split Cycle, GETCHA, and GOCHA, Copeland does not satisfy the resolvability criteria of Definitions \ref{TidemanResolve} and \ref{AsymptoticResolve}, yet Copeland is one of the most discriminating of all C1 voting methods (recall Section \ref{Rejectability} for the definition of C1, and see \citealt{Brandt2014} on the discriminating power of different C1 methods).  As for our choice to include Uncovered Set,\footnote{We used the Gillies version of the Uncovered Set for profiles with an even number of voters. For an odd number of voters with linear ballots, the different versions of Uncovered Set are equivalent.} it follows from a result of Moulin \citeyearpar[Theorem 1]{Moulin1986} that for any C1 voting method $F$ satisfying neutrality and expansion consistency and linear profile $\mathbf{P}$ with an odd number of voters, $UC(\mathbf{P})\subseteq F(\mathbf{P})$ (for an analogous result for an even number of voters, using a definition of the Uncovered Set that satisfies expansion consistency, see \citealt[Theorem 1]{Peris1999}). Thus, by comparing Split Cycle to the Uncovered Set, we are comparing Split Cycle to the most discriminating of all C1 methods satisfying expansion consistency, and by comparing Split Cycle to Copeland, we are comparing Split Cycle to one of the most discriminating of all C1 methods.\end{paragraph}

\begin{paragraph}{Discussion} We highlight the following takeaway points about the results of our simulations:

\begin{itemize}
\item The IC model can be viewed as a ``worst case scenario'' for irresoluteness (cf.~\citealt{Tsetlin2003}), so we expect that in practice the frequency of multiple winners will be substantially lower. We also ran our simulations for the IAC model but the graphs were almost indistinguishable from those of IC, so we omit them here. By contrast, we also tried one, two, and three-dimensional spatial models and a single-peaked model from \citealt{Boehmer2021}, and in these models, all of the methods were essentially resolute except with small numbers of voters, so we omit these graphs as well.
\item The boxen plots show that when there are multiple winners, generally there are very few for Split Cycle, Beat Path, and Copeland, whereas there are significantly more for Uncovered Set and many for GETCHA. This pattern holds across the different probability models for profiles.
\item Unlike for the other voting methods, for Copeland the proportion of profiles with multiple winners actually increases slightly as we increase the number of voters from 4 to 5,001, except under the Mallow's model with one reference ranking in which all methods tend toward resoluteness.
\item Unlike for the other voting methods, for Copeland the proportion of profiles with multiple winners is largely insensitive to the number of candidates---in fact, it decreases slightly from 10 to 30 candidates under the IC model.
\item Graphs for the urn model appear roughly as compressed versions of the graphs for the IC model, but there are some qualitative differences. For example, Split Cycle compares more favorably with Copeland in terms of irresoluteness according to the urn model than according to IC. For example, for 20 candidates and 5,000/5,001 voters, Copeland is more resolute than Split Cycle according to the IC model but has about the same frequency of irresoluteness as Split Cycle according to the urn model.
\item For the Mallow's model, the difference between one and two reference rankings is striking. With only one reference ranking, representing a society in which voters gravitate to different degrees toward a single ranking of the candidates, all of the methods are nearly resolute by 5,000/5,001 voters. By contrast, with two, reversed reference rankings, representing a society in which voters are divided into two groups gravitating toward reversed rankings, the results are not far from those of the IC model.
\end{itemize}
\end{paragraph}

\begin{figure} 
\begin{center}
\includegraphics[scale=0.5]{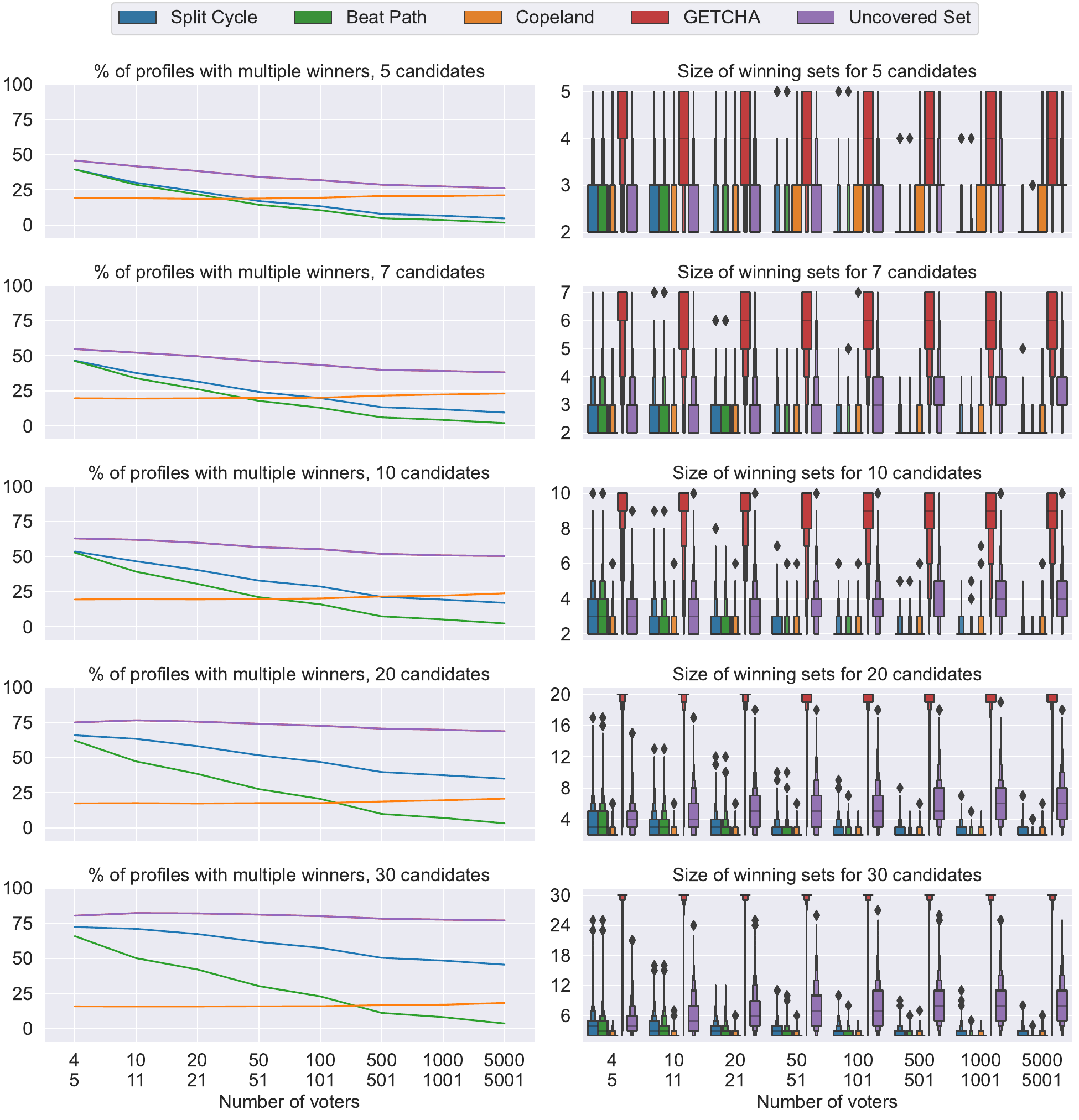}
\end{center}
\caption{The profiles were generated using the IC model. Results for the IAC model are almost the same. On the left, the purple line for the Uncovered Set is on top of the red line for GETCHA.}
 \label{icgraphs}
\end{figure}

\begin{figure} 
\begin{center}
\includegraphics[scale=0.5]{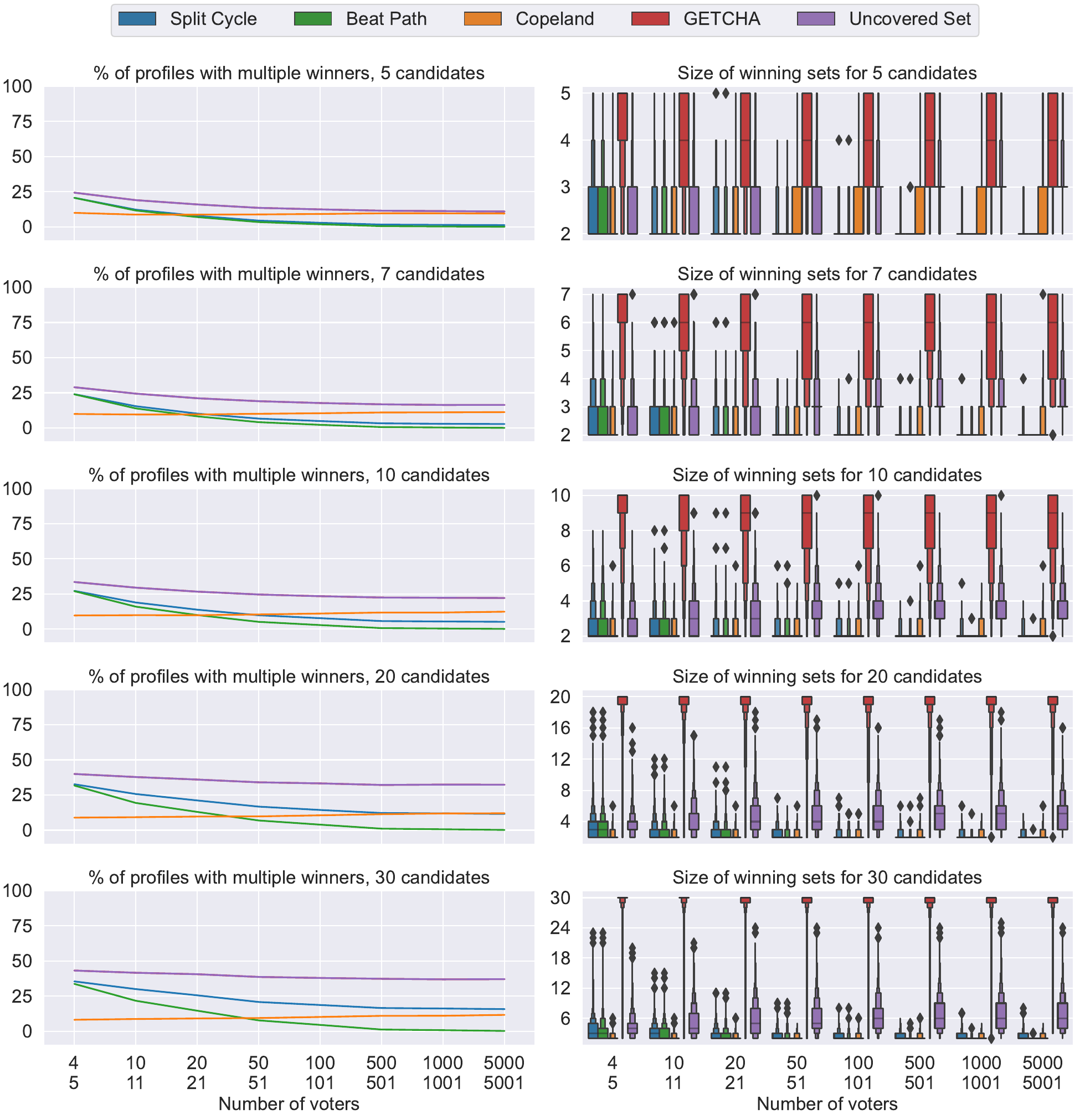}
\end{center}
\caption{The profiles were generated using the urn model with $\alpha$ chosen according to a Gamma distribution with shape parameter $k=0.8$ and scale parameter $\theta=1$ as in \citealt{Boehmer2021}. On the left, the purple line for the Uncovered Set is on top of the red line for GETCHA.}
 \label{urngraphs}
\end{figure}
 
\begin{figure} 
\begin{center}
\includegraphics[scale=0.5]{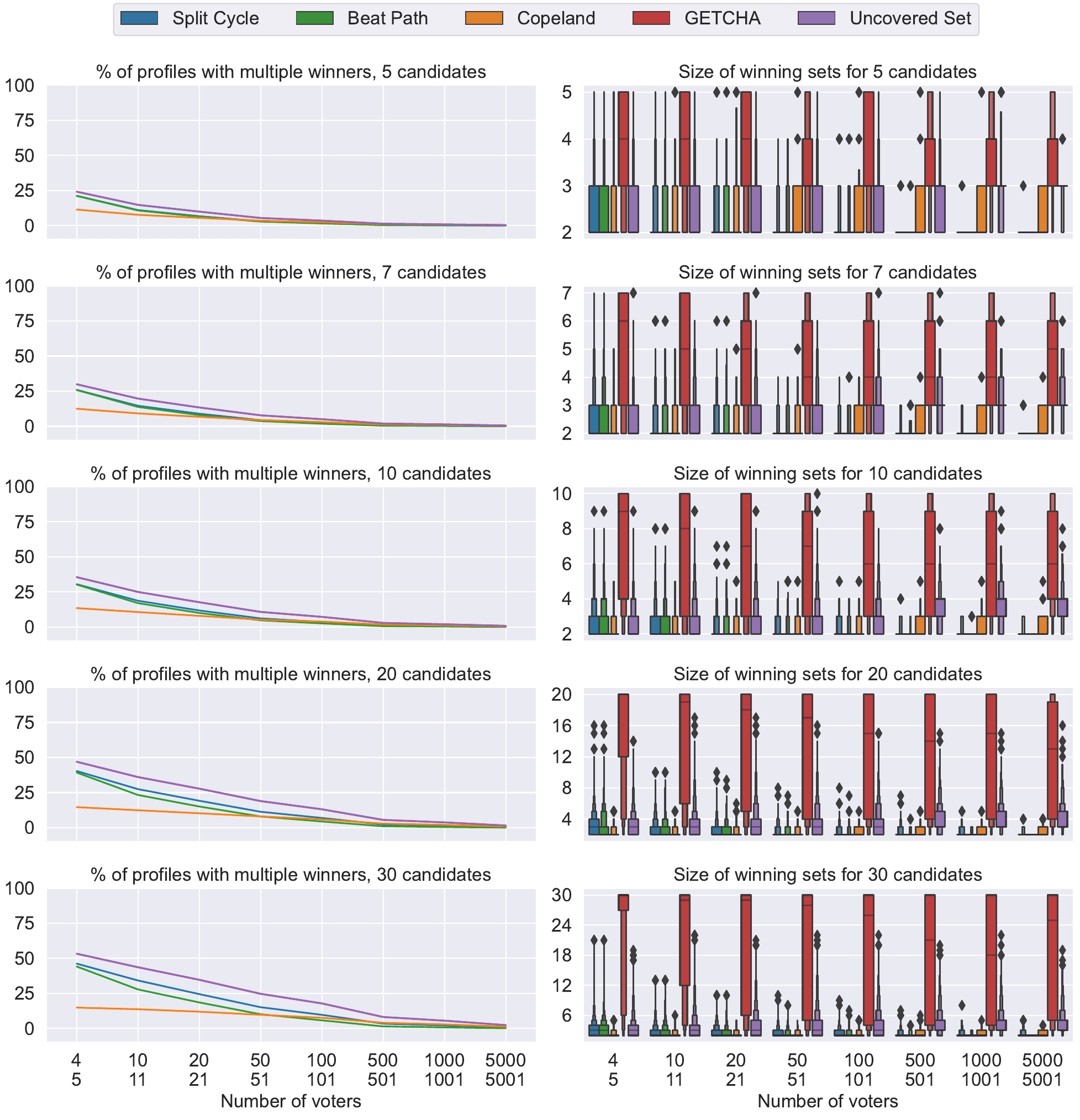}
\end{center}
\caption{The profiles were generated using the Mallows model with dispersion parameter $\phi$ chosen as described in the main text. On the left, the purple line for the Uncovered Set is on top of the red line for GETCHA.}
 \label{MallowsSingleRef}
\end{figure}

\begin{figure} 
\begin{center}
\includegraphics[scale=0.5]{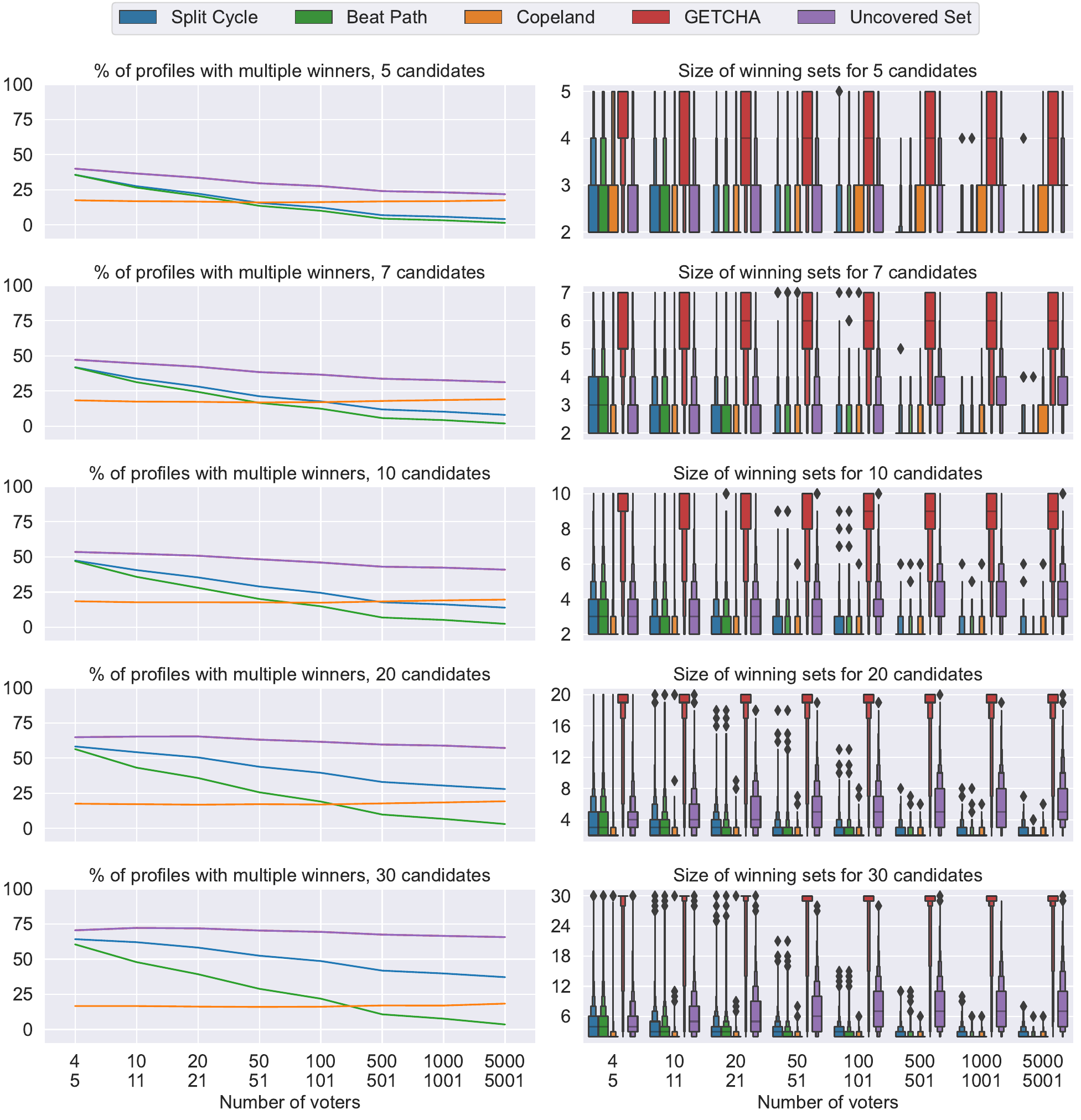}
\end{center}
\caption{The profiles were  generated using the Mallows model with {\em two} reference rankings, which are the reverse of each other. On the left, the purple line for the Uncovered Set is on top of the red line for GETCHA.}
 \label{MallowsDoubleRef}
\end{figure}

\begin{figure} 
\begin{center}
\includegraphics[scale=0.5]{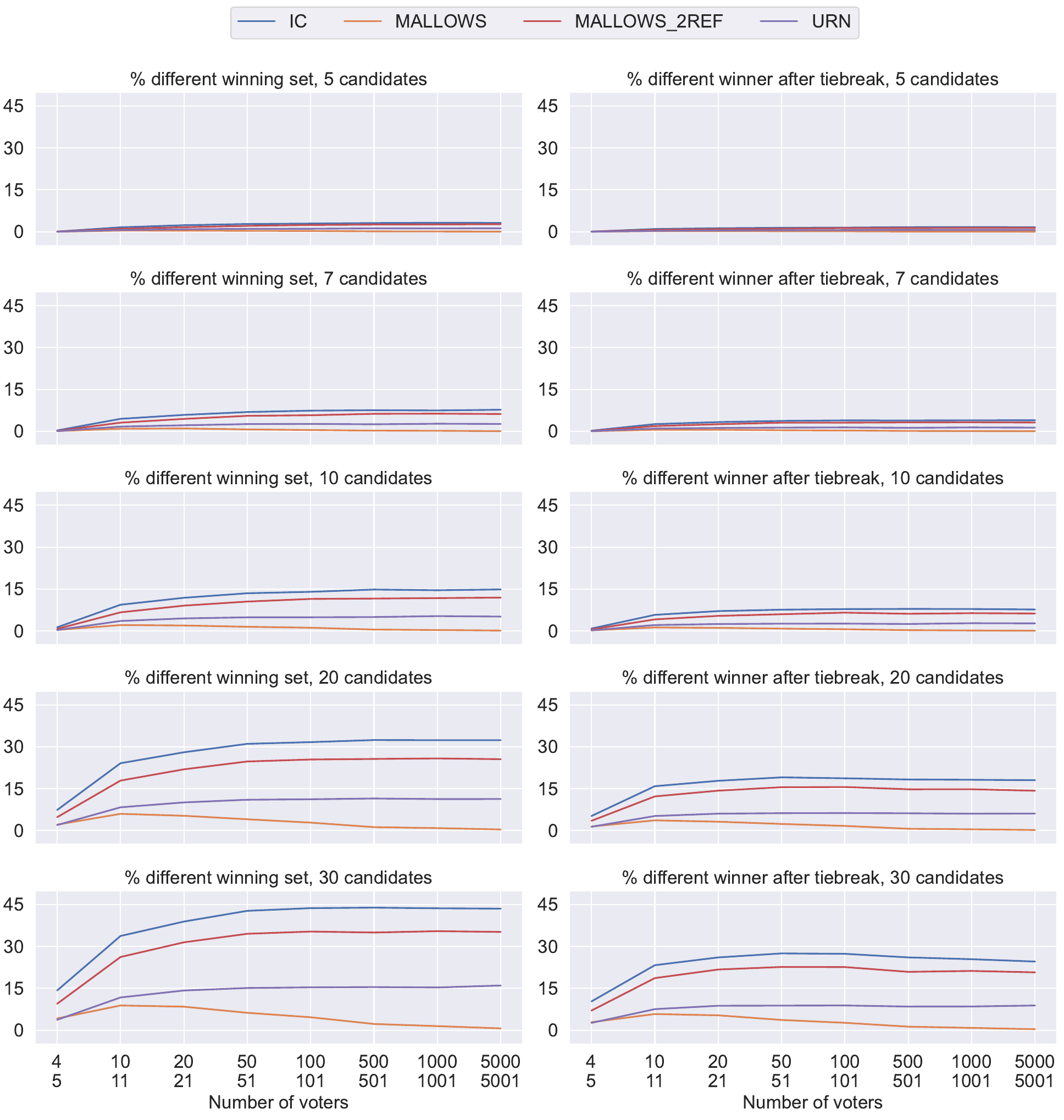}
\end{center}
\caption{The graphs in the left column show the percentage of profiles in which Split Cycle and Beat Path output different sets, sampling profiles according to five probability models. The graphs in the right column show the percentage of profiles such that (i) Split Cycle and Beat Path output different sets of winners and (ii) randomly selecting a Split Cycle winner and randomly selecting a Beat Path winner resulted in different ultimate winners.}
 \label{PercentDiffBPvsSC}
\end{figure}

\newpage

\onehalfspace

\bibliographystyle{plainnat}
\bibliography{splitcycle}

\end{document}